\documentclass{theoretics}

\usepackage{mathtools}
\usepackage{MnSymbol}
\usepackage{ifthen}
\usepackage{calc}
\usepackage{stmaryrd}
\usepackage{alphalph}
\usepackage{BOONDOX-cal}

\allowdisplaybreaks

\newcommand{\uend}{}
\newcommand{\uende}{\eqno}  

\newcommand{\case}[1]{\par\medskip\noindent\textit{Case #1: }}

\newenvironment{cs}{
  \begin{description}
    \renewcommand{\case}[1]{\item[\normalfont\itshape\mdseries Case ##1:]}
  }{
  \end{description}
}

\newlist{caselist}{description}{10}
\setlist[caselist]{font=\itshape\mdseries}

 \newlist{eroman}{enumerate}{2}
 \setlist[eroman,1]{label=(\roman*)}
 \setlist[eroman,2]{label=(\alph*)}

\definecolor{blau}{RGB}{0,84,159}
\definecolor{hellblau}{RGB}{142,168,229}
\definecolor{petrol}{RGB}{0,97,101}
\definecolor{tuerkis}{RGB}{0,152,161}
\definecolor{gruen}{RGB}{87,171,39}
\definecolor{maigruen}{RGB}{189,205,0}
\definecolor{gelb}{RGB}{255,237,0}
\definecolor{orange}{RGB}{255,128,0}
\definecolor{magenta}{RGB}{227,0,102}
\definecolor{rot}{RGB}{204,7,30}
\definecolor{bordeaux}{RGB}{161,16,53}
\definecolor{violett}{RGB}{97,33,88}
\definecolor{lila}{RGB}{122,111,172}
\definecolor{grey}{gray}{0.7}
\definecolor{mittelblau}{RGB}{0,128,255}

\newcommand{\bigmid}{\mathrel{\big|}}
\newcommand{\Bigmid}{\mathrel{\Big|}}
\newcommand{\ceil}[1]{\left\lceil#1\right\rceil}
\newcommand{\floor}[1]{\left\lfloor#1\right\rfloor}
\newcommand{\angles}[1]{\left\langle#1\right\rangle}
\renewcommand{\tilde}{\widetilde}
\renewcommand{\hat}{\widehat}

\renewcommand{\vec}[1]{\boldsymbol{#1}}

\newcommand{\biglmulti}{{\big\{\hspace{-4.3pt}\big\{}}
\newcommand{\bigrmulti}{{\big\}\hspace{-4.3pt}\big\}}}
\newcommand{\Biglmulti}{{\Big\{\hspace{-4.5pt}\Big\{}}
\newcommand{\Bigrmulti}{{\Big\}\hspace{-4.5pt}\Big\}}}

\renewcommand{\phi}{\varphi}
\renewcommand{\epsilon}{\varepsilon}

\newcommand{\Nat}{{\mathbb N}}
\newcommand{\PNat}{{\mathbb N}_{>0}}
\newcommand{\Real}{{\mathbb R}}
\newcommand{\PReal}{{\mathbb R}_{>0}}

\newcommand{\Int}{{\mathbb Z}}
\newcommand{\Rat}{{\mathbb Q}}

\newcommand{\LL}{\textsf{\upshape L}}
\newcommand{\LC}{\textsf{\upshape C}}
\newcommand{\FO}{\textsf{\upshape FO}}

\newcommand{\CC}{{\mathcal C}}

\newcommand{\CG}{{\mathcal G}}

\newcommand{\CI}{{\mathcal I}}

\newcommand{\CN}{{\mathcal N}}

\newcommand{\CQ}{{\mathcal Q}}
\newcommand{\CR}{{\mathcal R}}
\newcommand{\CS}{{\mathcal S}}

\newcommand{\CU}{{\mathcal U}}

\newcommand{\CY}{{\mathcal Y}}

\newcommand{\Ca}{\mathcal a}
\newcommand{\Cb}{\mathcal b}

\newcommand{\Cr}{\mathcal r}

\newcommand{\Cu}{\mathcal u}
\newcommand{\Cx}{\mathcal x}
\newcommand{\Cy}{\mathcal y}
\newcommand{\Cz}{\mathcal z}

\newcommand{\CGS}{{\mathcal G\!\mathcal S}}

\newcommand{\FA}{{\mathfrak A}}

\newcommand{\FC}{{\mathfrak C}}
\newcommand{\FF}{{\mathfrak F}}
\newcommand{\FG}{{\mathfrak G}}
\newcommand{\FN}{{\mathfrak N}}
\newcommand{\FL}{{\mathfrak L}}
\newcommand{\FR}{{\mathfrak R}}

\newcommand{\Fa}{{\mathfrak a}}

\newcommand{\logic}[1]{\textsf{\upshape #1}}
\newcommand{\FOC}[1][]{\logic{FO$^{#1}$+C}}
\newcommand{\FOCnu}[1][]{\logic{FO$^{#1}$+C}_{\textup{nu}}}
\newcommand{\nuni}{{}_{\textup{nu}}}
\newcommand{\GC}{\logic{GFO+C}}
\newcommand{\GCgc}{\GC^{\textup{gc}}}
\newcommand{\GCnu}{\GC_{\textup{nu}}}
\newcommand{\GCgcnu}{\GC^{\textup{gc}}_{\textup{nu}}}
\newcommand{\MC}{\logic{MFO+C}}

\newcommand{\MCnu}{\MC_{\textup{nu}}}

\newcommand{\PTIME}{\logic{PTIME}}
\newcommand{\TC}{\logic{TC}}

\newcommand{\zero}{\logic{0}}
\newcommand{\one}{\logic{1}}
\newcommand{\ord}{\logic{ord}}

\newcommand{\les}{\leqslant}

\DeclareMathOperator{\free}{free}
\DeclareMathOperator{\relu}{relu}
\DeclareMathOperator{\lsig}{lsig}
\DeclareMathOperator{\sig}{sig}
\DeclareMathOperator{\id}{id}
\DeclareMathOperator{\ar}{ar}
\DeclareMathOperator{\bsize}{bsize}
\DeclareMathOperator{\size}{size}
\DeclareMathOperator{\wt}{wt}

\DeclareMathOperator{\depth}{dp}

\newcommand{\agg}{\logic{agg}}
\newcommand{\comb}{\logic{comb}}
\newcommand{\msg}{\logic{msg}}

\newcommand{\ro}{\logic{ro}}

\newcommand{\SUM}{\logic{SUM}}
\newcommand{\MEAN}{\logic{MEAN}}
\newcommand{\MAX}{\logic{MAX}}

\newcommand{\sem}[1]{\left\lsem#1\right\rsem}
\newcommand{\num}[1]{\left\llangle#1\right\rrangle}
\newcommand{\Lin}[1]{\left\llangle#1\right\rrangle}
\newcommand{\Fin}[1]{\left\llangle#1\right\rrangle}

\newcommand{\pos}[1]{\left\langle#1\right\rangle}

\newcommand{\Bit}{\operatorname{Bit}}

\newcommand{\DRat}{\Int\big[{\textstyle\frac{1}{2}}\big]}

\DeclareMathOperator{\tp}{tp}
\newcommand{\ttv}{\texttt{\upshape v}}
\newcommand{\ttn}{\texttt{\upshape n}}
\newcommand{\tta}{\texttt{\{}}
\newcommand{\ttz}{\texttt{\}}}
\newcommand{\rtp}[1]{{#1\to\texttt{\upshape r}}}
\newcommand{\Ltp}[1]{{#1\to\texttt{\upshape L}}}
\newcommand{\Ftp}[1]{{#1\to\texttt{\upshape F}}}

\newcommand{\bool}{{\textup{bool}}}

\newcommand{\inorm}[1]{\left\|#1\right\|_{\infty}}

\newcommand{\rexp}{\boldsymbol{\rho}}
\newcommand{\rr}{\rexp_{\textup{sg}}}
\newcommand{\rI}{\rexp_{\textup{Ind}}}
\newcommand{\rs}{\rexp_{\textup{dn}}}
\newcommand{\rt}{\rexp_{\textup{bd}}}
\newcommand{\Zr}{Z_{\textup{sg}}}
\newcommand{\ZI}{Z_{\textup{Ind}}}
\newcommand{\Zs}{Z_{\textup{dn}}}
\newcommand{\Zt}{Z_{\textup{bd}}}
\newcommand{\Zlen}{Z_{\textup{len}}}
\newcommand{\Zth}{\vec Z_{\textup{th}}}
\newcommand{\Zsl}{\vec Z_{\textup{sl}}}
\newcommand{\Zco}{\vec Z_{\textup{co}}}
\newcommand{\Zwt}{\vec Z_{\textup{wt}}}
\newcommand{\Zbi}{\vec Z_{\textup{bi}}}
\newcommand{\Zac}{\vec Z_{\textup{ac}}}
\newcommand{\ZV}{Z_{\textup{V}}}
\newcommand{\ZE}{Z_{\textup{E}}}

\newcommand{\emptytuple}{\emptyset}
\DeclareMathOperator{\crep}{crep}
\newcommand{\uni}{\CU_{[0,1]}}

\DeclareMathOperator{\bin}{bin}

   \newcommand{\tm}{{\tilde m}}
    \newcommand{\tM}{{\tilde M}}

    {$\blacksquare$}}]{myremark} 
    
    \title{The Descriptive Complexity of Graph Neural Networks}
\ThCSauthor{Martin Grohe}{grohe@informatik.rwth-aachen.de}[0000-0002-0292-9142]
\ThCSaffil{RWTH Aachen University, Germany}

\ThCSthanks{The conference version of this article was published in the \itshape{Proceedings of the 38th annual ACM/IEEE Symposium on Logic in Computer Science (LICS 2023)}. The research was funded by the European Union (ERC, SymSim,
101054974). Views and opinions expressed are however those of the
author(s) only and do not necessarily reflect those of the European
Union or the European Research Council. Neither the European Union
nor the granting authority can be held responsible for them.}

\ThCSshortnames{M.\ Grohe}  
\ThCSshorttitle{The Descriptive Complexity of Graph Neural Networks}
\ThCSyear{2024}
\ThCSarticlenum{25}
\ThCSreceived{Sep 6, 2023}
\ThCSrevised{Oct 10, 2024}
\ThCSaccepted{Oct 24, 2024}
\ThCSpublished{Dec 4, 2024}
\ThCSdoicreatedtrue
\ThCSkeywords{Graph neural networks; circuit complexity; descriptive complexity; first-order logic with counting}

\begin{document}

\maketitle

\begin{abstract}
  We analyse the power of graph neural networks (GNNs) in terms of
  Boolean circuit complexity and descriptive complexity.

  We prove that the graph queries that can be computed by a
  polynomial-size bounded-depth family of GNNs are exactly those
  definable in the guarded fragment $\GC$ of first-order logic with
  counting and with built-in relations. This puts GNNs in the circuit
  complexity class (non-uniform) $\TC^0$. Remarkably, the GNN families may use
  arbitrary real weights and a wide class of activation functions
  that includes the standard ReLU, logistic ``sigmoid'', and hyperbolic
  tangent functions. If the GNNs are allowed to use random
  initialisation and global readout (both standard features of GNNs
  widely used in practice), they can compute exactly the same queries
  as bounded depth Boolean circuits with threshold gates, that is,
  exactly the queries in $\TC^0$.

  Moreover, we show that queries computable by a single GNN with
  piecewise linear activations and rational weights are definable in
  $\GC$ without built-in relations. Therefore, they are contained in
  uniform $\TC^0$.
\end{abstract}

\section{Introduction}

Graph neural networks (GNNs) \cite{GilmerSRVD17,ScarselliGTHM09} are
deep learning models for graph data that play a key role in machine learning on graphs
(see, for example, \cite{ChamiAPRM22}). A GNN
describes a distributed algorithm carrying out local computations at
the vertices of the input graph. At any time, each vertex has a
``state'', which is a vector of reals, and in each computation step it sends a
message to all its neighbours. The messages are also vectors of reals,
and they only depend on the current state of the sender and the
receiver. Every vertex aggregates the messages it receives and
computes its new state depending on the old state and the aggregated
messages. The message and state-update functions 
are computed by feedforward neural networks whose parameters are
learned from data.

In this article, we study the \emph{expressiveness} of GNNs: which
functions on graphs or their vertices can be computed by GNNs? We
provide answers in terms of Boolean circuits and logic, that is,
computation models of classical (descriptive) complexity theory. An
interesting and nontrivial aspect of this is that GNNs are
``analogue'' computation models operating on and with real
numbers. The weights of neural networks may be arbitrary reals, and
the activation functions may even be transcendental functions such as
the logistic function $x\mapsto\frac{1}{1+e^{-x}}$.

We always want functions on graphs to be
\emph{isomorphism invariant}, that is, isomorphic graphs are mapped to
the same value. Similarly, we want functions on vertices to be
\emph{equivariant}, that is, if $v$ is a vertex of a graph $G$ and $f$
is an isomorphism from $G$ to a graph $H$, then $v$ and $f(v)$ are
mapped to the same value. Functions computed by GNNs are always
invariant or equivariant, and so are functions defined in logic
(a.k.a.~\emph{queries}).

In a machine learning context, it is usually assumed that the vertices
of the input graph are equipped with additional features in the form
of vectors over the reals; we speak of \emph{graph signals} in this
article. The function values are also vectors over the reals. Thus a
function on the vertices of a graph is an equivariant transformation
between graph signals. When comparing GNNs and logics or Boolean
circuits, we focus on Boolean functions, where the input signal is
Boolean, that is, it associates a $\{0,1\}$-vector with every vertex
of the input graph, and the output is just a Boolean value $0$ or
$1$. In the logical context, it is natural to view Boolean signals as
vertex labels. Thus a Boolean signal in $\{0,1\}^k$ is described as a
sequence of $k$ unary relations on the input graph. Then an invariant
Boolean function becomes a \emph{Boolean query} on labelled graphs,
and an equivariant Boolean function on the vertices becomes a
\emph{unary query}.  To streamline the presentation, in this article we
focus on unary queries and equivariant functions on the vertices. All
our results also have versions for Boolean queries and functions on
graphs, but we only discuss these in occasional remarks. While we are
mainly interested in queries (that is, Boolean functions), our results
also have versions for functions with arbitrary real input and output
signals. These are needed for the proofs anyway. But since the exact
statements become unwieldy, we keep them out of the
introduction. Before discussing further background, let us state our
central result.

\begin{theorem}\label{theo:main1}
  Let ${\CQ}$ be a unary query on labelled graphs. Then the following are
  equivalent.
  \begin{enumerate}
  \item ${\CQ}$ is computable by a polynomial-weight bounded-depth family
    of GNNs with rpl-approxi\-mable activation functions.
  \item ${\CQ}$ is definable in the guarded fragment $\GCnu$ of first-order logic
    with counting and with built-in relations.
  \end{enumerate}
\end{theorem}

The result requires more explanations.  First of all, it is a
\emph{non-uniform} result, speaking about computability by families of
GNNs and a logic with built-in relations. A family
$\CN=(\FN^{(n)})_{n\in\Nat}$ of GNNs consists of a GNN $\FN^{(n)}$ for
input graphs of size $n$, for each $n\in\Nat$. \emph{Bounded depth}
refers to the number of message passing rounds, or layers, of the GNNs
as well as the depth of the feed-forward neural networks they use for
their message and state-update functions. We would like the GNN
$\FN^{(n)}$ to be of ``size'' polynomial in $n$, but since we allow
arbitrary reals as parameters of the neural networks, it is not clear
what this actually means. We define the \emph{weight} of a GNN to be
the number of computation nodes of the underlying neural networks plus
the sum of the absolute values of all parameters of these
networks. The class of \emph{rpl-approximable} functions (see Section~\ref{sec:func}; ``rpl'' abbreviates rational piecewise linear) contains all
functions that are commonly used as activation functions for neural
networks, for example, the rectified linear unit, the logistic
function, the hyperbolic tangent function, the scaled exponential
linear unit (see Section~\ref{sec:fnn} for background on neural
networks and their activation functions).

On the logical side, \emph{first-order logic with counting} $\FOC$ is
the two-sorted extension of first-order logic over relational
structures that has variables ranging over the non-negative integers,
bounded arithmetic on the integer side, and counting terms that give
the number of assignments satisfying a formula. In the
\emph{$k$-variable fragment} $\FOC[k]$, only $k$ variables ranging
over the vertices of the input graphs are allowed (but arbitrarily
many variables for the integer part). The \emph{guarded fragment}
$\GC$ is a fragment of $\FOC[2]$ where quantification over vertices is
restricted to the neighbours of the current vertex. \emph{Built-in
  relations} are commonly used in descriptive complexity to introduce
non-uniformity to logics and compare them to non-uniform circuit
complexity classes. Formally, they are just arbitrary relations on the
non-negative integers that the logic can access, independently of the
input structure. $\GCnu$ denotes the extension of $\GC$ by built-in
relations.

It is well-known that over ordered input structures, $\FOC$
with built-in relations captures the circuit complexity class
(non-uniform) $\TC^0$, consisting of Boolean functions (in our context:
queries) that are
computable by families of bounded-depth polynomial-size Boolean
circuits with threshold gates. This implies that, as a corollary to
Theorem~\ref{theo:main1}, we get the following.

\begin{corollary}\label{cor:main1}
  Every unary query that is computable by a polynomial-weight
  bounded-depth family of GNNs with rpl approximable activation
  functions is in $\TC^0$.
\end{corollary}

The strength of GNNs can be increased
by extending the input signals with a random component
\cite{AbboudCGL21,SatoYK21}. In \cite{AbboudCGL21}, it was even proved that such
\emph{GNNs with random initialisation} can approximate all
functions on graphs. The caveat of this result is that it is
non-uniform and that input graphs of size $n$ require GNNs of
size exponential in $n$ and depth linear in $n$. We ask which queries
can be computed by polynomial-weight, bounded-depth families of
GNNs. Surprisingly, this gives us a converse of
Corollary~\ref{cor:main1} and thus a characterisation of $\TC^0$.

\begin{theorem}\label{theo:main2}
  Let ${\CQ}$ be a unary query on labelled graphs. Then the following are
  equivalent.
  \begin{enumerate}
  \item ${\CQ}$ is computable by a polynomial-weight bounded-depth family
    of GNNs with random initialisation, global readout, and with rpl approximable activation functions.
  \item ${\CQ}$ is computable in $\TC^0$.
  \end{enumerate}
\end{theorem}

For a GNN with random initialisation to compute a query, it needs to
compute the correct answer with high probability, taken over the
random inputs. (We demand a probability $\ge 3/4$, but the exact value
is irrelevant).

Following \cite{AbboudCGL21}, we allow GNNs with
random initialisation to also use a feature known as \emph{global
  readout}, which means that in each message-passing round of a GNN
computation, the vertices not only receive messages from their
neighbours, but the aggregated state of all vertices. There is also a
version of Theorem~\ref{theo:main1} for GNNs with global readout. It
is an open question what the exact expressiveness of polynomial-weight
bounded-depth families of GNNs with random initialisation, but without
global readout, is.

\subsection*{Related Work}
A fundamental result on the expressiveness of GNNs
\cite{MorrisRFHLRG19,XuHLJ19} states that two graphs are
distinguishable by a GNN if and only if they are distinguishable by
the 1-dimensional Weisfeiler-Leman (WL) algorithm, a simple
combinatorial algorithm originally introduced as a graph isomorphism
heuristics \cite{Morgan65,WeisfeilerL68}. This result has had
considerable impact on the subsequent development of GNNs, because it
provides a yardstick for the expressiveness of GNN extensions (see
\cite{MorrisLMRKGFB23}). Its generalisation to higher-order GNNs and
higher-dimensional WL algorithms \cite{MorrisRFHLRG19} even gives a
hierarchy of increasingly more expressive formalisms against which
such extensions can be compared. However, these results relating GNNs
and their extensions to the WL algorithm only consider a restricted
form of expressiveness, the power to distinguish two
graphs. Furthermore, the results are \emph{non-uniform}, that is, the
distinguishing GNNs depend on the input graphs or at least on their
size, and the GNNs may be arbitrarily large and deep.  Indeed, the
GNNs from the construction in \cite{XuHLJ19} may be exponentially
large in the graphs they distinguish. Those of \cite{MorrisRFHLRG19}
are polynomial. Both have recently been improved by
\cite{AamandCINRSSW22}, mainly showing that the messages only need to
contain logarithmically many bits.

We are not the first to study the logical expressiveness of GNNs (see
\cite{Grohe21} for a recent survey). It
was proved in \cite{BarceloKM0RS20} that all unary queries definable in the
guarded fragment $\logic{GC}$ of the extension $\LC$ of first-order
logic by counting quantifiers $\exists^{\ge n}x$ (``there exist at
least $n$ vertices $x$ satisfying some formula'') are computable by a
GNN. The logic $\logic{GC}$ is weaker than our $\GC$ in that it does
not treat the numbers $n$ in the quantifiers $\exists^{\ge n}x$ as
variables, but as fixed constants. What is
interesting about this result, and what makes it incomparable to ours,
is that it is a \emph{uniform} result: a query definable in
$\logic{GC}$ is computable by a single GNN across all graph
sizes. There is a partial converse to this result, also from
\cite{BarceloKM0RS20}: all unary queries that are definable in first-order
logic and computable by a GNN are actually definable in
$\logic{GC}$. Note, however, that there are queries computable
by GNNs that are not definable in first-order logic.

A different approach to capturing GNNs by logic has been taken in
\cite{GeertsSV22}. There, the authors introduce a new logic
$\logic{MPLang }$ that operates directly on the reals. The logic, also
a guarded (or modal) logic, is simple and elegant and well-suited to
translate GNN computations to logic. The converse translation is more
problematic, though. But to be fair, it is also in our case, where it requires
families of GNNs and hence non-uniformity. However, the purpose of the
work in \cite{GeertsSV22} is quite different from ours. It is our goal
to describe GNN computations in terms of standard descriptive
complexity and thus to be able to quantify the computational power of
GNNs in the framework of classical complexity. It is the goal
of \cite{GeertsSV22} to expand logical reasoning to real-number
computations in a way that is well-suited to GNN computations. Of
course, both are valid goals.

There is another line of work that is important for us. In the 1990s,
researchers studied the expressiveness of feedforward neural networks
(FNNs) and compared it to Boolean computation models such as Turing
machines and circuits (for example,
\cite{KarpinskiM97,Maass97,MaassSS91,SiegelmannS95}). Like GNNs, FNNs
are analogue computation models operating on the reals, and this work
is in the same spirit as ours. An FNN has fixed numbers $p$ of inputs
and $q$ of outputs, and it thus computes a function from $\Real^p$ to
$\Real^q$. Restricted to Boolean inputs, we can use FNNs with $p$
inputs to decide subsets of $\{0,1\}^p$, and we can use families of
FNNs to decide languages. It was proved in \cite{Maass97} that a
language is decidable by a family of bounded-depth polynomial-weight
FNNs using piecewise-polynomial activation functions if and only if it
is in $\TC^0$. It may seem that our Corollary~\ref{cor:main1}, at
least for GNNs with piecewise-linear (or even piecewise-polynomial)
activations, follows easily from this result. But this is not the
case, because when processing graphs, the inputs to the FNNs computing
the message and update functions of the GNN may become large through
the aggregation ranging over all neighbours. Also, the arguments of
\cite{Maass97} do not extend to rpl-approximable activation functions
like the logistic function. There has been related work
\cite{KarpinskiM97} that extends to a wider class of activation
functions including the logistic function, using arguments based on
o-minimality.  But the results go into a different direction; they
bound the VC dimension of FNN architectures and do not relate them to
circuit complexity.

\subsection*{Techniques}

The first step in proving the difficult implication (1)$\Rightarrow$(2)
of Theorem~\ref{theo:main1} is to prove a uniform result for a simpler
class of GNNs; this may be of independent interest.

\begin{theorem}\label{theo:main3}
  Let ${\CQ}$ be a unary query computable by a GNN with rational
  weights and piecewise linear activations. Then ${\CQ}$ is definable in
  $\GC$. 
\end{theorem}

Compare this with the result of \cite{BarceloKM0RS20}: every query
definable in the (weaker) logic $\logic{GC}$ is computable by a GNN
and in fact a GNN with rational weights and piecewise linear
activations. Thus we may write
$ \logic{GC}\subseteq\logic{GNN}\subseteq\GC.  $ It is not hard to
show that both inclusions are strict.

To prove Theorem~\ref{theo:main3}, we need to show that the rational
arithmetic involved in GNN computations, including
unbounded linear combinations, can be simulated, at least
approximately, in the logic $\GC$. Establishing this is a substantial
part of this article, and it may be of independent interest.

So how do we prove the forward implication of Theorem~\ref{theo:main1}
from Theorem~\ref{theo:main3}? It was our first idea to look at the
results for FNNs. In principle, we could use the linear-programming
arguments of \cite{Maass97}. This would probably work, but would be
limited to piecewise linear or piecewise polynomial activations. We
could then use o-minimality to extend our results to wider classes of
activation functions. After all, o-minimality was also applied
successfully in the somewhat related setting of constraint databases
\cite{KuperLP00}. Again, this might work, but our analytical approach
seems simpler and more straightforward. Essentially, we use the
Lipschitz continuity of the functions computed by FNNs to show that we
can approximate arbitrary GNNs with rpl-approximable activations by
GNNs with rational weights and piecewise linear activations, and then
we apply Theorem~\ref{theo:main3}. 

Let us close the introduction with a few remarks on
Theorem~\ref{theo:main2}. The reader may have noted that assertion (1)
of the theorem involves randomness in the computation model, whereas
(2) does not. To prove the implication (1)$\Rightarrow$(2) we use the
well-known ``Adleman Trick'' that allows us to trade randomness for
non-uniformity. To prove the converse implication, the main insight is
that with high probability, random initialisation gives us a
linear order on the vertices of the input graph. Then we can use the
known fact that $\FOC$ with built-in relations captures $\TC^0$ on
ordered structures.

\subsection*{Structure of This Article}
After collecting preliminaries from different areas in 
Section~\ref{sec:prel}, in Section~\ref{sec:logic} we develop a machinery for carrying out the required
rational arithmetic in first-order logic with counting and its guarded
fragment. This is a significant
part of this article, which is purely logical and independent of neural networks. We then introduce GNNs
(Section~\ref{sec:gnn}) and prove the uniform Theorem~\ref{theo:main3}
(Section~\ref{sec:uniform}). We prove the forward direction of
Theorem~\ref{theo:main1} in Section~\ref{sec:nonuniform} and the
backward direction in Section~\ref{sec:converse}. Finally, we prove
Theorem~\ref{theo:main2} in Section~\ref{sec:ri}.

\section{Preliminaries}
\label{sec:prel}

By $\Int,\Nat,\PNat,\Rat,\Real$ we denote the sets of integers, nonnegative integers,
positive integers, rational numbers, and real numbers,
respectively. Instead of arbitrary rationals, we will often work with
\emph{dyadic rationals}, that is, rationals of the form $\frac{n}{2^\ell}$ for
$p\in\Int, \ell\in\Nat$. These are precisely the numbers that have a presentation
as finite
precision binary floating point numbers. We denote the set of dyadic
rationals by $\DRat$. 

We denote the binary
representation of $n\in\Nat$ by $\bin(n)$. The \emph{bitsize} of $n$
is the length of the
binary representation, that is,
\[
  \bsize(n)\coloneqq|\bin(n)|=
  \begin{cases}
    1&\text{if }n=0,\\
    \ceil{\log(n+1)}&\text{if }n>0,
  \end{cases}
\]
where $\log$ denotes the binary logarithm.  
We
denote the $i$th bit of the binary representation of 
$n\in\Nat$ by $\Bit(i,n)$, where we count bits starting from $0$ with
the lowest significant bit. It will be convenient to let $\Bit(i,n)\coloneqq 0$ for all
$i\ge\bsize(n)$.  So
\[
  n=\sum_{i=0}^{\bsize(n)-1}\Bit(i,n)\cdot 2^i=\sum_{i\in\Nat}\Bit(i,n)\cdot 2^i.
\]
The \emph{bitsize} of 
an integer $n\in\Int$ is
$\bsize(n)\coloneqq1+\bsize(|n|)$, and the \emph{bitsize}
of a dyadic rational $q=\frac{n}{2^{\ell}}\in\DRat$ in reduced form is
  $\bsize(q)\coloneqq\bsize(n)+\ell+1$.

We denote tuples (of numbers, variables, vertices, et cetera) using
boldface letters. Usually, a $k$-tuple $\vec t$ has entries
$t_1,\ldots,t_k$. The empty tuple is denoted $\emptytuple$ (just like
the empty set, this should never lead to any confusion), and for every
set $S$ we have $S^0=\{\emptytuple\}$. For tuples
$\vec t=(t_1,\ldots,t_k)$ and $\vec u=(u_1,\ldots,u_\ell)$, we let
$\vec t\vec u=(t_1,\ldots,t_k,u_1,\ldots,u_\ell)$. To improve
readability, we often write $(\vec t,\vec u)$ instead of
$\vec t\vec u$. This does not lead to any confusion, because we never
consider nested tuples.

For a vector $\vec
x=(x_1,\ldots,x_k)\in\Real^k$, the \emph{$\ell_1$-norm} (a.k.a.\
Manhattan norm) is $\|\vec x\|_1\coloneqq \sum_{i=1}^k|x_i|$,  the \emph{$\ell_2$-norm} (a.k.a.\
Euclidean norm) is $\|\vec x\|_2\coloneqq \sqrt{\sum_{i=1}^kx_i^2}$,
and the \emph{$\ell_\infty$-norm} (a.k.a.~maximum norm)
$\|\vec x\|_\infty\coloneqq\max_{i\in[k]}|x_i|$. As $\frac{1}{k}\|\vec
x\|_1\le\|\vec x\|_\infty\le\|\vec
x\|_2\le\|\vec x\|_1$, it does not make much of a difference
which norm we use; most often, it will be convenient for us to 
use the $\ell_{\infty}$-norm.

\subsection{Functions and Approximations}
\label{sec:func}
A function $f:\Real^p\to\Real^q$ is \emph{Lipschitz continuous} if there is some
constant $\lambda$, called a \emph{Lipschitz constant} for $f$, such
that for all $\vec x,\vec y\in\Real^p$ it holds that
$\inorm{f(\vec x)-f(\vec y)}\le \lambda\inorm{\vec x-\vec y}$. 

A function $L:\Real\to\Real$ is \emph{piecewise linear} if there are
$n\in\Nat$, $a_0,\ldots,a_n$, $b_0,\ldots,b_n$,
$t_1,\ldots,t_n\in\Real$ such that $t_1<t_2<\ldots<t_n$ and
\[
  L(x)=
  \begin{cases}
    a_0x+b_0&\text{if }x< t_1,\\
    a_ix+b_i&\text{if }t_i\le x<t_{i+1}\text{ for some }i<n,\\
    a_nx+b_n&\text{if }x\ge t_n
  \end{cases}
\]
if $n\ge 1$, or $L(x)=a_0x+b_0$ for all $x$ if
$n=0$. Note that there is a unique \emph{minimal representation} of
$L$ with minimal number $n+1$ of pieces. We call
$t_1,\ldots,t_n$ in the minimal representation of
$L$ the \emph{thresholds} of
$L$; these are precisely the points where $L$ is non-linear.
$L$ is \emph{rational} if all its parameters
$a_i,b_i,t_i$ in the minimal representation are dyadic
rationals.\footnote{Throughout this article we work with dyadic
  rationals. For this reason, we are a little sloppy in our
  terminology. For example, we call a function ``rational piecewise
  linear'' when the more precise term would be ``dyadic-rational
  piecwise linear''.} If
$L$ is rational, then its \emph{bitsize} of
$\bsize(L)$ is the sum of the bitsizes of all the parameters
$a_i,b_i,t_i$ of the minimal representation. Oberserve that if
$L$ is continuous then it is Lipschitz continuous with Lipschitz
constant $\max_{0\le i\le n}a_i$.

\begin{example}\label{exa:pl-activation}
  The most important example of a rational piecewise linear function
  for us is the \emph{rectified linear unit} $\relu:\Real\to\Real$
  defined by $\relu(x)\coloneqq\max\{0,x\}$.

  In fact, it is not hard to see that every piecewise linear function
  can be written as a linear combination of $\relu$-terms. For
  example, the \emph{identity} function $\id(x)=x$ can be written as
  $\relu(x)-\relu(-x)$, and the \emph{linearised sigmoid} function
  $\lsig:\Real\to\Real$, defined by $\lsig(x)=0$ if $x<0$, $\lsig(x)=x$
  if $0\le x< 1$, an $\lsig(x)=1$ if $x\ge 1$, can be written as
  $\relu(x)-\relu(x-1)$.
  \uend
\end{example}

We need a notion of approximation between functions on the
reals. Let $f,g:\Real\to\Real$ and $\epsilon\in\PReal$. Then
$g$ is an \emph{$\epsilon$-approximation} of $f$ if for all
$x\in\Real$ it holds that
\[
  \big|f(x)-g(x)\big|\le\epsilon |f(x)|+\epsilon.
\]
Note that we allow for both an additive and a multiplicative
approximation error. This notion of approximation is not symmetric, but if $g$ $\epsilon$-approximates $f$
for some $\epsilon < 1$ then $f$
$\frac{\epsilon}{1-\epsilon}$-approximates $g$. The main reason we
need to allow for a multiplicative approximation error is that we want
to approximate linear functions with irrational coefficients by linear
functions with rational coefficients.

We call a function $f:\Real\to\Real$ \emph{rpl-approximable} if
for every $\epsilon>0$ there is a continuous rational piecewise linear
function $L$ of bitsize polynomial in $\epsilon^{-1}$ that
$\epsilon$-approximates $f$.

\begin{example}\label{exa:rpl-app}
  The \emph{logistic function} $\sig(x)=\frac{1}{1+e^{-x}}$ and the \emph{hyperbolic
  tangent} $\tanh(x)=\frac{e^x-e^{-x}}{e^x+e^{-x}}$ are
  rpl-approximable. Examples of unbounded rpl-approximable
  functions are the \emph{soft plus function} $\ln(1+e^x)$ and the
  \emph{exponential linear units} defined by $\operatorname{elu}_\alpha(c)=x$
  if $x>0$ and $\alpha(e^x-1)$ if $x\le 0$, where $\alpha>0$ is a constant. We omit
  the straightforward proofs based on simple calculus.
  \uend
\end{example}

\begin{example}
  Examples of functions that are not rpl approximable are functions
  that are not Lipschitz continuous, such as the square function, and
  periodic functions such as the sine or cosine functions.
  \uend
\end{example}

\subsection{Graphs and Signals}
Graphs play two different roles in this article: they are the basic data
structures on which logics and graph neural networks operate, and they
form the skeleton of Boolean circuits and neural networks. In the
first role, which is the default, we assume graphs to be
undirected. This assumption is not essential for our results, but
convenient. 
In the second role, graphs are directed acyclic graphs (dags).

We always denote the vertex set of a graph or dag $G$ by $V(G)$ and
the edge set by $E(G)$. We denote edges by $vw$ (without
parentheses). We assume the vertex set of all graphs in this article to be finite and
nonempty. The \emph{order} of a graph $G$ is
$|G|\coloneqq|V(G)|$, and the \emph{bitsize} $\bsize(G)$ of $G$ is the
size of a representation of $G$. (For simplicity, we can just take
adjacency matrices, then $\bsize(G)=|G|^2$.) The class of all (undirected) graphs is
denoted by $\CG$. 

For a vertex $v$ in an (undirected) graph, we let $N_G(v)\coloneqq\{w\in
V(G)\mid vw\in E(G)\}$ be the \emph{neighbourhood} of $v$ in $G$, and we let
$N_G[v]\coloneqq\{v\}\cup N_G(v)$ be the \emph{closed
  neighbourhood}. Furthermore, we let $\deg_G(v)\coloneqq|N_G(v)|$ be
the \emph{degree} of $v$ in $G$. For a vertex $v$ in a directed graph $G$, we
let $N^+_G(v)\coloneqq\{w\in
V(G)\mid vw\in E(G)\}$ be the \emph{out-neighbourhood} of $v$ and $N^-_G(v)\coloneqq\{u\in
V(G)\mid uv\in E(G)\}$ the \emph{in-neighbourhood}, and we let
$\deg^+_G(v)\coloneqq|N^+_G(v)|$ and $\deg^-_G(v)\coloneqq|N^-_G(v)|$
be the \emph{out-degree} and \emph{in-degree}. We call nodes of in-degree $0$ \emph{sources} and nodes of
out-degree $0$ \emph{sinks}. The \emph{depth} $\depth_G(v)$ of node $v$ in a dag $G$ is the
length of the longest path from a source to $v$. The \emph{depth}
$\depth(G)$ of a
dag $G$ is the maximum depth of a sink of $G$. In notations such as $N_G,\deg_G$
we omit the index ${}_G$ if the graph is clear from the
context. 

When serving as data for graph neural networks, the vertices of graphs
usually have real-valued features, which we call \emph{graph signals}.
An \emph{$\ell$-dimensional signal} on a graph $G$ is a function
$\Cx:V(G)\to\Real^\ell$. We denote the class of all $\ell$-dimensional
signals on $G$ by $\CS_\ell(G)$ and the class of all pairs $(G,\Cx)$,
where $\Cx$ is an $\ell$-dimensional signal on $G$, by $\CGS_\ell$.
An $\ell$-dimensional signal is \emph{Boolean} if its range is
contained in $\{0,1\}^\ell$. By $\CS^\bool_\ell(G)$ and $\CGS^\bool_\ell$ we
denote the restrictions of the two classes to Boolean signals.

Isomorphisms between pairs $(G,\Cx)\in\CGS_\ell$ are required to
preserve the signals. We call a mapping $f:\CGS_\ell\to\CGS_m$ a \emph{signal
  transformation} if for all $(G,\Cx)\in\CGS_\ell$ we have
$f(G,\Cx)=(G,\Cx')$ for some $\Cx'\in\CS_m(G)$. Such a signal
transformation $f$
is \emph{equivariant} if for all isomorphic
$(G,\Cx),(H,\Cy)\in\CGS_\ell$, every isomorphism
$h$ from $(G,\Cx)$ to $(H,\Cy)$ is also an isomorphism from $f(G,\Cx)$
to $f(H,\Cy)$.

We can view signals $\Cx\in\CS_\ell(G)$ as matrices in the space
$\Real^{V(G)\times\ell}$. Flattening them to vectors of length
$|G|\ell$, we can apply the usual vector norms to graph signals. In
particular, we have
$\inorm{\Cx}=\max\big\{\inorm{\Cx(v)}\bigmid v\in V(G)\big\}$. Sometimes, we need
to restrict a signal to a subsets of $W\in V(G)$. We denote this
restriction by $\Cx|_W$, which may be viewed as a matrix in $\Real^{W\times \ell}$.

\subsection{Boolean Circuits}
A \emph{Boolean circuit} $\FC$ is a dag where all
nodes except for the sources are
labelled as \emph{negation},
\emph{disjunction}, or \emph{conjunction} nodes. Negation nodes must
have in-degree $1$.  Sources are \emph{input nodes}, and we always
denote them by $X_1,\ldots,X_p$. Similarly, sinks are \emph{output
  nodes}, and we denote them by $Y_1,\ldots,Y_q$. The number $p$ of input nodes is the \emph{input dimension} of
$\FC$, and the number $q$ of output nodes the \emph{output dimension}.
Most of the time, we consider circuits
that also have \emph{threshold nodes} of arbitrary positive in-degree,
where a $\ge t$-threshold node evaluates to $1$ if at least $t$ of its
in-neighbours evaluate to $1$. To distinguish them from the Boolean
circuits over the standard basis we refer to such circuits as
\emph{threshold circuits}.
The \emph{depth} $\depth(\FC)$ of a circuit $\FC$ is the maximum
length of a path from an input node to an output node. The
\emph{order} $|\FC|$ of $\FC$ is the number of nodes, and the
\emph{size} is the number of nodes plus the number of edges.

A circuit $\FC$ of input dimension $p$ and output dimension $q$
computes a function $f_\FC:\{0,1\}^p\to\{0,1\}^q$ defined in the
natural way. To simplify the notation, we simply denote this function
by $\FC$, that is, we write $\FC(\vec x)$ instead of $f_{\FC}(\vec x)$
to denote the output of $\FC$ on input $\vec x\in\{0,1\}^p$.

In complexity theory, we study which languages
$L\subseteq\{0,1\}^*$  or functions $F:\{0,1\}^*\to\{0,1\}^*$ can
be computed by families $\CC=(\FC_n)_{n\in\PNat}$ of circuits, where
$\FC_n$ is a circuit of input dimension $n$. Such a family $\CC$ \emph{computes}
$F$ if for all $n\in\PNat$, $\FC_n$ computes the restriction $F_n$ of
$F$ to $\{0,1\}^n$. We say that $\CC$ \emph{decides} $L$ if it
computes its characteristic function. \emph{Non-uniform $\TC^0$} is
the class of all languages that are decided by a family
$\CC=(\FC_n)_{n\in\PNat}$ of threshold circuits of \emph{bounded
  depth} and \emph{polynomial size}.
There is also a class \emph{(dlogtime)
  uniform $\TC^0$} where the family $\CC$ itself is required to be
easily computable; we refer the reader to \cite{BarringtonIS90}. We
will never work directly with uniform circuit families, but instead
use a logical characterisation in terms of first-order logic with
counting (Theorem~\ref{theo:bis}).

An important fact that we shall use is that standard arithmetic
functions on the bit representations of natural numbers can be
computed by bounded-depth polynomial-size threshold circuits. We
define $\operatorname{ADD}_{2n}:\{0,1\}^{2n}\to\{0,1\}^{n+1}$ to be
the bitwise addition of two $n$-bit numbers (whose result may be an
$(n+1)$-bit number). We let $\operatorname{ADD}:\{0,1\}^*\to\{0,1\}^*$
be the function that coincides with $\operatorname{ADD}_{2n}$ on
inputs of even size and maps all inputs of odd size to $0$. Similarly,
we define $\operatorname{SUB}_{2n}:\{0,1\}^{2n}\to\{0,1\}^{2n}$ and
$\operatorname{SUB}:\{0,1\}^*\to\{0,1\}^*$ for the truncated
subtraction $m\dotminus n\coloneqq\max\{0,m-n\}$, $\operatorname{MUL}_{2n}:\{0,1\}^{2n}\to\{0,1\}^{2n}$ and
$\operatorname{MUL}:\{0,1\}^*\to\{0,1\}^*$ for multiplication.  We
also introduce a binary integer division function
$\operatorname{DIV}_{2n}:\{0,1\}^{2n}\to\{0,1\}^{n}$ mapping $n$-bit
numbers $k,\ell$ to $\floor{k/\ell}$ (with some default value, say
$0$, if $\ell=0$).  The \emph{iterated addition} function
$\operatorname{ITADD}_{n^2}:\{0,1\}^{n^2}\to\{0,1\}^{2n}$ and the
derived $\operatorname{ITADD}:\{0,1\}^*\to\{0,1\}^*$ add $n$
numbers of $n$-bits each. Finally, we need the less-than-or-equal-to
predicate $\operatorname{LEQ}_{2n}:\{0,1\}^{2n}\to\{0,1\}$ and $\operatorname{LEQ}:\{0,1\}^{*}\to\{0,1\}$.

\begin{lemma}[\cite{BarringtonIS90,ChandraSV84},\cite{Hesse01}]
  \label{lem:tc0ar}
  $\operatorname{ADD}$, $\operatorname{MUL}$, $\operatorname{DIV}$, $\operatorname{ITADD}$, and $\operatorname{LEQ}$ are computable by
 dlogtime uniform families of bounded-depth polynomial-size threshold
  circuits. 
\end{lemma}

The fact that $\operatorname{ADD}$, $\operatorname{MUL}$,
$\operatorname{ITADD}$, and $\operatorname{LEQ}$ are computable by
families of bounded-depth polynomial-size threshold circuits goes back
to \cite{ChandraSV84} (also see \cite{Vollmer99}). The arguments given
there are non-uniform, but it is not hard to see that they can be
``uniformised'' \cite{BarringtonIS90}. The situation for $\operatorname{DIV}$ is more
complicated. It was known since the mid 1980s that
$\operatorname{DIV}$ is computable by a non-uniform (or
polynomial-time uniform) family of bounded-depth polynomial-size
threshold circuits, but the uniformity was only established 15 years
later in \cite{Hesse01,HesseAB02}.

\subsection{Feedforward Neural Networks}
\label{sec:fnn}
It will be convenient for us to formalise feedforward neural networks
(a.k.a.~multilayer perceptrons, MLPs) in a similar way as Boolean
circuits. A more standard ``layered'' presentation of FNNs can easily be seen as a special case.  A
\emph{feedforward neural network architecture} $\FA$ is a triple
\mbox{$\big(V,E,(\Fa_v)_{v\in V}\big)$}, where $(V,E)$ is a
directed acyclic graph that we call the \emph{skeleton} of $\FA$ and
for every vertex $v\in V$, $\Fa_v:\Real\to\Real$ is a continuous
function that we call the \emph{activation function} at $v$. A
\emph{feedforward neural network (FNN)} is a tuple
$\FF=(V,E, (\Fa_v)_{v\in V},\vec w,\vec b)$, where
$\big(V,E,(\Fa_v)_{v\in V} \big)$ is an FNN architecture,
$\vec w=(w_e)_{e\in E}\in\Real^E$ associates a \emph{weight} $w_e$
with every edge $e\in E$, and $\vec b=(b_v)_{v\in V}\in\Real^{V}$
associates a \emph{bias} $b_v$ with every node $v\in V$. As for
circuits, the sources of the dag are \emph{input nodes}, and we denote them by
$X_1,\ldots,X_p$. Sinks are \emph{output
  nodes}, and we denote them by $Y_1,\ldots,Y_q$.  We define the
\emph{order} $|\FF|$, the \emph{depth} $\depth(\FF)$, the \emph{input
  dimension}, and the \emph{output dimension} of $\FF$ in the same way
as we did for circuits.

To define the semantics, let $\FA=\big(V,E,(\Fa_v)_{v\in V}\big)$ be
an FNN architecture of input dimension $p$ and output dimension $q$. For each node $v\in V$, we define a function
$f_{\FA,v}:\Real^p\times\Real^E\times\Real^{V}\to\Real$
inductively as follows. Let $\vec x=(x_1,\ldots,x_p)\in\Real^p$, $\vec
w=(w_e)_{e\in E}\in\Real^E$, and $\vec
b=(b_v)_{v\in V}\in\Real^{V}$.
  Then
  \[
    f_{\FA,v}(\vec x,\vec w,\vec b)\coloneqq
    \begin{cases}
      x_i&\text{if $v$ is the input node $X_i$},\\
      \Fa_v\left(b_v+\sum_{v'\in N^-(v)}f_{\FA,v'}(\vec x,\vec
      w,\vec b)\cdot w_{v'v}\right)&\text{if $v$ is not an input node}.
    \end{cases}
  \]
  We define $f_{\FA}: \Real^p\times\Real^E\times\Real^{V}\to\Real^q$ by
  \[
    f_{\FA}(\vec x,\vec w,\vec b)\coloneqq\big(f_{\FA,Y_1}(\vec x,\vec w,\vec
    b),\ldots, f_{\FA,Y_q}(\vec x,\vec w,\vec
    b) \big).
  \]
  For an FNN $\FF=\big(V,E,
  (\Fa_v)_{v\in V},\vec w,\vec b\big)$ with architecture $\FA=\big(V,E,
  (\Fa_v)_{v\in V} \big)$ we define
  functions $f_{\FF,v}:\Real^p\to\Real$ for $v\in V$ and $f_{\FF}:\Real^p\to\Real^q$ by
  \begin{align*}
    f_{\FF,v}(\vec x)&\coloneqq f_{\FA,v}(\vec x,\vec w,\vec b),\\
    f_{\FF}(\vec x)&\coloneqq f_{\FA}(\vec x,\vec w,\vec b).
  \end{align*}
As for circuits, to simplify the notation we usually denote the
functions $f_{\FA}$ and $f_{\FF}$ by $\FA$ and  $\FF$, respectively.

\begin{myremark}
  The reader may have noticed that we never use the activation
  function $\Fa_v$ or the bias $b_v$ for input nodes $v=X_i$. We only
  introduce them for notational convenience. \emph{We may always assume that
    $\Fa_v\equiv0$ and $b_v=0$ for all input nodes $v$.}
  \uend
\end{myremark}

Typically, the weights $w_e$ and biases $b_e$ are learned from data. We
are not concerned with the learning process here, but only with the
functions computed by pre-trained models. 

\emph{Throughout this article, we assume the activation functions
  in neural networks to be Lipschitz
  continuous.} Our theorems can also be proved with weaker assumptions
on the activation functions, but assuming Lipschitz continuity
simplifies the proofs, and since all activation functions typically
used in practice are Lipschitz continuous, there is no harm in making
this assumption. Since linear functions are Lipschitz continuous and
the composition of Lipschitz continuous functions is Lipschitz
continuous as well, it follows that for all FNNs $\FF$ the function
$f_\FF$ is Lipschitz continuous. A consequence of the Lipschitz
continuity is that the output of an FNN can be linearly bounded in the
input. For later reference, we state these facts as a lemma.

\begin{lemma}\label{lem:fnngrowth}
  Let $\FF$ be an FNN of input dimension $p$.
  \begin{enumerate}
  \item There is a Lipschitz constant $\lambda=\lambda(\FF)\in\PNat$ for
    $\FF$ such that for all
    $\vec x,\vec x'\in\Real^p$,
    \[
      \inorm{\FF(\vec x)-\FF(\vec x')}\le\lambda\inorm{\vec
        x-\vec x'}.
    \]
  \item There is a $\gamma=\gamma(\FF)\in\PNat$ such that for all
    $\vec x\in\Real^p$,
    \[
      \inorm{\FF(\vec x)}\le\gamma\cdot\big(\inorm{\vec
        x}+1\big).
    \]
  \end{enumerate}
\end{lemma}
\begin{proof}
  Assertion (1) is simply a consequence of the fact that composition of Lipschitz continuous functions is Lipschitz
  continuous. For (2), note that by (1) we have
  \[
    \inorm{\FF(\vec x)}\le\lambda(\FF)\inorm{\vec
      x}+\inorm{\FF(\vec 0)}.
  \]
  We let
  $\gamma\coloneqq\max\big\{\lambda(\FF),\inorm{\FF(\vec 0)}\big\}$.
\end{proof}

\begin{myremark}
  The reader may wonder why we take the constants $\lambda,\gamma$  in
  Lemma~\ref{lem:fnngrowth} to be integers. The reason is that we can
  easily represent positive integers by closed terms $(\one+\ldots+\one)$
  in the logic $\FOC$, and this will be convenient later.
  \uend
\end{myremark}

We often make further restrictions on the FNNs we consider.
An FNN is
\emph{piecewise linear} if all its activation functions are piecewise
linear. An FNN is \emph{rational piecewise linear} if all weights and
biases are dyadic rationals and all activation functions are rational
piecewise linear. The relu function and the linearised sigmoid
function (see Example~\ref{exa:pl-activation}) are typical examples of
rational piecewise linear activation functions.  An FNN is \emph{rpl
  approximable} if all its activation functions are rpl
approximable. The logistic function and the hyperbolic tangent function
(see Example~\ref{exa:rpl-app}) are typical examples of rpl
approximable activation functions. 

It is a well known fact that FNNs can simulate threshold
circuits.

\begin{lemma}\label{lem:c2fnn}
  For every threshold circuit
$\FC$ of input dimension $p$ there is an FNN $\FF=(V,E,(\Fa_v)_{v\in V},$ $(w_e)_{e\in  E},(b_v)_{v\in V})$ of input dimension $p$ such that
$|\FF|=O(|\FC|)$, $\Fa_v=\relu$
for all $v$, $w_e\in\{1,-1\}$ for all $e$, $b_v\in\Nat$ is bounded by the maximum
threshold in $\FC$ for all $v$, and $\FC(\vec x)=\FF(\vec x)$ for all $\vec x\in\{0,1\}^p$.
\end{lemma}

\begin{proof}
  We we simulate $\FC$ gatewise, noting that a Boolean $\neg x$
  negation can be expressed as $\relu(1-x)$ and a threshold
  $\sum x_i\ge t$ can be expressed as
  $\relu(\sum x_i-t+1)-\relu(\sum x_i-t)$ for Boolean inputs $x,x_i$.
\end{proof}

\subsection{Relational Structures}
\label{sec:graph}
A \emph{vocabulary} is a finite set $\tau$ of
relation symbols. Each relation symbol $R\in\tau$ has an \emph{arity}
$\ar(R)\in\Nat$.
A \emph{$\tau$-structure} $A$ consists of a finite set
$V(A)$, the \emph{universe} or \emph{vertex set}, and a relation $R(A)\subseteq V(A)^k$ for every relation symbol
$R\in\tau$ of arity $\ar(R)=k$. For a $\tau$-structure $A$ and a subset $\tau'\subseteq\tau$, the
\emph{restriction of $A$ to $\tau'$} is the $\tau'$-structure
$A|_{\tau'}$ with $V(A|_{\tau'})\coloneqq V(A)$ and
$R(A|_{\tau'})\coloneqq R(A)$ for all $R\in\tau'$.
The \emph{order} of a structure $A$ is $|A|\coloneqq |V(A)|$. 

For example, a graph may be viewed as an $\{E\}$-structure $G$, where
$E$ is a binary symbol, such that $E(G)$ is symmetric and
irreflexive. A pair $(G,\Cb)\in\CGS^\bool_\ell$, that is, a graph with a
Boolean signal $\Cb:V(G)\to\{0,1\}^\ell$, may be viewed as an
$\{E,P_1,\ldots,P_\ell\}$-structure $G_{\Cb}$  with
    $V(G_{\Cb})=V(G)$, $E(G_{\Cb})=E(G)$, and $P_i(G_{\Cb})=\{v\in
    V(G)\mid\Cb(v)_i=1\}$. We may think of
the $P_i$ as labels and hence refer to
$\{E,P_1,\ldots,P_\ell\}$-structures whose $\{E\}$-restrictions are
undirected graphs as
\emph{$\ell$-labeled graphs}. In the
following, we do not distinguish between graphs with Boolean
signals and the corresponding labeled graphs. 

A \emph{$k$-ary query} on a class $\CC$ of structures is an
equivariant mapping ${\CQ}$ that associates with each structure
$A\in\CC$ a mapping ${\CQ}(A):V(A)^k\to\{0,1\}$. In this article, we are mainly interested in $0$-ary (or \emph{Boolean}) and unary
queries on (labelled) graphs. 
We observe that a Boolean query on $\ell$-labelled
graphs is an invariant
mapping from $\CGS_\ell^\bool$ to $\{0,1\}$ and a unary query is an
equivariant signal transformations from $\CGS_\ell^\bool$ to
$\CGS_1^\bool$.

\section{First-Order Logic with Counting}
\label{sec:logic}
Throughout this section, we fix a vocabulary $\tau$.
We introduce two types of variables, \emph{vertex variables} ranging
over the vertex set of a structure, and \emph{number variables}
ranging over $\Nat$. We typically denote vertex variables by $x$ and
variants such as $x',x_1$, number variables by $y$ and variants, and we
use $z$ and variants to refer to either vertex or number variables.

We define the sets of
\emph{$\FOC$-formulas} and \emph{$\FOC$-terms} of vocabulary $\tau$
inductively as follows:
\begin{itemize}
\item All number variables and $\zero,\one$ are $\FOC$-terms.

\item For all  $\FOC$-terms $\theta,\theta'$ the expressions
  $\theta+ \theta'$ and $\theta\cdot \theta'$
  are $\FOC$-terms.

\item For all  $\FOC$-terms $\theta,\theta'$ the expression
  $
   \theta\le\theta'
  $
  is an $\FOC$-formula.

\item For all vertex variables $x_1,\ldots,x_k$ and all $k$-ary $R\in
  \tau$ the expressions
  $
  x_1=x_2
  $
  and $R(x_1,\ldots,x_k)$
  are $\FOC$-formulas.
\item For all $\FOC$-formulas $\phi,\psi$ the expressions
  $\neg\phi$ and $\phi\wedge\psi$
  are $\FOC$-formulas.
\item For all $\FOC$-formulas $\phi$, all $k,\ell\in\Nat$ with
  $k+\ell\ge 1$, all vertex variables
  $x_1,\ldots,x_k$, all number variables $y_1,\ldots,y_\ell$, and all
  $\FOC$-terms $\theta_1,\ldots,\theta_\ell$,
  \[
    \#(x_1,\ldots,x_k,y_1<\theta_1,\ldots,y_\ell<\theta_\ell).\phi 
  \]
  is an $\FOC$-term (a \emph{counting term}).
\end{itemize}
A \emph{$\tau$-interpretation} is a pair $(A,\mathcal a)$, where $A$
is a $\tau$-structure and $\Ca$ is an \emph{assignment over $A$}, that
is, a mapping from the set of all variables to $V(A)\cup\Nat$ such
that $\Ca(x)\in V(A)$ for every vertex variable $x$ and
$\Ca(y)\in\Nat$ for every number variable $y$. For a tuple $\vec z=(z_1,\ldots,z_k)$ of distinct variables, and a
tuple $\vec c=(c_1,\ldots,c_k)\in (V(A)\cup\Nat)^k$ such that $c_i\in
V(A)$ if $z_i$ is a vertex variable and $c_i\in\Nat$ if $z_i$ is a
number variable, we let $\Ca\frac{\vec c}{\vec z}$ be the
assignment with $\Ca\frac{\vec c}{\vec z}(z_i)=c_i$ and
$\Ca\frac{\vec c}{\vec z}(z)=\Ca(z)$ for all $z\not\in\{z_1,\ldots,z_k\}$.
We inductively define a
value $\sem{ \theta}^{(A,\Ca)}\in\Nat$ for
each $\FOC$-term $\theta$ and a Boolean value
$\sem{ \phi}^{(A,\Ca)}\in\{0,1\}$ for each $\FOC$-formula
$\phi$.
\begin{itemize}
\item We let $\sem{
  y}^{(A,\Ca)}\coloneqq\Ca(y)$ and $\sem{\zero}^{(A,\Ca)}\coloneqq 0$,
$\sem{\one}^{(A,\Ca)}\coloneqq 1$.
\item We let $\sem{
  \theta+\theta'}^{(A,\Ca)}\coloneqq \sem{
  \theta}^{(A,\Ca)}+\sem{ \theta'}^{(A,\Ca)}$
  and $\sem{
  \theta\cdot\theta'}^{(A,\Ca)}\coloneqq \sem{
  \theta}^{(A,\Ca)}\cdot\sem{
  \theta'}^{(A,\Ca)}$. 
\item We let
  $\sem{ \theta\le\theta'}^{(A,\Ca)}=1$ if and
  only if
  $\sem{ \theta}^{(A,\Ca)}\le\sem{
  \theta'}^{(A,\Ca)}$.
\item We let $\sem{
  x_1=x_2}^{(A,\Ca)}=1$ if and only if
  $\Ca(x_1)=\Ca(x_2)$ and $\sem{
  R(x_1,\ldots,x_k)}^{(A,\Ca)}=1$ if and only if $\big(\Ca(x_1),\ldots,\Ca(x_k)\big)\in R(A)$.
\item We let $\sem{\neg
  \phi}^{(A,a)}\coloneqq 1-\sem{\phi}^{(A,a)}$ and $\sem{\phi\wedge\psi}^{(A,a)}\coloneqq \sem{
  \phi}^{(A,a)}\cdot \sem{
  \phi}^{(A,a)}$.
\item We let
  \[
    \sem{
    \#(x_1,\ldots,x_k,y_1<\theta_1,\ldots,y_\ell<\theta_\ell).\phi}^{(A,\Ca)}
  \]
  be the number of tuples $(a_1,\ldots,a_k,b_1,\ldots,b_\ell)\in
  V(A)^k\times\Nat^\ell$ such that
  \begin{itemize}
  \item $b_i< \sem{
      \theta_i}^{(A,\Ca\frac{(a_1,\ldots,a_k,b_1,\ldots,b_{i-1})}{(x_1,\ldots,x_k,y_1,\ldots,y_{i-1})}
                      )}$ for all
    $i\in[\ell]$;
  \item $\sem{ \phi}^{(A,\Ca\frac{(a_1,\ldots,a_k,b_1,\ldots,b_{\ell})}{(x_1,\ldots,x_k,y_1,\ldots,y_{\ell})})}=1$.
  \end{itemize}
  Note that we allow \emph{adaptive bounds}: the bound $\theta_i$ for the variable
  $y_i$ may depend on the values of all previous variables $x_j$ for
  $j\in[k]$ and $y_j$ for $j<i$. While it
  can be shown that this does not increase the expressive power of the
  plain logic,\footnote{This is a consequence of
    Lemma~\ref{lem:termbound}, by which we can replace all bounds
    $\theta$ in counting terms by $\ord^k$ for a suitable $k$, where
    $\ord\coloneqq\# x. x=x$ is a term defining the order of the input
    structure.} the adaptive bounds do add power to an
  extension of the logic with function variables (see
  Section~\ref{sec:arithmetic}).
\end{itemize}
For $\FOC$-formulas $\phi$, instead of $\sem{
\phi}^{(A,\Ca)}=1$ we also write $(A,\Ca)\models\phi$.

An \emph{$\FOC$-expression} is either an $\FOC$-term or an
$\FOC$-formula. The set $\free(\xi)$ of \emph{free variables} of an
$\FOC$-expression $\xi$
is defined inductively in the obvious way, where for a counting term
we let
\begin{align*}
  &\free\big(\#(x_1,\ldots,x_k,y_1<\theta_1,\ldots,y_\ell<\theta_\ell).\phi\big)\coloneqq\\
  &\hspace{2cm}\big(\free(\phi)\setminus\{x_1,\ldots,x_k,y_1,\ldots,y_\ell\}\big)\cup\bigcup_{i=1}^\ell\free\big(\theta_i\setminus\{x_1,\ldots,x_k,y_1,\ldots,y_{i-1}\}\big).
\end{align*}
A \emph{closed expression} is an expression without free
variables. Depending on the type of expression, we also speak of
\emph{closed terms} and \emph{closed formulas}.

For an expression $\xi$, the notation $\xi(z_1,\ldots,z_k)$ stipulates that
$\free(\xi)\subseteq\{z_1,\ldots,z_k\}$. It is easy to see that the
value $\sem{\xi}^{(A,\Ca)}$ only depends on the
interpretations $c_i\coloneqq\Ca(z_i)$ of the free variables. Thus we
may avoid explicit reference to the assignment $\Ca$ and write
$\sem{\xi}^A(c_1,\ldots,c_k)$ 
instead of $\sem{\xi}^{(A,\Ca)}$. If $\xi$ is a
closed expression, we just write $\sem{\xi}^A$. For formulas
$\phi(z_1,\ldots,z_k)$,
we also write $A\models\phi(c_1,\ldots,c_k)$ instead of
$\sem{\phi}^A(c_1,\ldots,c_k)=1$, and for closed formulas $\phi$ we
write $A\models\phi$.

Observe that every $\FOC$-formula $\phi(x_1,\ldots,x_k)$
of vocabulary $\tau$ defines a $k$-ary query on
the class of $\tau$-structures, mapping a structure $A$ to the set of
all $(a_1,\ldots,a_k)\in A^k$ such that
$A\models\phi(a_1,\ldots,a_k)$. 

We defined the logic $\FOC$ with a minimal syntax, avoiding
unnecessary operators. However, we can use other standard arithmetical
and logical operators as abbreviations:
\begin{itemize}
\item For $n\ge 2$, we can use $n$ as an abbreviation for the
  corresponding sum of $1$s.
\item We use $\ord$ as an abbreviation for the term $\# x.x=x$. Then
  $\sem{\ord}^A=|A|$ for all structures $A$, that is, $\ord$ defines
  the order of a structure.
\item We can express the relations $=,\ge,<,>$ on $\Nat$ using Boolean
  combinations and $\le$.
\item We can express Boolean connectives like $\vee$ or $\to$ using
  $\neg$ and $\wedge$.
\item For vertex variables $x$, we can express existential
  quantification $\exists x.\phi$ as $1\le\# x.\phi$. Then we can
  express universal quantification $\forall x.\phi$ using $\exists x$
  and $\neg$ in the usual way.

  In particular, this means that we can view first-order logic $\FO$
  as a fragment of $\FOC$.
\item For number variables $y$, we can similarly express bounded
  quantification $\exists y<\theta.\phi$ and
  $\forall y<\theta.\phi$.
  \item In counting terms, we do not have to use strict inequalities to
    bound number variables. For example, we write
    $\#(x,y\le\theta).\phi$ to abbreviate $\#(x,y<\theta+1).\phi$.
\item We can express truncated subtraction $\dotminus$, 
  minimum and maximum of two numbers, and integer division:
  \begin{align*}
    y\dotminus y'&\text{ abbreviates }\# (y''< y).(y'\le y''),\\
    \logic{min}(y,y')&\text{ abbreviates }\#(y''< y).y''<y',\\
    \logic{max}(y,y')&\text{ abbreviates }\#(y''< y+y').(y''<y\vee y''<y'),\\
    \logic{div}(y,y')&\text{ abbreviates }\#(y''\le
                       y).(0<y'\wedge 0<y''\wedge y'\cdot y''\le y).
  \end{align*}
  Note that for all structures $A$ and $b,b'\in\Nat$ we have
  \[
    \sem{\logic{div}}^A(b,b')=
    \begin{cases}
      \floor{\frac{b}{b'}}&\text{if }b'\neq 0,\\
      0&\text{if }b'=0.
    \end{cases}
  \]
\end{itemize}

\begin{myremark}
  There are quite a few different versions of first-order logic with
  counting in the literature. The logics that only involve counting
  quantifiers $\exists^{\ge n}$ for constant $n$ are strictly weaker
  than our $\FOC$, and so are logics with modular counting
  quantifiers.

  Counting logics with quantification over numbers have first been
  suggested, quite informally, by Immerman \cite{Immerman87}. The
  2-sorted framework was later formalised by Grädel and
  Otto~\cite{GradelO93}. Essentially, our $\FOC$ corresponds to Kuske
  and Schweikardt's \cite{KuskeS17} $\logic{FOCN}(\{\mathbb P_\le\})$,
  with one important difference: we allow counting terms also over
  number variables. This makes no difference over ordered structures,
  but it makes the logic stronger over unordered structures (at least
  we conjecture that it does; this is
  something that with current techniques one can probably only prove modulo some complexity
  theoretic assumptions). Other differences, such as that
   Kuske
  and Schweikardt use the integers for the numerical part, whereas we
  use the non-negative integers, are inessential. Importantly,
  the two logics and other first-order logics with counting, such as
  first-order logic with a majority quantifier, are
  equivalent over ordered arithmetic structures and thus all
  capture the complexity class uniform $\TC^0$, as will be discussed
  in the next section.
  \uend
\end{myremark}

A simple lemma that we will frequently use states that all
$\FOC$-terms are polynomially bounded. At this point, the reader may
safely ignore the reference to function variables in the assertion of
the lemma; we will only introduce them in
Section~\ref{sec:2nd-order}. We just mention them here to avoid
confusion in later applications of the lemma.

\begin{lemma}\label{lem:termbound}
  For every $\FOC$-term $\theta(x_1,\ldots,x_k,y_1,\ldots,y_k)$
  without function variables there
  is a polynomial $\pi(X,Y)$ such that for all structures $A$, all
  $a_1,\ldots,a_k\in V(A)$, and all $b_1,\ldots,b_\ell\in\Nat$ it
  holds that
  \[
    \sem{\theta}^A(a_1,\ldots,a_k,b_1,\ldots,b_\ell)\le
    \pi\Big(|A|,\max\big\{b_i\bigmid i\in[\ell]\big\}\Big).
  \]
\end{lemma}

\begin{proof}
  A straightforward induction on $\theta$.
\end{proof}

\subsection{Descriptive Complexity}\label{sec:dc}

We review some results relating the logic $\FOC$ to the complexity
class $\TC^0$. In the descriptive complexity theory of ``small'' complexity classes
(say, within \PTIME), we need to expand structures by a linear order
of the vertex set (and possibly additional arithmetical relations). We
introduce a distinguished binary relation symbol $\les$, which we
assume to be not contained in the usual vocabularies $\tau$. Note that
$\les$ is distinct from $\le$, which we use for the standard linear
order on $\Nat$. We denote the interpretation of $\les$ in a
structure $A$ by $\les^A$ instead of ${\les}(A)$, and we use the
symbol in infix notation.

An \emph{ordered $\tau$-structure} is a $\tau\cup\{\les\}$-structure $A$ where $\les^A$ is a linear order of
the vertex set $V(A)$.
It will be convenient to have the following notation for ordered
structures $A$. For $0\le i<n\coloneqq|A|$, we let $\pos{i}_A$ be the
$(i+1)$st element of the linear order $\les^A$, that is, we have
$V(A)=\{\pos{i}_A\mid 0\le i<n\}$ with $\pos{0}_A\les^A
\pos{1}_A\les^A\cdots\les^A \pos{n-1}_A$. We omit the subscript ${}_A$ if $A$
is clear from the context.
The reason that ordered structures are important in descriptive
complexity is that they have simple canonical representations as
bitstrings. To represent ordered graphs, we can take the adjacency
matrix with rows and columns arranged according to the given order and
then concatenate the rows of the matrix to obtain a string
representation. This can easily be generalised to arbitrary structures
(see, for example, \cite{Immerman99}).
Let $s(A)\in\{0,1\}^*$ denote the string representing an ordered
structure $A$. Then for every
class $\CC$ of ordered $\tau$-structures, we let
$L(\CC)\coloneqq\{s(A)\mid A\in\CC\}$.

\begin{theorem}[Barrington, Immerman, and
  Straubing~\cite{BarringtonIS90}]
  \label{theo:bis}
  Let $\CC$ be a class of ordered $\tau$-structures. Then $L(\CC)$ is
  in uniform $\TC^0$ if and only if there is a closed $\FOC$-formula $\psi$
  of vocabulary $\tau\cup\{\les\}$ such that for all ordered
  $\tau$-structures $A$ it holds that
    $
      A\in\CC\iff A\models\psi.
   $
 \end{theorem}

 We need to rephrase this theorem for queries over unordered
 structures. For a class $\CC$ of
 $\tau$-structures, we let $\CC_\les$ be the class of all ordered
 $\tau$ structures $A$ with $A|_\tau\in\CC$, and we let
 $L_\les(\CC)\coloneqq L(\CC_\les)$. Since Boolean queries can be
 identified with classes of structures, this gives an encoding of
 Boolean queries by languages. Extending this to  queries of higher arity
 for a $k$-ary query $\CQ$ on a class $\CC$ of $\tau$-structures, we let
 \[
   L_\les(\CQ)\coloneqq\big\{s(A)\#\bin(i_1)\#\ldots\#\bin(i_k)\bigmid
   A\in\CC_\les, {\CQ}(A|_\tau)(\angles{i_1},\ldots, \angles{i_k})=1\big\}
 \]
 We say that a formula $\phi(x_1,\ldots,x_k)$ of vocabulary $\tau\cup\{\les\}$ is \emph{order
   invariant} if for all ordered $\tau$-structures $A,A'$ with
 $A|\tau=A'|\tau$ and all $a_1,\ldots,a_k\in V(A)$ it holds that
 $A\models\phi(a_1,\ldots,a_k)\iff A'\models\phi(a_1,\ldots,a_k)$. We
 say that a $k$-ary query
 $\CQ$ on a class of $\tau$-structures is definable in
 \emph{order-invariant} $\FOC$ if there is an order invariant
 $\FOC$-formula $\phi(x_1,\ldots,x_k)$ of vocabulary $\tau\cup\{\les\}$ such
 that for all $A\in\CC_\les$ and all $a_1,\ldots,a_k\in V(A)$ it holds that $A\models\phi(a_1,\ldots,a_k)\iff \CQ(A|_\tau)(a_1,\ldots,a_k)=1$.

\begin{corollary}
   Let $\CQ$ be a query. Then $L_{\les}(\CQ)$ is in uniform $\TC^0$ if
   and only if $\CQ$ is definable in order-invariant $\FOC$.
\end{corollary}

\subsection{Non-Uniformity and Built-in Relations}
To capture non-uniformity in descriptive complexity, we add
\emph{built-in relations}. The classical way of doing this is to only consider structures with
universe $\{0,\ldots,n-1\}$, for some $n\in\Nat$, and then add
relation symbols $S$ to the language that have a fixed interpretation
$S^{(n)}\subseteq \{0,\ldots,n-1\}^k$ in all structures with universe
$\{0,\ldots,n-1\}$. Slightly more abstractly, we can consider ordered
structures and transfer the definition of $S^{(n)}$ to all linearly
ordered structures of order $n$ via the natural mapping
$i\mapsto\angles i$.

We take a different approach to built-in relations here, which
allows us to also use them over structures that are not necessarily
ordered. A \emph{built-in numerical relation} is simply a relation over $\Nat$,
that is, a subset $N\subseteq\Nat^k$ for some $k\ge0$, the
\emph{arity} of $N$. We use the same letter $N$ to denote both the
relation $N\subseteq\Nat^k$ and a $k$-ary relation symbol representing
it in the logic. In other words, the relation symbol $N$ will be
interpreted by the same relation $N\subseteq\Nat^k$ in all
structures. We extend the logic $\FOC$ by new atomic formulas
$N(y_1,\ldots,y_k)$ for all $k$-ary numerical relations $N$ and number
variables $y_1,\ldots,y_k$, with the obvious semantics. By
$\FOCnu$ we denote the extension of $\FOC$ to formulas using arbitrary
built-in numerical
relations.\footnote{Think of the index 'nu' as an abbreviation of
  either 'numerical' or 'non-uniform'.}
We will later use the same notation $\logic F\nuni$ for fragments
$\logic F$ of $\FOC$ to denote the 
extension of $\logic F$ to formulas using 
built-in numerical
relations. 

Then it easily follows from Theorem~\ref{theo:bis} that
$\FOCnu$ captures (non-uniform) $\TC^0$. (Or it can be proved
directly, in fact, it is much easier
to prove than Theorem~\ref{theo:bis}.)

\begin{corollary}\label{cor:nuTC}
 Let $\CC$ be a class of ordered $\tau$-structures. Then $L(\CC)$ is
  in $\TC^0$ if and only if there is a closed $\FOCnu$-formula $\psi$
  of vocabulary $\tau\cup\{\les\}$ such that for all ordered
  $\tau$-structures $A$ it holds that
    $
      A\in\CC\iff A\models\psi.
   $ 
 \end{corollary}

We also state the version of
this result for queries.

 \begin{corollary}\label{cor:TC0FO2}
   Let $\CQ$ be a query. Then $L_{\les}(\CQ)$ is in $\TC^0$ if
   and only if $\CQ$ is definable in order-invariant $\FOCnu$.
\end{corollary}

\subsection{Types and Second-Order Variables}
\label{sec:2nd-order}

The counting extension of first-order logic refers to a 2-sorted
extension of relational structures and adheres to a strict type
discipline. For the extension we are going to introduce next, we need
to make this formal. We assign a \emph{type} to each variable: a
vertex variable has type $\ttv$, and number variable has type
$\ttn$. A $k$-tuple $(z_1,\ldots,z_k)$ of variables has a type
$(t_1,\ldots,t_k)\in\{\ttv,\ttn\}^k$, where $t_i$ is the type of
$z_i$. We denote the type of a tuple $\vec z$ by $\tp(\vec z)$. For a structure $A$ and a type
$\vec t=(t_1,\ldots,t_k)\in\{\ttv,\ttn\}^k$ we let $A^{\vec t}$ be the
set of all $(c_1,\ldots,c_k)\in (V(A)\cup\Nat)^k$ such that
$c_i\in V(A)$ if $t_i=\ttv$ and $c_i\in\Nat$ if $t_i=\ttn$.

Now we extend our logic by relation variables (denoted by uppercase
letters $X,Y$) and function variables (denoted by $U,V$). Each
relation variable $X$ has a type $\tp(X)$ of the form
$\tta\vec t\ttz$, and each function variable $U$ has a type $\tp(U)$
of the form $\vec t\to\ttn$, for some $\vec t\in\{\ttv,\ttn\}^k$. We
extend the logic $\FOC$ by allowing additional atomic formulas
$X(\xi_1,\ldots,\xi_k)$ and terms $U(\xi,\ldots,\xi_k)$, where $X$ is a
relation variable of type $\tta(t_1,\ldots,t_k)\ttz$ for some tuple
$(t_1,\ldots,t_k)\in\{\ttv,\ttn\}^k$, $U$ a function variable of type
$(t_1,\ldots,t_k)\to\ttn$, and for all $i\in[k]$, if $t_i=\ttv$ then
$\xi_i$ is a vertex variable and if $t_i=\ttn$ then $\xi_i$ is a term.

To define the semantics, let $A$
be a structure and $\Ca$ an assignment over $A$. Then $\Ca$
maps each relation variable $X$ of type $\tta\vec t\ttz$ to a subset $\Ca(X)\subseteq
A^{\vec t}$ and each function variable $U$ of type $\vec t\to\ttn$ to a function
$\Ca(U):A^{\vec t}\to\Nat$. Moreover, for a tuple $\vec\xi=(\xi_1,\ldots,\xi_k)$ of vertex
variables and terms we let $\sem{\vec\xi}^{(A,\Ca)}=(c_1,\ldots,c_k)$
where $c_i=\Ca(\xi_i)$ if $\xi_i$ is a vertex variables and
$c_i=\sem{\xi_i}^{(A,\Ca)}$ if $c_i$ is a term.
We let
\begin{align*}
  \sem{X(\vec \xi)}^{(A,\Ca)}&\coloneqq
  \begin{cases}
    1&\text{if }\sem{\vec\xi}^{(A,\Ca)}\in\Ca(X),\\
    0&\text{otherwise},
  \end{cases}
  \\
  \sem{U(\vec \xi)}^{(A,\Ca)}&\coloneqq \Ca(U)\big(\sem{\vec\xi}^{(A,\Ca)}\big). 
\end{align*}
Observe that a function variable $U$ of type $\emptytuple\to\ttn$ is essentially just a number
variable, if we identify a $0$-ary function with the value it takes on
the empty tuple $\emptytuple$. It is still useful sometimes to use
$0$-ary function variables. We usually write $\Ca(U)$ instead of
$\Ca(U)(\emptytuple)$ to denote their value. We call a relation
variable \emph{purely numerical} if it is of type $\tta\ttn^k\ttz$,
for some $k\ge 0$. Similarly, we call a function variable \emph{purely numerical} if it is of type $\ttn^k\to\ttn$.

To distinguish them from the
``second-order'' relation and function variables, we refer to our
original ``first-order'' vertex variables and number variables as
\emph{individual variables}.
When we list variables of an expression in parentheses, as in $\xi(\vec z)$,
we only list the free individual
variables, but not the free relation or function
variables. Thus $\xi(\vec z)$ stipulates that all free individual
variables of $\xi$ occur in $\vec z$. However, $\xi$ may have free
relation variables and free function variables that are not listed in
$\vec z$. For a
structure $A$, an assignment $\Ca$, and a tuple $\vec c\in A^{\tp(\vec
  z)}$, we write $\sem{\xi}^{(A,\Ca)}(\vec c)$ instead of
$\sem{\xi}^{(A,\Ca\frac{\vec c}{\vec z})}$. If $\xi$ is a formula, we
may also write
$(A,\Ca)\models\xi(\vec c)$. 

The role of relation variables and function variables is
twofold. First, we will use them to specify "inputs" for formulas, in
particular for formulas defining numerical functions. (In the next
sections we will see how to use relation and function variables to
specify natural and rational numbers.) And second, we may just use relation and
function variables as placeholders for formulas and terms that we may
later substitute for them.

Note that Lemma~\ref{lem:termbound} no longer holds if the term
$\theta$ contains function variables, because these variables may be
interpreted by functions of super-polynomial growth.

Let us close this section by emphasising that we do not allow quantification over
relation or function variables. Thus, even in the presence of such
variables, our logic remains ``first-order''.

\subsection{Arithmetic in \texorpdfstring{$\FOC$}{FO+C}}\label{sec:arithmetic}
In this section, we will show that arithmetic on bitwise
representations of integers is expressible in $\FOC$. Almost none of
the formulas we shall define make any reference to a structure $A$;
they receive their input in the form of purely numerical relation
variables and function variables and only refer to the numerical part,
which is the same for all structures.  We call an $\FOC$-expression
$\xi$ \emph{arithmetical} if it contains no vertex variables.  It is
worth mentioning that arithmetical $\FOC$-formulas without relation
and function variables are formulas of \emph{bounded arithmetic} (see,
for example, \cite{HajekP93}) augmented by bounded counting terms.

Note that if $\xi$ is an arithmetical expression,
then for all structures $A,A'$ and all assignments $\Ca,\Ca'$ over
$A,A'$, respectively, such that $\Ca(y)=\Ca'(y)$ for all number variables
$y$ and $\Ca(Z)=\Ca'(Z)$ for all purely numerical relation or function
variables $Z$, we have
$\sem{\xi}^{(A,\Ca)}=\sem{\xi}^{(A',\Ca')}$. Thus there is no need to
mention $A$ at all; we may
write $\sem{\xi}^{\Ca}$ instead of $\sem{\xi}^{(A,\Ca)}$. In fact, we
can even use the notation $\sem{\xi}^{\Ca}$ if $\Ca$ is only a
partial assignment that assigns values only to number variables and
purely numerical relation and function variables. We call such a partial assignment a \emph{numerical}
assignment. As usual, if $\xi=\xi(y_1,\ldots,y_k)$ has all free individual
variables among $y_1,\ldots,y_k$, for $b_1,\ldots,b_k\in\Nat$ we may
write
$\sem{\xi}^\Ca(b_1,\ldots,b_k)$ instead of
$\sem{\xi}^{\Ca\frac{b_1\ldots b_k}{y_1\ldots,y_k}}$. If, in addition,
  $\xi$ has no free relation or function variables, we may just write 
  $\sem{\xi}(b_1,\ldots,b_k)$.

  The following well-known lemma is the foundation for expressing
  bitwise arithmetic in $\FOC$. There is also a (significantly deeper) version of the lemma
  for first-order logic without counting, which goes back to
  Bennett~\cite{Bennett62} (see \cite[Section~1.2.1]{Immerman99} for a
  proof).\footnote{I find it worthwhile to give the simple proof of the lemma
  for $\FOC$; I am not aware that this can be found in the literature.}

\begin{lemma}\label{lem:bit}
  There is an arithmetical $\FOC$-formula $\logic{bit}(y,y')$ such that for all $i,n\in\Nat$,
  \[
    \sem{\logic{bit}}(i,n)=\Bit(i,n).
  \]
\end{lemma}

\begin{proof}
  Clearly, there is a formula $\logic{pow2}(y)$ expressing that $y$ is
  a power of $2$; it simply states that all divisors of $y$ are
  divisible by $2$.
  Then the formula
  \[
    \logic{exp2}(y,y')\coloneqq \logic{pow2}(y')\wedge y =\# y''<y'.\logic{pow2}(y'')
  \]
  expresses that $y'=2^y$.

  To express the bit predicate, we define the auxiliary formula
  \[
    \logic{pow2bit}(y,y')\coloneqq \logic{pow2}(y)\wedge\exists
    y_1<y'.\exists y_2<y.y'=2y_1y+y+y_2,
  \]
  expressing that $y$ is $2^i$ for some $i$ and the $i$th bit of $y'$
  is $1$. We let
  \[
    \logic{bit}(y,y')\coloneqq\exists
    y''<y'.\big(\logic{exp2}(y,y'')\wedge \logic{pow2bit}(y'',y')\big).
    \qedhere
  \]
\end{proof}

\begin{corollary}\label{cor:bit}
  There is an arithmetical $\FOC$-term $\logic{len}(y)$ such that
  for all $n\in\Nat$,
  \[
    \sem{\logic{len}}(n)=|\bin(n)|.
  \]
\end{corollary}

\begin{proof}
  Observe that for $n\ge 1$, 
  \[
    |\bin(n)|=1+\max\big\{i\bigmid\Bit(i,n)=1\big\}.
  \]
  Noting that $|\bin(n)|\le n$ for all $n\ge1$, the following
  term defines the length for all $n\ge 1$:
  \[
    1+\#z<y.\exists y'\le
    y.\Big(z<y'\wedge\logic{bit}(y',y)\wedge\forall y''\le y\big(
    y'<y''\to\neg\logic{bit}(y'',y)\big)\Big).
  \]
  Note that that this term also gives us the correct result for
  $n=0$, simply because the formula $\exists y'\le
    y.\Big(z<y'\wedge\ldots)$ is false for all $z$ if $y=0$.
\end{proof}

In the proof of the previous lemma we used a trick that is worthwhile
being made explicit. Suppose we have a formula $\phi(y,\vec z)$ that
defines a function $\vec z\mapsto y$, that is, for all structures
$A$ and $\vec c\in A^{\tp(\vec z)}$ there is a unique $b=f_A(\vec
c)\in\Nat$ such
that $A\models\phi(b,\vec c)$. Often, we want a term expressing
the same function. In general, there is no such term, because the
function may grow too fast. (Recall that all terms are polynomially
bounded by Lemma~\ref{lem:termbound}, but the function $f_A$ may grow
exponentially fast.)
Suppose, however, that we have a
term $\theta(\vec z)$ that yields an upper bound for this function,
that is, $f_A(\vec c)<\sem{\theta}^A(\vec c)$ for all $A$ and $\vec c\in A^{\tp(\vec z)}$. Then we 
obtain a term $\eta(\vec z)$ such that $\sem{\eta}^A(\vec c)=f_A(\vec
c)$ as follows: we let
\[
  \eta(\vec z)\coloneqq\# y'<\theta(\vec z).\exists y<\theta(\vec
  z).\big(y'<y\wedge\phi(y,\vec z)\big).
\]

It is our goal for the rest of this section to express bitwise
arithmetic in $\FOC$.  We will use relation variables to encode binary
representations of natural numbers. Let $Y$ be a relation variable of
type $\ttn$, and let $\Ca$ be a numerical interpretation. We think of
$Y$ as representing the number whose $i$th bit is $1$ if and only if
$i\in\Ca(Y)$. But as $\Ca(Y)$ may be infinite, this representation is
not yet well defined. We also need to specify a bound on the number of
bits we consider, which we can specify by a function variable
$U$ of type $\emptyset\to\ttn$. Then the pair $(Y,U)$ represents the number
\begin{equation}\label{eq:num1}
  \num{Y,U}^{\Ca}\coloneqq \sum_{\substack{i\in\Ca(Y),i<\Ca(U)}}2^i.
\end{equation}
We can also specify numbers by formulas and terms. We let $\hat y$ be a distinguished number variable
(that we fix for the rest of this article). Let $\chi$ be a formula and
$\theta$ a term. We usually assume that $\hat y$ occurs freely in
$\chi$ and does not occur in $\theta$, but neither is necessary. Let
$A$ be a structure and $\Ca$ an assignment over $A$. Recall that
$\Ca\frac{i}{\hat y}$ denotes the assignment that maps
$\hat y$ to $i$ and coincides with $\Ca$ on all other variables.
We let
\begin{equation}\label{eq:num2}
  \num{\chi,\theta}^{(A,\Ca)}\coloneqq \sum_{\substack{
      (A,\Ca\frac{i}{\hat y})\models \chi,\\i<\sem{\theta}^{(A,\Ca)}}}2^i.
\end{equation}
If $\chi$ and $\theta$ are arithmetical, we may write
$\num{\chi,\theta}^{\Ca}$ instead of $\num{\chi,\theta}^{(A,\Ca)}$.

The following Lemmas~\ref{lem:ar1}, \ref{lem:ar2}, and \ref{lem:ar2a} follow easily
from the facts that the arithmetic operations are in uniform $\TC^0$
(Lemma~\ref{lem:tc0ar}) and $\FOC$ captures uniform $\TC^0$
(Theorem~\ref{theo:bis}). However, we find it helpful to sketch at
least some of the proofs, in particular the proof of
Lemma~\ref{lem:ar2} for iterated addition. Researchers in the
circuit-complexity community are well-aware of the fact that
iterated addition is in \emph{uniform} $\TC^0$, and at least implicitly
this is shown in \cite{BarringtonIS90}. Nevertheless, I think it is
worthwhile to give a proof of this lemma. Our proof is purely logical,
circumventing circuit complexity altogether, so it may be of
independent interest.

\begin{lemma}\label{lem:ar1}
  Let $Y_1,Y_2$ be relation variables of type $\tta\ttn\ttz$, and let
  $U_1,U_2$ be function variables of type $\emptytuple\to\ttn$.
  \begin{enumerate}
  \item There are arithmetical $\FOC$-formulas $\logic{add}$,
    $\logic{sub}$ and arithmetical $\FOC$-terms
    $\logic{{bd-add}}$, $\logic{{bd-sub}}$ such
    that for all structures $A$ and assignments $\Ca$ over $A$,
    \begin{align*}
      \num{\logic{add}, \logic{{bd-add}}}^{\Ca}
      &=\num{Y_1,U_1}^{\Ca}+
        \num{Y_2,U_2}^{\Ca},\\
      \num{\logic{sub}, \logic{{bd-sub}}}^{\Ca}
      &=\num{Y_1,U_1}^{\Ca}\dotminus
        \num{Y_2,U_2}^{\Ca}.
    \end{align*}
  \item There is an arithmetical $\FOC$-formula $\logic{leq}$ such
    that for all structures $A$ and assignments $\Ca$ over $A$,
    \[
      \sem{\logic{leq}}^\Ca=1\iff \num{Y_1,U_1}^{\Ca}\leq
      \num{Y_2,U_2}^{\Ca}.
    \]
  \end{enumerate}
\end{lemma}

\begin{proof}
  The key observation is that we can easily define the
  carry bits. Suppose that we want to add numbers $m,n$. Then for
  $i\ge 0$, the $i$th carry is $1$ if any only if there is a $j\le i$
  such that $\Bit(j,m)=\Bit(j,n)=1$, and for $j<k\le i$, either $\Bit(k,m)=1$ or $\Bit(k,n)=1$.

  We can use a similar observation for subtraction.

  Less-than-or-equal-to can easily be expressed directly.
\end{proof}

To define families of numbers, we use relation and function
variables of higher arity, treating the additional entries as
parameters. 
For a type $\vec t\in\{\ttv,\ttn\}^k$, let $Y$ be a relation variable of type
$\tta\ttn\vec
t\ttz$, and let $U$ be a number variable of type
$\vec t\to\ttn$. Then for every structure $A$, assignment $\Ca$, and tuple
$\vec c\in A^{\vec t}$ we let 
\begin{equation}\label{eq:num4}
  \num{Y,U}^{(A,\Ca)}(\vec
  c)=\sum_{\substack{(j,\vec
      c)\in\Ca(Y),\\j<\Ca(U)(\vec c)}}2^j.
\end{equation}
We can slightly extend this definition to a setting where $U$ is a function variable of type
$\vec t'\to\ttn$ for some subtuple $\vec t'$ of $\vec t$. For example,
in the following lemma we have $\vec t=\ttn$ and $\vec
t'=\emptytuple$.

\begin{lemma}\label{lem:ar2}
  Let $Y$ be a relation variable of type $\tta(\ttn,\ttn)\ttz$, and let
  $U$ be a function variable of type 
  $\emptytuple\to\ttn$. 
  Then there is an arithmetical $\FOC$-formula $\logic{s-itadd}$
  and  an arithmetical $\FOC$-term
    $\logic{{bd-s-itadd}}$\footnote{The 's' in $\logic{s-itadd}$ indicates that
 this is a simple version of iterated addition.} such
    that for all numerical assignments $\Ca$ we have
    \[
      \num{\logic{s-itadd}, \logic{{bd-s-itadd}}}^{\Ca}=
      \sum_{i<\Ca(U)}\num{Y,U}^{\Ca}(i).
    \]
\end{lemma}

\newcommand{\code}[1]{\ullcorner#1\ulrcorner}

The proof of this lemma requires some preparation.
Our first step will be to fix an encoding of sequences by natural
numbers. We first encode a sequence $\vec
i=(i_0,\ldots,i_{k-1})\in\Nat^*$ by the string 
\[
\sigma_1(\vec
i)\coloneqq\#\bin(i_0)\#\bin(i_2)\#\ldots\#\bin(i_{k-1})
\]
over the alphabet
$\{0,1,\#\}$. Then we replace every $0$ in $\sigma_1(\vec i)$ by $01$,
every $1$ by $10$, and every $\#$ by $11$ to obtain a string $\sigma_2(\vec
i)=s_{\ell-1}\ldots s_0$ over the alphabet $\{0,1\}$. We read $\sigma_2(\vec
i)$ as a binary number and let
\[
  \code{\vec
  i}\coloneqq
\begin{cases}
\displaystyle\bin^{-1}\big(\sigma_2(\vec i)\big)=\sum_{j=0}^{\ell-1}s_{j}2^j&\text{if }\vec i\neq\emptyset,\\
  0&\text{if }\vec i=\emptyset.
\end{cases}
\]

\begin{example}
  Consider the sequence $\vec i=(2,5,0)$. We have $\sigma_1(\vec
  i)=\#10\#101\#0$ and $\sigma_2(\vec i)=111001111001101101$. This yields
  \[
    \code{\vec i}=237\,165. \qedhere 
  \]
\end{example}

It is easy to see that the mapping $\code{\cdot}:\Nat^*\to\Nat$ is
injective, but not bijective.
Observe that for $\vec
i=(i_0,\ldots,i_{k-1})\in\Nat^*$ we have
\begin{equation}
  \label{eq:95}
  \bsize\big(\code{\vec i}\big)=\sum_{j=0}^{k-1}2(\bsize(i_j)+1)\le 4\sum_{j=0}^{k-1}\bsize(i_j)
\end{equation}
and thus
\begin{equation}
  \label{eq:96}
  \code{\vec i}<2^{4\sum_{j=0}^{k-1}\bsize(i_j)}.
\end{equation}

\begin{lemma}\label{lem:itadd1}
  \begin{enumerate}
  \item There is an arithmetical $\FOC$-formula $\logic{seq}(y)$ such
    that for all $n\in\Nat$ we have
    \[
      \sem{\logic{seq}}(n)=1\iff n=\code{\vec i}\text{ for some
        }\vec i\in\Nat.
    \]
    
  \item There is an arithmetical $\FOC$-term $\logic{seqlen}(y)$ such
    that for all $n\in\Nat$ we have
    \[
      \sem{\logic{seqlen}}(n)=
      \begin{cases}
        k&\text{if }n=\code{(i_0,\ldots,i_{k-1})}\text{ for some
        }k,i_0,\ldots,i_{k-1}\in\Nat,\\
        0&\text{otherwise}.
      \end{cases}
    \]
  \item There is an arithmetical $\FOC$-term $\logic{entry}(y,y')$ such
    that for all $j,n\in\Nat$ we have
    \[
      \sem{\logic{entry}}(j,n)=
      \begin{cases}
        i&\text{if }n=\code{(i_0,\ldots,i_{k-1})},j<k,i=i_j\text{ for some
        }k,i_0,\ldots,i_{k-1}\in\Nat,\\
        n&\text{otherwise}.
      \end{cases}
    \]
  \end{enumerate}
\end{lemma}

\begin{proof}
  Let
  $
    S=\{\sigma_2(\vec i)\mid \vec i\in\Nat^*\}.
  $
  Let $\bin(n)\eqqcolon\vec s=s_{\ell-1}\ldots s_0$. We want to detect
  if $\vec s\in S$.
  As a special case, we note that if $\ell=0$ and thus $\vec s=\emptyset$,
  we have $\vec s=\sigma_2(\emptyset)$. In the following, we assume that
  $\ell>0$. Then if $\ell$ is odd, we have $\vec s\not\in S$.
  Furthermore, if there is a $p<\frac{\ell}{2}$ such that $s_{2p}=0$ and
  $s_{2p+1}=0$, again we have $\vec s\not\in S$,  because $\vec s$ is not obtained from
  a string $\vec s'\in\{0,1,\#\}^*$ by replacing $0$s by $01$, $1$s
  by $10$, and $\#$s by $11$. Otherwise, we let $\vec
  s'=s'_{\frac{\ell}{2}-1}\ldots s'_0\in\{0,1,\#\}^*$ be the corresponding
  string. Then $\vec s'= \sigma_1(\vec i)$ for some
  $\vec i\in\Nat^*$ if and only if $\vec s'$ satisfies the following conditions:
  \begin{itemize}
  \item $s'_0\neq\#$ and $s'_{\ell/2-1}=\#$;
  \item for all $p<\frac{\ell}{2}$, if $s'_p=\#$ and $s'_{p-1}=0$ then either
    $p=1$ or $s_{p-2}=\#$.
  \end{itemize}
  We can easily translate the conditions to conditions on the string
  $\vec s$ and, using the bit predicate, to
  conditions on $n$. As the bit predicate is definable in $\FOC$ (by
  Lemma~\ref{lem:bit}), we can express these conditions by an
  arithmetical $\FOC$-formula $\logic{seq}(y)$.
  
  \medskip To prove (2), we observe that if $n=\code{\vec i}$ for some
  $\vec i\in\Nat^*$, then the length of the sequence
  $\vec i$ is the number of $\#$s in the string
  $\sigma_1(\vec i)$, or equivalently, the
  number of $p<\frac{\ell}{2}$ such that $\Bit(2p,n)=1$ and
$\Bit(2p+1,n)=1$ Using the formula $\logic{seq}(y)$ and the bit
predicate, we can easily express this by a term $\logic{seqlen}(y)$.

  \medskip
  To prove (3), we first write an arithmetical formula
  $\logic{isEntry}(y,y',y'')$ such that
  \[
    \sem{\logic{isEntry}}(j,n,i)=1\iff n=\code{(i_0,\ldots,i_{k-1})},j<k,i=i_j\text{ for some
    }k,i_0,\ldots,i_{k-1}\in\Nat.
  \]
  Once we have this formula, we let
  \[
    \logic{entry}(y,y')=\logic{min}\Big(y',\#z<y'.\exists
    y''<y'.\big(\logic{isEntry}(y,y',y'')\wedge z<y''\big)\Big).
  \]
  To define $\logic{isEntry}(y,y',y'')$, observe that for $\vec
  i=(i_0,\ldots,i_{k-1})\in\Nat^*$, the $j$th entry $i_j$ is located
  between the $j$th and $(j+1)$st '$\#$' in the string $\sigma_1(\vec i)$
  and thus between the $j$th and $(j+1)$st occurrence of '$11$' at
  positions $2p,2+1$ in the string $\sigma_2(\vec i)=\bin(\code{\vec
    i})$. Using the bit predicate, we can thus extract the bit
  representation of $i_j$ from $\code{\vec i}$ in $\FOC$.
\end{proof}

\begin{lemma}\label{lem:itadd2}
  There are arithmetical $\FOC$-terms $\logic{flog}(y)$ and
  $\logic{clog}(y)$ such that for all for all $n\in\PNat$,
  \[
    \sem{\logic{flog}}(n)=\floor{\log n}
    \quad\text{and}\quad
    \sem{\logic{clog}}(n)=\ceil{\log n}.
  \]
\end{lemma}

\begin{proof}
  This is straightforward, observing that $\floor{\log n}$ is the
  highest $1$-bit in the binary representation of $n$.
\end{proof}

\begin{proof}[Proof of Lemma~\ref{lem:ar2}]
  For the proof, it will be convenient to fix a numerical assignment
  $\Ca$. Of course the $\FOC$-expressions we shall define will not
  depend on $\Ca$. Let $m\coloneqq\Ca(U)$, and for $0\le i<m$, let
  $n_i\coloneqq\num{Y,U}^{\Ca}(i)$. Then $n_i<2^m$ and, and for $0\le
  j<m$ we have
  \[
   \Bit(j,n_i)=1\iff (j,i)\in\Ca(Y).
  \]
  It is our goal to compute $\sum_{i=0}^{m-1} n_i$ in such a way that
  the computation can be expressed in $\FOC$. More precisely, we want
  to define a formula $\logic{s-itadd}$
  and  a term
  $\logic{{bd-s-itadd}}$, which both may use the variables $Y$ and
  $U$, such that
    \[
      \num{\logic{s-itadd}, \logic{{bd-s-itadd}}}^{\Ca}=\sum_{i=0}^{m-1} n_i.
    \]
  Since
  $
    \sum_{i=0}^{m-1}
    n_i<2^{2m},
  $
  we can simply let $\logic{bd-s-itadd}\coloneqq 2U$. Thus we only
    need to define the formula $\logic{s-itadd}(\hat y)$ in such a way
    that for all $j<2m$ we have
  \begin{equation}
    \label{eq:99}
      \sem{\logic{s-itadd}}^\Ca(j)=\Bit\left(j, \sum_{i=0}^{m-1}
        n_i\right).
    \end{equation}
    Without of loss of generality we may assume that $m$ is
    sufficiently large, larger than some absolute constant that can be
    extracted from the proof. If $m$ is smaller than this constant,
    we simply compute the sum by repeatedly applying
    Lemma~\ref{lem:ar1}(1).

Our construction will be inductive, repeatedly transforming the
initial sequence of numbers $n_i$ into new sequences that have the
same sum. It will be convenient to use the index ${}^{(0)}$ for the
initial family. Thus we let 
$m^{(0)}\coloneqq m$, and $n^{(0)}_i\coloneqq n_i$,
for all
$i\in\{0,\ldots,m^{(0)}-1\}$. Furthermore, it will be useful to
let $\ell^{(0)}\coloneqq m^{(0)}$, because at later stages $t$ of the
construction the size $m^{(t)}$ of the current family of numbers will
no longer be identical with their bitlength $\ell^{(t)}$. 

For $j<\ell^{(0)}$, let
$n^{(0)}_{i,j}\coloneqq \Bit(j,n^{(0)}_{i})$ and
\begin{equation}
    \label{eq:93}
  s^{(0)}_j\coloneqq\sum_{i=0}^{m^{(0)}-1}n^{(0)}_{i,j}=\big|\big\{i\bigmid
  n^{(0)}_{i,j}=1\big\}\big|.
\end{equation}
Then
\[
  \sum_{i=0}^{m^{(0)}-1}n^{(0)}_i=\sum_{i=0}^{m^{(0)}-1}\sum_{j=0}^{\ell^{(0)}-1} 2^j n^{(0)}_{i,j}
  =\sum_{j=0}^{\ell^{(0)}-1} 2^j\sum_{i=0}^{m^{(0)}-1} n^{(0)}_{i,j}
  =\sum_{j=0}^{\ell^{(0)}-1} 2^j s^{(0)}_j.
\]
Let $m^{(1)}\coloneqq\ell^{(0)}$ and $n^{(1)}_i\coloneqq 2^is^{(0)}_i$
for $i<\ell^{(0)}$. Then
\begin{equation*}
  \sum_{i=0}^{m-1}n_i=\sum_{i=0}^{m^{(0)}-1}n_i^{(0)}=\sum_{i=0}^{m^{(1)}-1}n_i^{(1)}. 
\end{equation*}
Moreover,
\[
  n^{(1)}_i\le
  2^{\ell^{(0)}-1}s^{(0)}_i=2^{\ell^{(0)}-1+\log
    s^{(0)}_i}<2^{\ell^{(0)}+\floor{\log s^{(0)}_i}}.
\]
Let $p^{(1)}\coloneqq \floor{\log m^{(0)}}$ and $\ell^{(1)}\coloneqq
\ell^{(0)}+p^{(1)}$. Noting that $s^{(0)}_i\le m^{(0)}$
for all $i$,
we thus have
\[
  n^{(1)}_i<2^{\ell^{(1)}}
\]
for all $i$. Finally, note that
\[
  \Bit(j,n^{(1)}_i)\neq0\implies i\le j\le i+p^{(1)}
\]
This completes the base step of the construction. For the inductive
step, suppose that we have defined $m^{(k)},\ell^{(k)},p^{(k)}\in\PNat$ and
$n^{(k)}_i\in\Nat$ for $i\in\{0,\ldots,m^{(k)}-1\}$ such that
\[
  \sum_{i=0}^{m-1}n_i=\sum_{i=0}^{m^{(k)}-1}n_i^{(k)}
\]
and
for all $i$:
\begin{itemize}
\item $n^{(k)}_i<2^{\ell^{(k)}}$ for all $i$;
\item $\Bit(j,n^{(k)}_i)\neq0\implies i\le j\le i+p^{(k)}$.
\end{itemize}
For $j<\ell^{(k)}$, let
$n^{(k)}_{i,j}\coloneqq \Bit(j,n^{(k)}_{i})$ and
\[
  s^{(k)}_j\coloneqq\sum_{i=0}^{m^{(k)}-1}n^{(k)}_{i,j}=\sum_{i=\max\{0,j-p^{(k)}\}}^j n^{(k)}_{i,j}.
\]
Then
\[
  \sum_{i=0}^{m^{(k)}-1}n^{(k)}_i
=\sum_{j=0}^{\ell^{(k)}-1} 2^j s^{(k)}_j.
\]
Let $m^{(k+1)}\coloneqq\ell^{(k)}$ and $n^{(k+1)}_i\coloneqq
2^is^{(k)}_i$ for $i<\ell^{(k)}$. Then
\[
  \sum_{i=0}^{m-1}n_i=\sum_{i=0}^{m^{(k)}-1}n_i^{(k)}=\sum_{i=0}^{m^{(k+1)}-1}n_i^{(k+1)}.
\]
Moreover,
\[
  n^{(k+1)}_i\le
  2^{\ell^{(k)}-1}s^{(k)}_i
  =2^{\ell^{(k)}-1+\log s^{(k)}_i}
  <2^{\ell^{(k)}+\floor{\log s^{(k)}_i}}.
\]
Note that $s^{(k)}_i\le p^{(k)}+1$. Let $p^{(k+1)}\coloneqq
\floor{\log(p^{(k)}+1)}$ and $\ell^{(k+1)}\coloneqq
\ell^{(k)}+p^{(k+1)}$. Then
\[
  n^{(k+1)}_i<2^{\ell^{(k+1)}}
\]
and for all $j$,
\[
  \Bit(j,n^{(k+1)}_i)\neq 0\implies i\le j\le i+p^{(k+1)}.
\]
Observe that if
$p^{(k)}=1$ then $p^{(k')}=1$ for all $k'>k$. Let $k^*$ be the least
$k$ such that $p^{(k)}=1$. It is easy to see that
\begin{equation}
  \label{eq:100}
  \begin{array}{rcccll}
    &&p^{(k^*-1)}&=&2,\\
    3&\le&p^{(k^*-2)}&\le& 6,\\
    7&\le& p^{(k^*-3)}&\le& 126,\\
    127&\le& p^{(k)}&&&\text{for }k<k^*-3.
  \end{array}
\end{equation}

\begin{techclaim}\label{dummy}\label{claim:a:1}
  \begin{enumerate}
\item $\displaystyle\sum_{i=2}^{k^*}ip^{(i)}\le 7\log\log
    m$;
  \item $\displaystyle \sum_{i=1}^{k^*}\bsize(p^{(i)})\le 7\log\log
    m$.
  \end{enumerate}
\end{techclaim}
\begin{subproof}
  By induction on $k=k^*,k^*-1,\ldots,2$ we prove
  \begin{equation}
    \label{eq:98}
    \sum_{i=k}^{k^*}i\cdot p^{(i)}\le 3k p^{(k)}
  \end{equation}
  As base case, we need to check \eqref{eq:98} for
  $k\in\{k^*,k^*-1,k^*-2,k^*-3\}$. Using \eqref{eq:100}, this is straightforward. For example, if $k=k^*-2$ and
  $p^{(k)}=3$ we have 
    \[
      \sum_{i=k}^{k^*}ip^{(i)}=3(k^*-2)+2(k^*-1)+k^*=6k^*-8\le
      9k^*-18=3kp^{(k)}.
    \]
    The inequality $6k^*-8\le
      9k^*-18$ holds because if $2\le k\le k^*-2$ then
    $k^*\ge 4$, which implies $3k^*\ge 12$. 
    Similarly, if $k=k^*-3$ and $p^{(k)}=15$ we have
  $p^{(k+1)}=\floor{\log(p^{(k)}+1)}=4$ and thus
    \[
      \sum_{i=k}^{k^*}ip^{(i)}=15(k^*-3)+4(k^*-2)+2(k^*-1)+k^*=22k^*-55\le
      45k^*-135=3kp^{(k)}.
    \]
    The inequality $22k^*-55\le
      45k^*-135$ holds because if $2\le k\le k^*-3$ then
    $k^*\ge 5$.

  For the inductive step $k+1\mapsto
  k$, where $2\le k<k^*-3$, we argue as follows:
  \begin{align*}
    \sum_{i=k}^{k^*}i\cdot p^{(i)}
    &=k\cdot p^{(k)}+\sum_{i=k+1}^{k^*}i\cdot p^{(i)}\\
    &\le k\cdot p^{(k)}+3(k+1)p^{(k+1)}&\text{induction hypothesis}\\
    &\le k\cdot p^{(k)}+3(k+1)\log(p^{(k)}+1)\\
    &\le k\cdot p^{(k)}+4k\log(p^{(k)}+1)&\text{because }k\ge 2.
  \end{align*}
  Since by \eqref{eq:100} we have $p^{(k)}\ge 127$ for $k< k^*-3$, we have $p^{(k)}\ge
  4\log(p^{(k)}+1)$.  Inequality \eqref{eq:98} follows:
  \[
    \sum_{i=k}^{k^*}i\cdot p^{(i)}\le p^{(k)}+4k\log(p^{(k)}+1)\le 2kp^{(k)}.
  \]
  For $k=2$, this yields
  \[
    \sum_{i=2}^{k^*}i\cdot p^{(i)}\le 6p^{(2)}\le 6\log (p^{(1)}+1)\le
    6\log (\log m+1)\le 7\log\log m.
  \]

  \medskip
  To prove (2), note that $\bsize(p^{(k^*)})=\bsize(1)=1$ and
  $\bsize(p^{(k)})=\ceil{\log(p^{(k)}+1)}\le p^{(k+1)}+1$ for
  $k<k^*$. Thus (2) holds if $k^*=1$. If $k^*\ge 2$ we have
  \[
    \sum_{i=1}^{k^*}\bsize(p^{(i)})
    =1+\sum_{i=2}^{k^*}(p^{(i)}+1)
    =k^*+
    \sum_{i=2}^{k^*}p^{(i)}\le k^*+1+\sum_{i=2}^{k^*-1} p^{(i)}\le \sum_{i=2}^{k^*}ip^{(i)},
  \]
  and (2) follows from (1).
\end{subproof}

\begin{techclaim}\label{claim:a:2}
  There is an arithmetical $\FOC$-term $\logic{pseq}(y)$ such
  that
  \[
    \sem{\logic{pseq}}^\Ca(m)=\code{(p^{(1)},\ldots,p^{(k^*)})}.
  \]
  \end{techclaim}
\begin{subproof}
  Let
  \begin{align*}
    \phi(y,z)\coloneqq\,&
                          \logic{seq}(z)\wedge\logic{entry}(0,z)=\logic{flog}(y)\\
    &\wedge
    \forall
    y'<\logic{seqlen}(z)-1.\logic{entry}(y'+1,z)=\logic{flog}(\logic{entry}(y',z)+1)
  \end{align*}
  where the formula $\logic{seq}$ and the terms $\logic{entry}$,
  $\logic{seqlen}$ are from Lemma~\ref{lem:itadd1} and the term
  $\logic{flog}$ is from Lemma~\ref{lem:itadd2}. Then for all
  $q\in\Nat$ we have
  \[
    \sem{\phi}^\Ca(m,q)=1\iff q=\code{(p^{(1)},\ldots,p^{(k^*)})}.
  \]
  By Claim~\ref{claim:a:1}(2) and \eqref{eq:96} we have
  $\code{(p^{(1)},\ldots,p^{(k^*)}}<m$ (for sufficiently large $m$). Thus the term
  \[
    \logic{pseq}(y)\coloneqq \#y'<y.\exists z<y.\big(\phi(y,z)\wedge
    y'<z\big).
  \]
  defines $\code{(p^{(1)},\ldots,p^{(k^*)})}$.\uend
\end{subproof}

\medskip
Recall that for $1\le k\le k^*$ and $0\le j<\ell^{(k)}$ we have
\[
  s_j^{(k)}=\sum_{i=\max\{0,j-p^{(k)}\}}^jn_{i,j}^{(k)}.
\]
To simplify the notation, in the following, we let $n_i^{(k)}=0$ and
$s_i^{(k)}=0$ for all $k$ and $i<0$. Of course then the bits
$n_{i,j}^{(k)}$ are also $0$, and we can write
\[
  s_j^{(k)}=\sum_{i=j-p^{(k)}}^jn_{i,j}^{(k)}.
\]
As $n_{i,j}^{(k)}=\Bit(j,n_i^{(k)})$ and $n_i^{(k)}=2^is_i^{(k-1)}$,
we thus have
\begin{align}
  \notag
  s_j^{(k)}&=\sum_{i=j-p^{(k)}}^j\Bit(j, 2^is_i^{(k-1)})
  =\sum_{i=j-p^{(k)}}^j\Bit(j-i, s_i^{(k-1)})\\
    \label{eq:97}
  &=\big|\big\{ i\bigmid j-p^{(k)}\le i\le j,\Bit(j-i,
  s_i^{(k-1)})=1\big\}\big|.
\end{align}

\begin{techclaim}\label{claim:a:3}
  For every $t\ge 0$ there is an arithmetical $\FOC$-terms $\logic{s}^{(t)}(y)$ 
  such that for all $j\in\Nat$ we have
  \[
    \sem{\logic{s}^{(t)}}^\Ca(j)=s^{(t)}_j.
  \]

 \end{techclaim}
\begin{subproof}
  We define the terms inductively, using \eqref{eq:93} for the base
  step and \eqref{eq:97} for the inductive step.
  \uend
\end{subproof}

\medskip
Note that this claim only enables us to define the numbers $s^{(k^*)}_j$ by
a formula that depends on $k^*$ and hence on the input, or more
formally, the assignment $\Ca$. Claim~\ref{claim:a:6} below will show that we can
also define the $s^{(k*)}_j$ by a formula that is independent of the
input. (We will only use Claim~\ref{claim:a:3} to define the $s^{(2)}_j$.)

\begin{techclaim}\label{claim:a:4}
  There is an arithmetical $\FOC$-formula $\logic{step}(y,z)$
  such that for all $k\in[k^*]$, $j\in\{0,\ldots,\ell^{(k)}-1\}$, and
  $j,k,s,t\in\Nat$, if
  \[
    t=\code{(s^{(k-1)}_{j-p^{(k)}},s^{(k-1)}_{j-p^{(k)}+1},\ldots,i_{s}=s^{(k-1)}_j)}
  \]
  then
  \[
    \sem{\logic{step}}^\Ca(s,t)=1\iff s=s^k_j.
  \]

 \end{techclaim}
\begin{subproof}
  This follows easily from \eqref{eq:97}.
\end{subproof}

To compute $s_j^{(k^*)}$, we need to know $s_j^{(k^*-1)},\ldots,
s_{j -p^{(k^*)}}^{(k^*-1)}$. To compute these numbers, we need to know 
$s_j^{(k^*-2)},\ldots,
s_{j -p^{(k^*)}-p^{(k^*-1)}}^{(k^*-2)}$, et cetera. Thus if we want to
compute $s_j^{(k^*)}$ starting from values $s^{(k)}_{j'}$, we
need to know $s_{j'}^{(k)}$ for
\[
  j-\sum_{i=k+1}^{k^*}p^{(i)}\le j'\le j.
\]
For $2\le k<k^*$, we let
\[
  \vec s_{jk}\coloneqq
  (s^{(k)}_{j-\sum_{i=k+1}^{k^*}p^{(i)}},\ldots,s^{(k)}_j),
\]
and we let $\vec s_{jk^*}\coloneqq(s^{(k^*)}_j)$. Concatenating these
sequences, we let
\[
  \vec s_j\coloneqq\vec s_{j2}\vec s_{j3}\ldots\vec s_{jk^*}.
\]

\begin{techclaim}\label{claim:a:5}
  There is an arithmetical $\FOC$-formula $\logic{all-s}(y,z)$ such
  that for all $t\in\Nat$ and $j\in\{0,\ldots,\ell^{(k^*)}-1\}$ we have
  \[
    \sem{\logic{all-s}}^\Ca(j,t)=1\iff t=\code{\vec s_j}.
  \]

  \end{techclaim}
\begin{subproof} This follows from Claim~\ref{claim:a:2} (to get the $p^{(i)}$ as well as
  $k^*$), Claim~\ref{claim:a:3} (to get the base values $s^{(2)}_{j'}$), and Claim~\ref{claim:a:4}
  (to verify the internal values of the sequence), just requiring a
  bit of arithmetic on the indices.
  \uend
\end{subproof}

\begin{techclaim}\label{claim:a:6}
  There is an arithmetical $\FOC$-term $\logic{skstar}(y)$ such
  that for all $j\in\Nat$ we have
  \[
    \sem{\logic{skstar}}^\Ca(j)=s^{(k^*)}_j.
  \]
  \end{techclaim}
\begin{subproof}
  We first prove that for all $j$ we have
  \begin{equation}
    \label{eq:94}
    \code{\vec s_j}< m.
  \end{equation}
  The key observation is that the length of the sequence $\vec s_j$ is
  \[
    |\vec s_j|=\sum_{k=2}^{k*}|\vec
    s_{jk}|=\sum_{k=2}^{k*}\left(1+\sum_{i=k+1}^{k^*}p^{(i)}\right)
    \le \sum_{k=2}^{k*}\sum_{i=k}^{k^*}p^{(i)}\le\sum_{k=2}^{k^*}i
    p^{(i)}\le7\log\log m,
  \]
  where the last inequality holds by Claim~\ref{claim:a:1}. Moreover, for $2\le k\le
  k^*$ and $0\le j<\ell^{(k)}$ we have $s^{(k)}_j\le p^{(k)}+1\le
  p^{(2)}+1\le \log m +1$ and thus $\bsize(s^{(k)}_j)=O(\log\log
  m)$. Thus $\vec s_j$ is a sequence of $O(\log\log m)$ numbers, each of
  bitsize $O(\log\log m)$. Thus the bitsize of the sequence is
  $O((log\log m)^2)$, and it follows from \eqref{eq:96} that for some
  constant $c$,
  \[
    \code{\vec s_j}\le 2^{c(\log\log m)^2}<m,
  \]
  again assuming that $m$ is sufficiently large. This proves
  \eqref{eq:94}.

  We let
  \[
    \logic{skstar}(y)\coloneq\exists z<U.\Big(\logic{all-s}(y,z)\wedge
    y=\logic{entry}\big(\logic{seqlen}(z)-1,z\big)\Big),
  \]
  where the formula $\logic{all-s}(y,z)$ is from Claim~\ref{claim:a:5} and the terms
  $\logic{entry}$ and $\logic{seqlen}$ are from
  Lemma~\ref{lem:itadd1}. 
  \uend
\end{subproof}

\medskip
It remains to compute $\sum_{i=0}^{m^{(k^*)}-1}n^{(k^*+1)}_i$, where
$n_i^{(k^*+1)}=2^is_i^{(k^*)}$. Since $s^{(k^*)}_j\le 2$, we have
$\Bit(j,n^{(k^*)}_i)=0$ unless $j\in\{i-1,i\}$. We split the sum into
the even and the odd entries:
\[
  \sum_{i=0}^{m^{(k^*)}-1}n^{(k^*+1)}_i=
  \underbrace{
    \sum_{i=0}^{\ceil{m^{(k^*)}/2}-1}n^{(k^*)}_{2i}
  }_{\eqqcolon n^*_1}
  +
  \underbrace{
    \sum_{i=0}^{\floor{m^{(k^*)}/2}-1}n^{(k^*)}_{2i+1}
  }_{\eqqcolon n^*_2}.
\]
Since the entries in the two partial sums have no non-zero bits in
common, it is easy to define the two partial sums (of course, using
Claim~\ref{claim:a:6}). Then we can apply Lemma~\ref{lem:ar1} to define
$n_1^*+n_2^*$. 
\end{proof}

\begin{lemma}\label{lem:ar2a}
  Let $Y_1,Y_2$ be relation variables of type $\tta\ttn\ttz$, and let
  $U_1,U_2$ be function variables of type $\emptytuple\to\ttn$.
  Then there are arithmetical $\FOC$-formulas $\logic{mul}$, $\logic{div}$, and  $\FOC$-terms
    $\logic{{bd-mul}}$, $\logic{{bd-div}}$ such
    that for all numerical assignments $\Ca$,
    \begin{align*}
            \num{\logic{mul}, \logic{{bd-mul}}}^{\Ca}
      &=\num{Y_1,U_1}^{\Ca}\cdot
        \num{Y_2,U_2}^{\Ca},\\
            \num{\logic{div}, \logic{{bd-div}}}^{\Ca}
      &=\floor{\frac{\num{Y_1,U_1}^{\Ca}}
        {\num{Y_2,U_2}^{\Ca}}}&\text{if }\num{Y_2,U_2}^{\Ca}\neq0.
    \end{align*}
 \end{lemma}

 \begin{proof}
   The expressibility multiplication follows easily from the
   expressibility of iterated addition (Lemma~\ref{lem:ar2}). Division
   is significantly more difficult, we refer the reader to \cite{HesseAB02}. 
 \end{proof}

We need be an extension of Lemma~\ref{lem:ar2} where the
family of numbers is no longer indexed by numbers, but by arbitrary
tuples of vertices and numbers, and where the bounds on the bitsize of
the numbers in the family are not uniform (in Lemma~\ref{lem:ar2} we use
the $0$-ary function variable $U$ to provide a bound on the bitsize of
all numbers in our family). Before presenting this
extension, we consider a version of iterated addition as well as
taking the maximum and minimum of a family of numbers given directly
as values of a function instead of the binary representation that we
considered in the previous lemmas. 

\begin{lemma}\label{lem:ar3a}
  Let $X$ be a relation variable of type $\tta\ttv^k\ttn^\ell\ttz$,
  and let $U,V$ be function variables of types
  $\ttv^k\ttn^\ell\to\ttn$, $\ttv^k\to\ttn$, respectively.
  \begin{enumerate}
  \item
    There is a $\FOC$-term $\logic{u-itadd}$\footnote{The 'u' in $\logic{u-itadd}$
      indicates that we use a unary representation here.}
    such that for
    all structures $A$ and assignments $\Ca$ over $A$,
    \[
      \sem{\logic{u-itadd}}^{(A,\Ca)}=\sum_{(\vec a,\vec
        b)}\Ca(U)(\vec a,\vec b),
    \]
    where the sum ranges over all $(\vec a,\vec b)\in\Ca(X)$ such that
    $\vec b=(b_1,\ldots, b_\ell)\in\Nat^\ell$ with
    $b_i<\Ca(V)(\vec a)$ for all $i\in[\ell]$.
  \item There are $\FOC$-terms $\logic{u-max}$ and $\logic{u-min}$ such
    that for all structures $A$ and assignments $\Ca$ over $A$,
    \begin{align*}
      \sem{\logic{u-max}}^{(A,\Ca)}&=\max_{(\vec a,\vec b)}\Ca(U)(\vec
                                     a,\vec b),\\
      \sem{\logic{u-min}}^{(A,\Ca)}
                                   &=\min_{(\vec a,\vec b)}\Ca(U)(\vec a,\vec b),
    \end{align*}
    where $\max$ and $\min$ range  over all $(\vec a,\vec b)\in\Ca(X)$
    such that $\vec b=(b_1,\ldots, b_\ell)\in\Nat^\ell$ with
    $b_i<\Ca(V)(\vec a)$ for all $i\in[\ell]$.
  \end{enumerate}
  Furthermore, if $k=0$ then the terms $\logic{u-itadd}$,
  $\logic{u-max}$, and $\logic{u-min}$ are arithmetical.
\end{lemma}

\begin{proof}
  To simplify the notation, in the proof we assume $\ell=1$. The
  generalisation to arbitrary $\ell$ is straightforward.
  
  The trick is to use adaptive bounds in our counting terms. Let $A$
  be a structure and $\Ca$ an assignment over $A$. Let $\CI\subseteq V(A)^k\times\Nat$ be the set
  of all all $(\vec a,b)\in\Ca(X)$ with
    $b<\Ca(V)(\vec a)$. We exploit that
    for all $(\vec a,b)\in {\CI}$, $\Ca(U)(\vec a,b)$ is the
    number of $c\in\Nat$ such that $c<\Ca(U)(\vec a,b)$.
    This implies
    \begin{align*}
      \sum_{(\vec a,b)}\Ca(U)(\vec a,b)&=\Big|\Big\{(\vec a,b,c)\Bigmid
      (\vec a,b)\in {\CI}\text{ and }c<\Ca(U)(\vec a,b)\Big\}\Big|\\
      &=\Big|\Big\{(\vec a,b,c)\Bigmid
      (\vec a,b)\in\Ca(X)\text{ with }b<\Ca(V)(\vec
        a)\text{ and }c<\Ca(U)(\vec a,b)\Big\}\Big|.
    \end{align*}
    Thus 
    \[
      \logic{u-itadd}\coloneqq\#\big(\vec x,y< V(\vec
      x),z< U(\vec
      x,y)\big).X(\vec x,y).
    \]
    satisfies assertion (1).

    \medskip
    Once we have this, assertion (2) is easy because maximum and
    minimum are bounded from above by the sum. For example, we let
    \[
      \logic{u-max}\coloneqq \# z<\logic{u-itadd}.\exists\vec
      x.\exists y<V(\vec x).\big(X(\vec x)\wedge z< U(\vec x,y)\big).
      \qedhere
    \]
\end{proof}

The following lemma is the desired generalisation of
Lemma~\ref{lem:ar2}.

\begin{lemma}\label{lem:ar3}
  Let $X,Y$ be relation variables of type $\tta\ttv^k\ttn^\ell\ttz$
  and $\tta\ttn\ttv^k\ttn^\ell\ttz$, respectively, and let $U,V$ be
  function variables of type $\ttv^k\ttn^\ell\to\ttn$ and
  $\ttv^k\to\ttn$, respectively.
  Then there is an $\FOC$-formula $\logic{itadd}$
  and  an $\FOC$-term
    $\logic{{bd-itadd}}$ such
    that for all structures $A$ and assignments $\Ca$ over $A$,
    \[
      \num{\logic{itadd}, \logic{{bd-itadd}}}^{(A,\Ca)}=
      \sum_{(\vec a,\vec b)}\num{Y,U}^{(A,\Ca)}(\vec a,\vec b),
    \]
    where the sum ranges over all $(\vec a,\vec b)\in\Ca(X)$ such that $\vec b=(b_1,\ldots, b_\ell)\in\Nat^\ell$ with
    $b_i<\Ca(V)(\vec a)$ for all $i\in[\ell]$.

    If $k=0$, then the formula $\logic{itadd}$ and the term
    $\logic{{bd-itadd}}$ are arithmetical.
\end{lemma}

\begin{proof}
  Let $A$ be a structure and $\Ca$ an assignment over $A$. Let
  $\CI\subseteq V(A)^k\times\Nat^\ell$ be the set of all all
  $(\vec a,\vec b)\in\Ca(X)$ such that $\vec b=(b_1,\ldots,b_\ell)$
  with $b_i<\Ca(V)(\vec a)$. Let $m\coloneqq|\CI|$ and
  $p\coloneqq \max\big\{\Ca(U)(\vec a,\vec b)\bigmid (\vec a,\vec
  b)\in \CI\big\}$. Note that $p=\sem{\logic{u-max}}^{(A,\Ca)}$ for
  the term $\logic{u-max}$ of Lemma~\ref{lem:ar3a}.  We have to add
  the family of $m$ numbers
  ${n}_{\vec c}\coloneqq\num{Y,U}^{(A,\Ca)}(\vec c)$ for
  $\vec c\in \CI$. We think of these numbers as $p$-bit numbers,
  padding them with zeroes if necessary. Note that we cannot directly
  apply Lemma~\ref{lem:ar2} to add these numbers, because the family,
  being indexed by vertices, is not ordered, and Lemma~\ref{lem:ar2}
  only applies to ordered families indexed by numbers. But there is a
  simple trick to circumvent this difficulty. (We applied the same
  trick in the proof of Lemma~\ref{lem:ar2}).

  For all $i< p$, we let
  ${n}_{{\vec c},i}\coloneqq\Bit(i,{n}_{\vec c})$ be the $i$th bit
  of ${n}_{\vec c}$, and we let
  \[
    s_{i}\coloneqq\sum_{{\vec c}\in \CI}{n}_{{\vec c},i}.
  \]
  We have
  \begin{align*}
    \sum_{{\vec c}\in \CI}{n}_{\vec c}
    =\sum_{{\vec c}\in \CI}\sum_{i=0}^{p-1}{n}_{{\vec c},i}\cdot2^i
    =\sum_{i=0}^{p-1}s_i\cdot2^i.
  \end{align*}
  This reduces the problem of adding the \emph{unordered} family of
  $m$ $p$-bit numbers ${n}_{{\vec c}}$ to adding the \emph{ordered} family
  of $p$ $p+\log m$-bit numbers ${n}_i'\coloneqq s_i\cdot2^i$.

  The partial sums $s_i$ are definable by a counting term in $\FOC$, because
  $s_i$ is the number of ${\vec c}\in \CI$ such that ${n}_{{\vec
      c},i}=1$. As the bit predicate is definable in $\FOC$, we can
  obtain the bit-representation of these numbers and then shift $s_i$ by $i$ to obtain the bit
  representation of ${n}_i'=s_i\cdot 2^i$. Then we can apply
  Lemma~\ref{lem:ar2} to compute the sum.
\end{proof}

\begin{lemma}\label{lem:ar3b}
  Let $X,Y$ be relation variables of type $\tta\ttv^k\ttn^\ell\ttz$
  and $\tta\ttn\ttv^k\ttn^\ell\ttz$, respectively, and let $U,V$ be
  function variables of type $\ttv^k\ttn^\ell\to\ttn$ and
  $\ttv^k\to\ttn$, respectively.
  Then there are $\FOC$-formulas $\logic{itmax}$, $\logic{itmin}$
  and  $\FOC$-terms
    $\logic{{bd-itmax}}$, $\logic{{bd-itmin}}$ such
    that for all structures $A$ and assignments $\Ca$ over $A$,
    \begin{align*}
      \num{\logic{itmax},\logic{bd-itmax}}^{(A,\Ca)}
      &=\max_{(\vec a,\vec b)}\num{Y,U}^{(A,\Ca)}(\vec a,\vec b),\\
      \num{\logic{itmin},\logic{bd-itmin}}^{(A,\Ca)}
      &=\min_{(\vec a,\vec b)}\num{Y,U}^{(A,\Ca)}(\vec a,\vec b),
    \end{align*}
    where max and min range over all $(\vec a,\vec b)\in\Ca(X)$ such that $\vec b=(b_1,\ldots, b_\ell)\in\Nat^\ell$ with
    $b_i<\Ca(V)(\vec a)$ for all $i\in[\ell]$.

    If\/ $k=0$, then the formula $\logic{itmax}$, $\logic{itmin}$ and the term
    $\logic{{bd-itmin}}$, $\logic{{bd-itmin}}$ are arithmetical.
\end{lemma}

\begin{proof}
   Again, to reduce the notational overhead we assume $\ell=1$. We
   only give the proof for the maximum; the proof for the minimum is
   completely analogous.

   The following formula says that $(\vec x,y)$ is the index of the
   maximum number in the family:
   \[
     \logic{maxind}(\vec x,y)\coloneqq X(\vec x,y)\wedge y<V(\vec
     x)\wedge \forall \vec x'\forall y'< V(\vec x').\big(X(\vec
     x',y)\to\logic{leq}'(\vec x',y',\vec x, y)\big),
   \]
   where $\logic{leq}'(\vec x',y',\vec x,\vec y)$ is the formula
   obtained from the formula $\logic{leq}$ of Lemma~\ref{lem:ar1}(2)
   by substituting $Y_1(z)$ with $Y(z,\vec x',y')$, $U_1()$ with
   $U(\vec x',y')$, $Y_2(z)$ with $Y(z,\vec x,y)$, $U_2()$ with
   $U(\vec x,y)$.

   Then we let
   \begin{align*}
     \logic{itmax}(\hat y)
     &\coloneqq\exists \vec x\exists y< V(\vec
       x).\big(\logic{maxind}(\vec x,y)\wedge Y(\hat y,\vec
       x,y)\big),\\
     \logic{bd-itmax}&\coloneqq\# z<\logic{u-max}.\forall\vec x\forall
     y<V(\vec x)\big(\logic{maxind}(\vec x,y)\to z<U(\vec x,y)\big),
   \end{align*}
   where $\logic{u-max}$ is the formula of Lemma~\ref{lem:ar3a}.
\end{proof}

\subsection{Rational Arithmetic}
\label{sec:rat-arithmetic}
We need to lift the results of the previous section to arithmetic on
rational numbers. However, we will run into a problem with iterated
addition, because the denominator of the sum can get too large. To
avoid this problem, we will work with arithmetic on dyadic
rationals. Then we have a problem with division, because the dyadic
rationals are not closed under division, but division is not as
important for us as iterated addition.

Our representation system for dyadic rationals by relation and
function variables, or by formulas and terms, is based on a
representations of dyadic rationals by tuples
$(r,I,s,t)\in\{0,1\}\times2^{\Nat}\times\Nat\times\Nat$: such a
tuple represents the number
\[
  \num{r,I,s,t}\coloneqq(-1)^r\cdot2^{-s}\cdot\sum_{i\in I, i<t}2^i.
\]
This representation is not unique: there are distinct tuples $(r,I,s,t)$
and $(r',I',s',t')$ representing the same number. For
example, $\num{r,I,s,t}=\num{r,I',s+1,t+1}$, where $I'=\{i+1\mid i\in
I,i<t\}$. However, each dyadic rational $q$ has a unique
representation $\crep(q)=(r,I,s,t)$ satisfying the following conditions:
\begin{eroman}
\item $i<t$ for all $i\in I$;
\item $s=0$ or $0\in I$ (that is, the fraction $\frac{\sum_{i\in
      I,i<t}2^i}{2^s}$ is reduced); 
\item if $I=\emptyset$ (and hence $\num{s,I,s,t}=0$) then $r=s=t=0$.
\end{eroman}
We call $\operatorname{crep}(q)$ the \emph{canonical representation}
of $q$. 

To represent a dyadic rational in our logical framework, we thus need
four variables. As this tends to get a bit unwieldy, we
introduce some shortcuts. 
An \emph{r-schema} of type $\rtp{\vec t}$ for some
$\vec t\in\{\ttv,\ttn\}^k$ is a tuple $\vec Z=(\Zr,\ZI,\Zs,\Zt)$, where $\Zr$ is
a relation variable of type $\tta\vec t\ttz$, $\ZI$ is a relation
variable of type $\tta \ttn\vec t\ttz$, and $\Zs,\Zt$ are function
variables of type $\vec t\to\ttn$.
For a structure $A$, an interpretation $\Ca$ over $A$, and a tuple
$\vec c\in A^{\vec t}$ we
let
\begin{equation}\label{eq:num5}
\num{\vec Z}^{(A,\Ca)}(\vec c)\coloneqq (-1)^r\cdot
  2^{-\Ca(\Zs)(\vec c)}\cdot\sum_{\substack{(i,\vec c)\in\Ca(\ZI),\\i<\Ca(\Zt)(\vec c)}}2^i,
\end{equation}
where $r=1$ if $\vec c\in\Ca(\Zr)$ and $r=0$ otherwise. Note that
with
$I=\{i\in\Nat\mid (i,\vec c)\in\Ca(\ZI)\}$, $s=\Ca(\Zs)(\vec c)$, and
$t=\Ca(\Zt)(\vec c)$ we have $\num{\vec Z}^{(A,\Ca)}(\vec
c)=\num{r,I,s,t}$.

An \emph{r-expression} is a tuple
$\rexp(\vec z)=\big(\rr(\vec z),\rI(\hat y,\vec z),\rs(\vec
z),\rt(\vec z)\big)$ where $\vec z$ is a tuple of individual
variables, $\rr(\vec z),\rI(\hat y,\vec z)$ are $\FOC$-formulas, and
$\rs(\vec z),\rt(\vec z)$ are $\FOC$-terms. For a structure $A$, an
interpretation $\Ca$ over $A$, and a tuple $\vec c\in A^{\tp(\vec z)}$
we let
\begin{equation}\label{eq:num5a}
\num{\rexp}^{(A,\Ca)}(\vec c)\coloneqq \num{r,I,s,t},
\end{equation}
where $r=1$ if $(A,\Ca)\models \rr(\vec c)$ and $r=0$ otherwise, $I$
is the set of all $i\in\Nat$ such that $A\models\rI(i,\vec c)$,
$s=\sem{\rs}^{(A,\Ca)}(\vec c)$, and $t=\sem{\rt}^{(A,\Ca)}(\vec
c)$. We sometimes say that $\rexp$ \emph{defines} the representation
$(r,I,s,t)$ of the dyadic rational $\num{r,I,s,t}$ \emph{in
  $(A,\Ca)$}.

For a fragment $\LL$ of
$\FOC$, such as those introduced in Section~\ref{sec:guarded},
an \emph{r-expression in $\LL$} is an r-expression consisting of formulas and
terms from $\LL$. An
\emph{arithmetical r-expression} is an r-expression consisting of
arithmetical formulas and terms.

For an r-schema $\vec Z$ of type $\rtp{\ttn^k}$, for some $k\ge 0$,
and a numerical assignment $\Ca$ we may just write
$\num{\vec Z}^\Ca(\vec c)$ without referring to a
structure. Similarly, for an arithmetical r-expression $\rexp(\vec z)$
we may write $\num{\rexp}^\Ca(\vec c)$. We use a similar notation
for other objects, in particular the L,F-schemas and L,F-expressions
that will be introduced in Section~\ref{sec:fnnsim}.

\begin{lemma}\label{elm:reduced}
  Let $\vec Z$ be an r-schema of type $\rtp{\emptytuple}$. Then there is an
  arithmetical r-expression $\logic{crep}$ such that for all
  structures $A$ and assignments $\Ca$ over $A$, $\logic{crep}$
  defines the canonical representation of $\num{\vec Z}^{(A,\Ca)}$ in $(A,\Ca)$.
\end{lemma}

\begin{proof}
  Straightforward.
\end{proof}

Using this lemma, in the following we can always assume that the
formulas and terms defining arithmetical operations, as for example,
in Lemma~\ref{lem:ar4}, \ref{lem:ar5}, et cetera, return their results
in canonical representation.

\begin{lemma}\label{lem:ar4}
  Let $\vec Z_1,\vec Z_2$ be r-schemas of type $\rtp{\emptyset}$.
  \begin{enumerate}
  \item There are arithmetical r-expressions $\logic{add}$,
    $\logic{sub}$, and $\logic{mul}$
such
    that for all structures $A$ and assignments $\Ca$ over $A$,
    \begin{align*}
      \num{\logic{add}}^{(A,\Ca)}
      &=\num{\vec Z_1}^{(A,\Ca)}+
        \num{\vec Z_2}^{(A,\Ca)},\\
      \num{\logic{sub}}^{(A,\Ca)}
      &=\num{\vec Z_1}^{(A,\Ca)}-
        \num{\vec Z_2}^{(A,\Ca)},\\
\num{\logic{mul}}^{(A,\Ca)}
      &=\num{\vec Z_1}^{(A,\Ca)}\cdot
        \num{\vec Z_2}^{(A,\Ca)}.
    \end{align*}
  \item There is an arithmetical $\FOC$-formula $\logic{leq}$ such
    that for all structures $A$ and assignments $\Ca$ over $A$,
    \[
      (A,\Ca)\models\logic{leq}\iff \num{\vec Z_1}^{(A,\Ca)}\leq
      \num{\vec Z_2}^{(A,\Ca)}.
    \]
  \end{enumerate}
\end{lemma}

\begin{proof}
  These are straightforward consequences of Lemmas~\ref{lem:ar1} and \ref{lem:ar2a}.
\end{proof}

\begin{lemma}\label{lem:ar5}
  Let $\vec Z$ be an r-schema of type
  $\rtp{\ttv^k\ttn^\ell}$. Furthermore, let $X$ be a relation variable of
  type $\tta\ttv^k\ttn^\ell\ttz$, and let $V$ be function variable of
  type $\ttv^k\to\ttn$.
  \begin{enumerate}
  \item There is an r-expression $\logic{it-add}$
  such that for all structures $A$ and assignments $\Ca$ over $A$,
  \[
    \num{\logic{itadd}}^{(A,\Ca)}=
      \sum_{(\vec a,\vec b)}\num{\vec Z}^{(A,\Ca)}(\vec a,\vec b),
    \]
    where the sum ranges over all $(\vec a,\vec b)\in\Ca(X)$ such that $\vec b=(b_1,\ldots, b_\ell)\in\Nat^\ell$ with
    $b_i<\Ca(V)(\vec a)$ for all $i\in[\ell]$.
  \item There are r-expressions $\logic{max}$ and $\logic{min}$
  such that for all structures $A$ and assignments $\Ca$ over $A$,
  \begin{align*}
    \num{\logic{max}}^{(A,\Ca)}&=
                                 \max_{(\vec a,\vec b)}\num{\vec Z}^{(A,\Ca)}(\vec a,\vec b),\\
        \num{\logic{min}}^{(A,\Ca)}&=
                                 \min_{(\vec a,\vec b)}\num{\vec Z}^{(A,\Ca)}(\vec a,\vec b),
  \end{align*}
    where $\max$ and $\min$ range over all $(\vec a,\vec b)\in\Ca(X)$
    such that $\vec b=(b_1,\ldots, b_\ell)\in\Nat^\ell$ with
    $b_i<\Ca(V)(\vec a)$ for all $i\in[\ell]$.
  \end{enumerate}
  Furthermore, if $k=0$, then the r-expressions $\logic{itadd}$,
  $\logic{max}$, and $\logic{min}$ are arithmetical.
  \end{lemma}

  \begin{proof}
    To express iterated addition, we first split the family of numbers
    into the positive and negative numbers. We take the sums over
    these two subfamilies separately and then combine the results
    using Lemma~\ref{lem:ar4}.  To take the sum over a family of
    nonnegative dyadic rationals, we apply Lemma~\ref{lem:ar3} for the
    numerator and Lemma~\ref{lem:ar3a} for the denominator.

    \medskip
    To express maximum and minimum, it clearly suffices to express the
    maximum and minimum of a family of nonnegative dyadic rationals
    $\left(p_i\cdot 2^{-s_i}\right)_{i\in\CI}$ for some definable finite
    index set $\CI$. Using Lemma~\ref{lem:ar3a}, we can determine
    $s\coloneqq\max_{i\in\CI}s_i$. Then we need to determine maximum
    and minimum of the natural numbers $q_i\coloneqq
    p_i2^{s-s_i}$, which we can do by applying Lemma~\ref{lem:ar3b}.
  \end{proof}

  For division, the situation is slightly more complicated, because
  the dyadic rationals are not closed under division. We only get an
  approximation. We use a $0$-ary function variable to control the
  additive approximation error.

\begin{lemma}\label{lem:ar6}
  Let $\vec Z_1,\vec Z_2$ be r-schemas of type $\rtp{\emptytuple}$, and
  let $W$ be a function variable of type $\emptytuple\to\ttn$. Then
  there is an arithmetical r-expression $\logic{div}$
such
    that for all structures $A$ and assignments $\Ca$ over $A$, if 
    $\num{\vec Z_2}^{(A,\Ca)}\neq 0$ then
    \[
      \left|\frac{\num{\vec Z_1}^{(A,\Ca)}}{\num{\vec Z_2}^{(A,\Ca)}}
      -
      \num{\logic{div}}^{(A,\Ca)}
      \right|
      < 2^{-\Ca(W)}.
    \]
\end{lemma}

\begin{proof}
  This follows easily from  Lemma~\ref{lem:ar1}.
\end{proof}

\subsection{Evaluating Feedforward Neural Networks}
\label{sec:fnnsim}

The most important consequence the results of the previous section
have for us is that we can simulate rational piecewise-linear
FNNs.

Let us first see how we deal with the activation functions.  To
represent a rational piecewise linear function, we need an integer $k$
as well as three families of dyadic rationals: the thresholds
$(t_i)_{1\le i\le k}$, the slopes $(a_i)_{0\le i\le k}$ and the
  constant terms $(b_i)_{0\le i\le k}$. An \emph{L-schema} of type
  $\Ltp{\vec t}$, for some $\vec t\in\{\ttv,\ttn\}^k$, is a tuple
  $\vec Z=(\Zlen,\Zth,\Zsl,\Zco)$, where $\Zlen$ is a function
  variable of type $\vec t\to\ttn$ and $\Zth,\Zsl,\Zco$ are
  r-schemas of type $\ttn\vec t$. Let $A$ be a structure, $\Ca$ an
  assignment over $A$, and $\vec c\in A^{\vec t}$. Let
  $k\coloneqq\Ca(\Zlen)(\vec c)$. For $1\le i\le k$,
  let 
  $t_i\coloneqq\num{\Zth}^{(A,\Ca)}(i,\vec c)$. For $0\le i\le k$, let
  $a_i\coloneqq\num{\Zsl}^{(A,\Ca)}(i,\vec c)$ and
$b_i\coloneqq\num{\Zco}^{(A,\Ca)}(i,\vec c)$. Then if
$t_1<\ldots<t_k$ and for all $i\in [k]$ we have
$a_{i-1}t_i+b_{i-1}=a_it_i+b_i$, we define $\Lin{\vec Z}^{(A,\Ca)}:\Real\to\Real$ to
be the rational piecewise linear function with thresholds $t_i$, slopes
$a_i$, and constants $b_i$. The condition
$a_{i-1}t_i+b_{i-1}=a_it_i+b_i$ guarantees that this function is
continuous. Otherwise, we define $\Lin{\vec
  Z}^{(A,\Ca)}:\Real\to\Real$ to be identically $0$. We can also
define \emph{L-expressions} consisting of formulas of the appropriate
types.

\begin{lemma}
  Let $\vec Y$ be an L-schema of type $\Ltp{\emptyset}$, and let $\vec Z$
  be an r-schema of type $\rtp{\emptyset}$. Then there is an
  arithmetical r-expression
  $\logic{apply}$ such that for all numerical assignments
  $\Ca$,
  \[
    \num{\logic{apply}}^{\Ca}=\Lin{\vec
      Y}^{\Ca}\Big(\num{\vec Z}^{\Ca}\Big).
  \]
\end{lemma}

\begin{proof}
  This follows easily from Lemma~\ref{lem:ar4}.
\end{proof}

To represent an FNN we need to represent the skeleton as well as
all activation functions and parameters. An \emph{F-schema} of type
$\Ftp{\vec t}$ for some $\vec t\in\{\ttv,\ttn\}^k$ is a
tuple $\vec Z = (\ZV,\ZE,\Zac,\Zwt,\Zbi)$ where $\ZV$ is a function
variable $\vec t\to\ttn$, $\ZE$ is a relation
variable of type $\tta\ttn^2\vec t\ttz$, $\Zac$ is an L-schema of
type $\Ltp{\ttn\vec t}$, $\Zwt$ is an r-schema of
type $\rtp{\ttn^2\vec t}$, and $\Zbi$ is an r-schema of
type $\rtp{\ttn\vec t}$. Then for every structure $A$, every
assignment $\Ca$ over $A$, and every tuple $\vec c\in A^{\vec t}$, we
define $(V,E,(\Fa_v)_{v\in
  V},(w_e)_{e\in E},(b_v)_{v\in V})$ as follows:
\begin{itemize}
\item $V\coloneqq\{0,\ldots,\Ca(\ZV)(\vec c)\}$;
\item $E\coloneqq\{ij\in V^2\mid ij\vec c\in \Ca(\ZE)\}$;
\item $\Fa_i=\Lin{\Zac}^{(A,\Ca)}(i,\vec c)$ for $i\in V$;
\item $w_{ij}=\num{\Zwt}^{(A,\Ca)}(i,j,\vec c)$ for $ij\in E$;
\item $b_i=\num{\Zbi}^{(A,\Ca)}(i,\vec c)$ for $i\in V$.
\end{itemize}
Then if $(V,E)$ is a dag, $\Fin{\vec Z}^{(A,\Ca)}\coloneqq (V,E,(\Fa_v)_{v\in
  V},(w_e)_{e\in E},(b_v)_{v\in V}) $ is an FNN. The input nodes $X_1,\ldots,X_p$ of this FNN are the sources of the
dag $(V,E)$ in their natural order (as natural numbers). Similarly,
the output nodes $Y_1,\ldots,Y_q$ of the FNN are the sinks of the
dag $(V,E)$ in their natural order. If $(V,E)$ is not a
dag, we simply define $\Fin{\vec Z}^{(A,\Ca)}$  to be the trivial FNN with a
single node, which computes the identity function. We can also define
\emph{F-expressions} consisting of formulas of the appropriate types.

\begin{lemma}\label{lem:ar6a}
  Let $\vec Z$ be an F-schema of type $\Ftp{\emptytuple}$, and let
  $\vec X$ be an r-schema of type $\rtp{\ttn}$. Then for every
  $t\ge 0$ there is an arithmetical r-expression $\logic{eval}_t(y)$ such that the
  following holds. Let $\Ca$ be a numerical assignment
  and $\FF\coloneqq\Fin{\vec Z}^{\Ca}$. Suppose that the input
  dimension of $\FF$ is $p$, and let
  \[
    \vec x\coloneqq\big(\num{\vec X}^{\Ca} (0),\ldots, \num{\vec
      X}^{\Ca} (p-1)\big).
  \]
  Then for every node $v$ of $\FF$ of depth $t$ it holds that
  \[
    f_{\FF,v}(\vec x)=\num{\logic{eval}_t}^{\Ca}(v).
  \]
\end{lemma}

\begin{proof}
  Using the formulas for multiplication and iterated addition, it easy
  to construct $\logic{eval}_t$ by induction on $t$.
\end{proof}

\begin{corollary}\label{cor:ar6b}
  Let $\vec Z$ be an F-schema of type $\Ftp{\emptytuple}$, and let
  $\vec X$ be r-schemas of type $\rtp{\ttn}$. Then for every $d>0$ there is an arithmetical r-expression $\logic{eval}_d(y)$ such that the
  following holds. Let $A$ be a structure, $\Ca$ an assignment,
  and $\FF\coloneqq\Fin{\vec Z}^{(A,\Ca)}$. Suppose that the
  depth of $(V,E)$ is at most $d$, and let $p$ be the input dimension
  and $q$ the output dimension. Let
  \[
    \vec x\coloneqq\big(\num{\vec X}^{(A,\Ca)}(0),\ldots, \num{\vec
      X}^{(A,\Ca)}(p-1)\big).
  \]
  Then 
  \[
    \FF(\vec x)=\big(\num{\logic{eval}_d}^{(A,\Ca)}(0),\ldots, \num{\logic{eval}_d}^{(A,\Ca)}(q-1)\big).
  \]
\end{corollary}

\begin{corollary}\label{cor:ar6a}
  Let $\FF$ be a rational piecewise linear FNN of input dimension $p$
  and output dimension $q$, and let $\vec X_1,\ldots,\vec X_p$ be r-schemas of type
  $\rtp{\emptytuple}$.  Then for all $i\in[q]$ there is an
  arithmetical r-expression
  $\logic{eval}_{\FF,i}$ such that for all structures $A$ and
  assignments $\Ca$ over $A$,
  \[
    \FF\Big(\num{\vec X_1}^{(A,\Ca)},\ldots, \num{\vec
      X_p}^{(A,\Ca)}\Big)=\Big(\num{\logic{eval}_{\FF,1}}^{(A,\Ca)},\ldots,
    \num{\logic{eval}_{\FF,q}}^{(A,\Ca)}\Big).
  \]
\end{corollary}

\subsection{Fragments of \texorpdfstring{$\FOC$}{FO+C}}
\label{sec:guarded}

To describe the expressiveness of graph neural networks, we need to
consider various fragments of $\FOC$. For $k\ge 1$, the
\emph{$k$-variable fragment} $\FOC[k]$ of $\FOC$ consists of all
formulas with at most $k$ vertex variables. Importantly, the number of
number variables is unrestricted. We call an $\FOC[k]$-formula
\emph{decomposable} if it contains no relation variables or function
variables and every subformula with exactly $k$ free vertex variables
is a Boolean combination of relational atoms and formulas with at most
$k-1$ free vertex variables.  Equivalently, an $\FOC[k]$-formula is
decomposable if it contains no relation variables or function
variables and every subformula of the form $\theta\le\theta'$, for
terms $\theta,\theta'$, has at most $k-1$ free vertex variables. Note
that this implies that a decomposable $\FOC[k]$-formula contains no
terms with $k$ free vertex variables.

\begin{example}
  The $\FOC[2]$-formula
  \[
    \phi(z)\coloneqq\exists x_1.\exists
    x_2.\big(E(x_1,x_2)\wedge z=\#x_2.E(x_2,x_1)\wedge
    z=\#x_1. E(x_2,x_1)\big)
  \]
  is decomposable, whereas the $\FOC[2]$-formula
  \[
    \psi(z)\coloneqq\exists x_1.\exists
    x_2.\Big(E(x_1,x_2)\wedge z=\big(\#x_2. E(x_2,x_1)\big)\cdot
    \big(\#x_1. E(x_2,x_1)\big)\Big)
  \]
  is not. However, $\psi(z)$ is decomposable if viewed as an
  $\FOC[3]$-formula.
  \uend
\end{example}

\begin{lemma}\label{lem:fo2fo2}
  Let $\phi$ be an $\FOC$-formula of vocabulary $\tau\cup\{\les\}$
  with at most one free vertex variable and no relation or function variables. Then there is a decomposable $\FOC[2]$-formula
  $\phi'$ such that for all ordered $\tau$-structures $A$ and all
  assignments $\Ca$ over $A$ it holds that
  $A\models\phi\iff A\models\phi'$.
\end{lemma}

\begin{proof}
  We first define a bijection between the vertices of the structure
  $A$ and an initial segment of $\Nat$. We simply let
  $\logic{bij}(x,y)\coloneqq \# x'.x'\les x=\#(y'\le\ord).y'\le
  y$. We introduce a distinguished number variable $y_x$ for
  every vertex variable $x$.

  To obtain $\phi'$ from $\phi$, we first replace quantification over
  $x$ in counting terms by quantification over $y_x$, that is, we
  replace $\#(x,\ldots)$ by $\#(y_x<\ord,\ldots)$. Furthermore, we
  replace atomic formulas $x=x'$ by $y_x=y_{x'}$ (or, more formally,
  $y_x\le y_{x'}\wedge y_{x'}\le y_x$) and atomic formulas $R(x,x')$ by
  $\exists x_1.\exists
  x_2.\big(\logic{bij}(x_1,y_x)\wedge\logic{bij}(x_2,y_{x'})\wedge
  R(x_1,x_2))$. Let $\psi$ be the resulting formula. If $\phi$ has no
  free variables, we let $\phi'\coloneqq\psi$.  If $\phi$ has one free
  variable $x$, we let
  \[
    \phi'\coloneqq\exists
    y_{x}<\ord.\big(\logic{bij}(x,y_{x})\wedge\psi\big).
  \]
 Then $\phi'$ is 
  equivalent to $\phi$, and it only contains the vertex variables 
  $x_1,x_2$ and hence is in $\FOC[2]$. It is easy to check that the
  formula is decomposable.
\end{proof}

\begin{myremark}
  Note that Lemma~\ref{lem:fo2fo2} implies that on ordered structures, every
  $\FOC[2]$-formula with at most one free variable is equivalent to a decomposable formula. It is an
  open problem whether this holds on arbitrary structures.

  The definition of decomposable $\FOC[k]$ is not
  particularly intuitive, at least at first glance. However, we wonder
  if ``decomposable $\FOC[k]$'' is what we actually want as the
  $k$-variable fragment of $\FOC$. This view is not only supported by
  Lemma~\ref{lem:fo2fo2}, but also by the observation that the logic
  $\LC^k$ (the $k$-variable fragment of the extension of first-order
  logic by counting quantifiers $\exists^{\ge n}x$) is contained in
  decomposable $\FOC[k]$. Furthermore, characterisations of
  $k$-variable logics in terms of pebble games or the WL-algorithm
  only take atomic properties of $k$-tuples into account.
  \uend
\end{myremark}

The \emph{guarded fragment $\GC$} is a fragment of $\FOC[2]$ where
quantification and counting is restricted to range over neighbours of
a free variable. We fix two variables $x_1,x_2$.  A \emph{guard} is an
atomic formula of the form $R(x_i,x_{3-i})$ for some binary relation
symbol $R$.  We inductively define the sets of \emph{$\GC$-terms} and
\emph{$\GC$-formulas} as follows.
\begin{itemize}
\item
  All number variables and $\zero,\one,\ord$ are $\GC$-terms.
\item For all  $\GC$-terms $\theta,\theta'$, the expressions 
  $
    \theta+ \theta'$ and $\theta\cdot \theta'
  $
  are $\GC$-terms.
\item For all function variables $U$ of type $(t_1,\ldots,t_k)\to\ttn$
  and all tuples $(\xi_1,\ldots,\xi_k)$, where $\xi_i$ is a vertex
  variable if $t_i=\ttv$ and $\xi_i$ is a $\GC$-term if $t_i=\ttn$,
  the expression
  $U (\xi_1,\ldots,\xi_k)$ is a $\GC$-term.
\item For all  $\GC$-terms $\theta,\theta'$, the expression $\theta\le\theta'$
  is a $\GC$-formula.
\item All relational atoms whose variables are among $x_1,x_2$ are
  $\GC$-formulas.
\item For all relation variables $X$ of type
  $\tta(t_1,\ldots,t_k)\ttz$ and all tuples 
  $(\xi_1,\ldots,\xi_k)$, where $\xi_i$ is a vertex
  variable if $t_i=\ttv$ and $\xi_i$ is a $\GC$-term if $t_i=\ttn$,
  the expression 
  $X(\xi_1,\ldots,\xi_k)$ is a $\GC$-formula.
\item  For all $\GC$-formulas $\phi,\psi$ the expressions $\neg\phi$
  and $\phi\wedge\psi$ 
  are $\GC$-formulas.
\item For all $\GC$-formulas $\phi$, guards
    $\gamma$, number variables
  $y_1,\ldots,y_k$, all $\GC$-terms $\theta_1,\ldots,\theta_k$, and
  $i=1,2$, 
  \begin{equation}
    \label{eq:4}
    \#(x_{3-i},y_1<\theta_1,\ldots,y_k<\theta_k).(\gamma\wedge \phi),
  \end{equation}
  is a $\GC$-term.
\item For all $\GC$-formulas $\phi$, number variables
  $y_1,\ldots,y_k$, and $\GC$-terms $\theta_1,\ldots,\theta_k$, 
  \begin{equation}
    \label{eq:5}
    \#(y_1<\theta_1,\ldots,y_k<\theta_k).\phi,
  \end{equation}
  is a $\GC$-term.
\end{itemize}
Observe that a $\GC$-term or $\GC$-formula either has at least one
free vertex variable or contains no vertex variable at
all. Note that
we add $\ord$ as a ``built-in'' constant that is always interpreted by
the order of the input structure. We need access to the order of a
structure to bound quantification on numbers, and the closed
$\FOC[2]$-term $\ord=\#x.x=x$ defining the order is not in
$\GC$.

An r-expression is \emph{guarded} if all its formulas and terms are in
$\GC$.

\begin{myremark}\label{rem:guarded-variant}\label{rem:mc1}Our definition of the guarded fragment is relatively liberal in
  terms of which kind of formulas $\phi$ we allow inside the
  guarded counting operators in \eqref{eq:4}. In particular, we allow
  both $x_i$ and $x_{3-i}$ to occur freely in $\phi$. A more
  restrictive alternative
  definition, more in the spirit of a modal logic, would be to stipulate
  that the variable $x_i$ must not occur freely in $\phi$ and the
  bounding terms $\theta_i$. Let us call
  the restriction of $\GC$ where terms of the form \eqref{eq:4} are
  only allowed for formulas $\phi$ and terms $\theta_1,\ldots,\theta_k$ in which $x_i$ does not occur
  freely the \emph{modal fragment} of $\FOC$, denoted by $\MC$.

  While our main focus will be on $\GC$, we will explain how to adapt
  our main results to $\MC$ in a sequence of remarks. The distinction
  between the modal and guarded fragments is closely related to a
  similar distinction for graph
  neural networks (see Remark~\ref{rem:ac_vs_mpnn}).

  It is proved in the subsequent article \cite{GroheR24} that $\GC$ and
  $\MC$ have the same expressive power. The characterisation theorems
  of this article will be used there to transfer this to graph neural
  networks.\uend
\end{myremark}

By definition, $\GC$ is contained in $\FOC[2]$. The converse does not
hold. Let us introduce an intermediate fragment $\GCgc$ which extends
$\GC$ and is still in $\FOC[2]$. We call $\GCgc$ the \emph{guarded
  fragment with global counting}. In addition to the guarded counting
terms in \eqref{eq:4}, in $\GCgc$ formulas we also allow a restricted
form of unguarded counting in the form
  \begin{equation}
    \label{eq:6}
    \#(x_{3-i},y_1<\theta_1,\ldots,y_k<\theta_k).\phi,
  \end{equation}
  where the variable $x_i$ must not occur freely in
  $\phi$. Intuitively, such a term makes a ``global'' calculation that
  is unrelated to the ``local'' properties of the free
  variable $x_i$.

  Let us call a $\GC$-formula or a $\GCgc$-formula
  \emph{decomposable} if it is decomposable as an $\FOC[2]$-formula.

\begin{lemma}\label{lem:GCgr}\sloppy
  For every decomposable $\FOC[2]$-formula $\phi$ there is
  a decomposable $\GCgc$-formula $\phi'$ such that for all graphs $G$, possibly
  labelled, and all assignments $\Ca$ over $G$ we have
  \[
    (G,\Ca)\models\phi\iff(G,\Ca)\models\phi'.
  \]
\end{lemma}

\begin{proof}
  Throughout this proof, formulas are without relation or function
  variables.  We write $x$ and $x'$ to refer to the variables
  $x_1,x_2$, with the understanding that if $x$ refers to $x_i$ then
  $x'$ refers to $x_{3-i}$.  We need to replace unguarded terms by
  combinations of terms of the form \eqref{eq:4}, \eqref{eq:5}, and
  \eqref{eq:6}.

  \begin{techclaim}\label{claim:b:1}
   Every term $\eta\coloneqq\#(x',y_1<\theta_1,\ldots,y_k<\theta_k).\psi$, where
   $\psi$ is a $\GCgc$-formula and the $\theta_i$ are $\GCgc$-terms,
   is equivalent to a $\GCgc$-term.

  \end{techclaim}
\begin{subproof} By Lemma~\ref{lem:termbound}, we may assume without loss of
   generality that the terms $\theta_i$ do not contain the variables
   $x,x',y_1,\ldots,y_k$ (and therefore the counting operator is
   nonadaptive). Indeed, by the lemma we can find a term $\theta$
   built from 
   $\ord$ and the free number variables of $\eta$
   such that $\theta$ bounds all $\theta_i$. Then $\eta$ is equivalent to 
   $\#(x',y_1<\theta,\ldots,y_k<\theta).(\psi\wedge
   y_1<\theta_1\wedge\ldots\wedge y_k<\theta_k)$.

  Also without loss of generality we may assume that $k\ge 1$,
  because we can always append $y<\one$ for a fresh variable $y$ in
  the counting operator without changing the result. To simplify the
  notation, let us assume that $k=1$. The generalisation to larger $k$
  is straightforward. Hence
  \begin{equation}
    \label{eq:7}
    \eta=\#(x',y<\theta).\psi\\
  \end{equation}
  and $\theta$ is a term in which neither $x$ nor $x'$ is free.

   Since $\psi$ is decomposable, we can re-write $\psi$ as
  \[
    (E(x,x')\wedge\psi_1)\vee(\neg E(x,x')\wedge\neg x=x'\wedge\psi_2)\vee\psi_3,
  \]
  where $\psi_1$ and $\psi_2$ are Boolean combinations of formulas
  with at most one free vertex variable (either $x$ or $x'$), and $x'$
  does not occur freely in $\psi_3$. To see this, note that in graphs
  all vertices $x,x'$
  satisfy exactly one of $E(x,x')$, $\neg E(x,x')\wedge\neg x=x'$, and
  $x=x'$. Thus $\psi$ is equivalent to the disjunction
  \[
    (E(x,x')\wedge\psi) \vee(\neg E(x,x')\wedge\neg
    x=x'\wedge\psi)\vee(x=x'\wedge \psi).
  \]
  Then to obtain $\psi_1$ and $\psi_2$, we eliminate all atomic formulas
  $E(x,x'),x=x'$ in both free variables from $\psi$, and  to obtain
  $\psi_3$ we substitute $x$ for all free occurrence of $x'$. Thus
  $\eta$ is equivalent to the term
  \begin{align}
    \notag
    &\#(x',
      y<\theta).(E(x,x')\wedge\psi_1)\\
    \label{eq:8}
    +\;&\#(x',
       y<\theta).(\neg E'(x,x')\wedge\neg
       x=x'\wedge\psi_2)\\
    \notag
    +\;&\#(y<\theta).\psi_3.
  \end{align}
  The first and the third term in this sum are already $\GCgc$-terms
  of the forms \eqref{eq:4} and \eqref{eq:5}, respectively. We only
  need to worry about the second,
  \[
    \eta_2\coloneqq \#(x',
    y<\theta).(\neg E'(x,x')\wedge\neg
    x=x'\wedge\psi_2).
  \]
  We can equivalently re-write $\eta_2$ as
  \begin{align*}
    &\#(x',
      y<\theta).\psi_2\\
    \dotminus\;&\#(x',
                 y<\theta).(E(x,x')\wedge\psi_2)\\
    \dotminus\;&\#(y<\theta).\psi_2\textstyle\frac{x}{x'},
  \end{align*}
  where $\psi_2\frac{x}{x'}$ denotes the formula obtained from $\psi_2$ by
  replacing all free occurrences of $x'$ by $x$. The second and the third term are already $\GC$-terms. We only
  need to worry about the first,
  \[
    \eta_2'\coloneqq \#(x',
    y<\theta).\psi_2
  \]
  Recall
  that $\psi_2$ is a Boolean combination of formulas with only one
  free vertex variable. Bringing this Boolean combination into
  disjunctive normal form, we obtain an equivalent formula
  $\bigvee_{i\in I}(\chi_{i}\wedge\chi'_{i})$, where $x$ does not
  occur freely in $\chi'_{i}$ and $x'$ does not
  occur freely in $\chi_{i}$. We can further ensure that the disjuncts
  are mutually exclusive, that is, for every pair $x,x'$ there is at
  most one $i$ such that that it satisfies
  $(\chi_{i}\wedge\chi'_{i})$. For example, if $I=\{1,2\}$, we
  note that
  $(\chi_{1}\wedge\chi'_{1})\vee (\chi_{2}\wedge\chi'_{2}) $ is
  equivalent to
  \begin{align*}
    &((\chi_{1}\wedge\chi_{2})\wedge(\chi'_{1}\wedge\chi_{2}'))\\
    \vee &((\chi_{1}\wedge\neg\chi_{2})\wedge(\chi'_{1}\wedge\chi_{2}'))\\
    \vee &((\chi_{1}\wedge\chi_{2})\wedge(\chi'_{1}\wedge\neg\chi_{2}'))\\
    \vee
    &((\chi_{1}\wedge\neg\chi_{2})\wedge(\chi'_{1}\wedge\neg\chi_{2}'))\\
        \vee &((\neg\chi_{1}\wedge\chi_{2})\wedge(\chi'_{1}\wedge\chi_{2}'))\\
        \vee &((\chi_{1}\wedge\chi_{2})\wedge(\neg \chi'_{1}\wedge\chi_{2}'))\\
        \vee &((\neg\chi_{1}\wedge\chi_{2})\wedge(\neg
               \chi'_{1}\wedge\chi_{2}')).
  \end{align*}
  Then $\eta_2'$ is equivalent to the term
  \[
    \sum_{i\in I}\#(x',
    y<\theta).(\chi_{i}\wedge\chi'_i).
  \]
  Consider a summand
  $\eta_{2,i}\coloneqq \#(x',
  y<\theta).(\chi_{i}\wedge\chi'_{i})$.
  Let $\zeta\coloneqq \# x'.(\chi_{i}\wedge\chi'_{i})$ and note that
  for all graphs $G$ and assignments $\Ca$ we have
  \[
    \sem{\eta_{2,i}}^{(G,\Ca)}=\sum_{b<\sem{\theta}^{(G,\Ca)}}\sem{\zeta}^{(G,\Ca\frac{b}{y})}=\big|\big\{(b,c)\bigmid
    b<\sem{\theta}^{(G,\Ca)},c<\sem{\zeta}^{(G,\Ca\frac{b}{y})}\big\}\big|.
  \]
  Let $\eta'_{2,i}\coloneqq\#(y<\theta,z<\ord).z<\zeta$. Since we
  always have $\sem{\zeta}^{(G,\Ca)}\le|G|$, the terms $\eta_{2,i}$ and
  $\eta_{2,i}'$ are equivalent.

  The final step is to turn $\zeta$ into
  a $\GCgc$-term. Recall that $\zeta=\# x'.(\chi_{i}\wedge\chi'_{i})$
  and that
  $x'$ is not free in $\chi_i$. If $x$ does not satisfy $\chi_{i}$, then the term
  $\zeta$ evaluates to $0$, and otherwise it has the same value as the
  term $\# x'.\chi'_{i}$, which is of the form \eqref{eq:6}.
  Note that the term
  $\#(y'<\one).\chi_{i}$, where $y'$ is a fresh number variable not
  occurring in $\chi_{i}$, is of the form \eqref{eq:5} and evaluates
  to $1$ if $\chi_{i}$ is satisfied and to $0$ otherwise.  Thus the
  term 
  \[
    \#(y'<\one).\chi_{i}\;\cdot\; \#x'.\chi_{i}'
  \]
  is a $\GCgc$-term equivalent to $\zeta$. This completes the
  proof of the claim.
  \uend
\end{subproof}

 \begin{techclaim}\label{claim:b:2}
   Every term $\eta\coloneqq\#(x,x',y_1<\theta_1,\ldots,y_k<\theta_k).\psi$, where
   $\psi$ is a $\GCgc$-formula and the $\theta_i$ are $\GCgc$-terms,
   is equivalent to a $\GCgc$-term.

   \end{techclaim}
\begin{subproof}
   Arguing as in the proof of Claim~\ref{claim:b:1}, we may assume that 
   \begin{equation}   
    \label{eq:9}
    \eta=\#(x,x',y<\theta).\psi,
  \end{equation}
  where $\theta$ is a term in which the variables $x,x'$ are not free.

  We proceed very similarly to the proof of Claim~\ref{claim:b:1}. The first step is to rewrite the term as
    the sum
  \begin{align}
    \notag
    &\#(x,x',
      y<\theta).(E(x,x')\wedge\psi_1)\\
    \label{eq:10}
    +\;&\#(x,x',
       y<\theta).(\neg E'(x,x')\wedge\neg
       x=x'\wedge\psi_2)\\
    \notag
    +\;&\#(x,y<\theta).\psi_3,
  \end{align}
  where $\psi_1$ and $\psi_2$ are Boolean combinations of formulas
  with at most one free vertex variable and $x'$ is not free in
  $\psi_3$. Then the third term is already of the form \eqref{eq:6},
  and we only have to deal with the first two.
  
  Let us look at the term
  $\eta_1\coloneqq \#(x,x',
  y<\theta).(E(x,x')\wedge\psi_1)$. Let
  $\zeta_1\coloneqq\#(x',
  y<\theta).(E(x,x')\wedge\psi_1)$. Then
  $\zeta_1$ is a term of the form \eqref{eq:4}. Moreover, for every
  graph $G$ and assignment $\Ca$ we have
  \begin{align*}
    \sem{\eta_1}^{(G,\Ca)}&=\sum_{a\in
                            V(G)}\sem{\zeta_1}^{(G,\Ca\frac{a}{x})}=\big|\big\{(a,c)\bigmid a\in V(G),c<\sem{\zeta_1}^{(G,\Ca\frac{a}{x})}\big\}\big|.
  \end{align*}
  We let
  \[
    \eta_1'\coloneqq\#(x,z<\ord).z<\zeta_1.
  \]
  Then $\eta_1'$ is a term of the form \eqref{eq:6} that is equivalent
  to $\eta_1$.

  It remains to deal with the second summand in \eqref{eq:10}, the
  term 
  \[
  \eta_2\coloneqq\#(x,x',
  y<\theta).(\neg E'(x,x')\wedge\neg
  x=x'\wedge\psi_2).
  \]
  We rewrite this term as
  \begin{align}
    \label{eq:11}
    &\#(x,x',
      y<\theta).\psi_2\\
     \label{eq:12}
    \dotminus \;&\#(x,x',
      y<\theta).(E(x,x')\wedge\psi_2)\\
     \label{eq:13}
    \dotminus \;&\#(x,
      y<\theta).\psi_2\textstyle\frac{x}{x'}.
  \end{align}
  The term \eqref{eq:13} is of the form \eqref{eq:6}, and we have
  just seen how to deal with a term of the form \eqref{eq:12}. Thus
  it remains to deal with the first term $\eta_2'\coloneqq\#(x,x',
      y<\theta).\psi_2$. As in the proof of Claim~\ref{claim:b:1}, we can find an equivalent formula 
  $\bigvee_{i\in I}(\chi_{i}\wedge\chi'_{i})$, where $x$ does not
  occur freely in $\chi'_{i}$ and $x'$ does not
  occur freely in $\chi_{i}$, and the disjuncts
  are mutually exclusive. Then $\eta_2'$ is equivalent to the sum
  \[
    \sum_{i\in I}\#(x,x',
    y<\theta).(\chi_i\wedge\chi_i').
  \]
  Consider one of the terms in the sum, $\eta_{2,i}'\coloneqq \#(x,x',
  y<\theta).(\chi_i\wedge\chi_i')$. 
We let $\zeta\coloneqq \# (x,x').(\chi_{i}\wedge\chi'_{i})$  As in the proof of Claim~\ref{claim:b:1},
  for all graphs $G$ and assignments $\Ca$ we have
  \[
    \sem{\eta_{2,i}'}^{(G,\Ca)}=\sum_{b<\sem{\theta}^{(G,\Ca)}}\sem{\zeta}^{(G,\Ca\frac{b}{y})}=\big|\big\{(b,c)\bigmid
    b<\sem{\theta}^{(G,\Ca)},c<\sem{\zeta}^{(G,\Ca)}\big\}\big|.
  \]
  Let $\eta''_{2,i}\coloneqq\#(y<\theta,z<\ord\cdot\ord).z<\zeta$. Since we
  always have $\sem{\zeta}^{(G,\Ca)}\le|G|^2$, the terms $\eta_{2,i}'$ and
  $\eta_{2,i}''$ are equivalent.

  To turn $\zeta$ into
  a $\GCgc$-term $\zeta'$, we observe that for all graphs $G$ and assignments
  $\Ca$ we have
  \[
    \sem{\zeta}^{(G,\Ca)}=\big|\big\{ a\in V(G)\bigmid
    (G,\Ca\textstyle\frac{a}{x})\models\chi_i\big\}\big|\cdot
    \big|\big\{ a'\in V(G)\bigmid
    (G,\Ca\textstyle\frac{a'}{x'})\models\chi_i'\big\}\big|.
  \]
  We let $\zeta'\coloneqq\#x.\chi_i\cdot\# x'.\chi_i'$.
  \uend
\end{subproof}

\medskip
With these two claims, it is easy to inductively translate a
decomposable $\FOC[2]$-formula into an $\GCgc$-formula.
\end{proof}

Combining Lemmas~\ref{lem:GCgr} and \ref{lem:fo2fo2} with
Corollary~\ref{cor:TC0FO2}, we obtain the following.

 \begin{corollary}\label{cor:GCgc_TC0}
   Let $\CQ$ be a unary query. Then $L_\les(\CQ)$ is in $\TC^0$ if
   and only if $\CQ$ is definable in order-invariant $\GCgcnu$.
\end{corollary}

\subsection{Arithmetic in \texorpdfstring{$\GC$}{GFO+C}}
\label{sec:GCarithmetic}

Since all arithmetical $\FOC$-formulas and terms are in $\GC$, most results of
Sections~\ref{sec:arithmetic}--\ref{sec:fnnsim} apply to
$\GC$. Exceptions are Lemmas~\ref{lem:ar3a}, \ref{lem:ar3}, and
\ref{lem:ar5} on iterated addition, which may involve vertex
variables. Here we prove variants of these lemmas for the guarded
fragment.

\begin{lemma}\label{lem:ar8}
  Let $X$ be a relation variable of type $\tta\ttv^2\ttn^\ell\ttz$,
  and let $U,V$ be function variables of types
  $\ttv^2\ttn^\ell\to\ttn$, $\ttv^2\to\ttn$, respectively.
  Furthermore, let 
  $\gamma(x,x')$ be a guard.
  \begin{enumerate}
  \item There is a $\GC$-term $\logic{u-itadd}(x)$ such that for all
    structures $A$, assignments $\Ca$ over $A$, and $a\in V(A)$,
    \[
      \sem{\logic{u-itadd}}^{(A,\Ca)}(a)=\sum_{(a',\vec
        b)}\Ca(U)(a,a',\vec b),
    \]
    where the sum ranges over all $(a',\vec b)\in V(A)\times\Nat^\ell$
    such that $A\models\gamma(a,a')$ and $(a,a',\vec b)\in\Ca(X)$ and $\vec b=(b_1,\ldots, b_\ell)\in\Nat^\ell$ with
    $b_i<\Ca(V)(a,a')$ for all $i\in[\ell]$. 
    \item
      There are $\GC$-terms $\logic{u-max}(x)$ and $\logic{u-min}(x)$ such that for
  all structures $A$, assignments $\Ca$ over $A$, and $a\in
  V(A)$, 
    \begin{align*}
      \sem{\logic{u-max}}^{(A,\Ca)}(a)&=\max_{(a',\vec
                                        b)}\Ca(U)(a,a',\vec b),\\
            \sem{\logic{u-min}}^{(A,\Ca)}(a)
      &=\min_{(a',\vec
                                        b)}\Ca(U)(a,a',\vec b),
    \end{align*}
    where $\max$ and $\min$ range over all $(a',\vec b)\in V(A)\times\Nat^\ell$
    such that $A\models\gamma(a,a')$ and $(a,a',\vec b)\in\Ca(X)$ and $\vec b=(b_1,\ldots, b_\ell)\in\Nat^\ell$ with
    $b_i<\Ca(V)(a,a')$ for all $i\in[\ell]$. 
    \end{enumerate}
\end{lemma}

\begin{proof}
  The proof is very similar to the proof of Lemma~\ref{lem:ar3a}, we
  just have to make sure that the terms we define are guarded. Again,
  we assume for simplicity that $\ell=1$.

  We let
  \[
    \logic{u-itadd}(x)\coloneqq\#\big(x',y< V(x,x'),z< U(x,x',y)\big).\big(\gamma(x,x')\wedge X(x,x',y)\big)
  \]
  and
  \begin{align*}
    \logic{u-max}(x)&\coloneqq \# z<\logic{u-itadd}(x).\exists
                   (x',y<V(x,x')).\big(\gamma(x,x')\wedge X(x,x',y)\wedge z< U(x,x')\big),\\
    \logic{u-min}(x)&\coloneqq \# z<\logic{u-itadd}(x).\forall
                      (x',y<V(x,x')).\big(\gamma(x,x')\wedge X(x,x',y)\to z< U(x,x')\big).
  \end{align*}
\end{proof}

\begin{lemma}\label{lem:ar7}
  Let $\vec Z$ be an r-schema of type $\rtp{\ttv^2\ttn^\ell}$.
  Let $X$ be a relation variable of type $\tta\ttv^2\ttn^\ell\ttz$,
  and let $V$ be a function variable of type
  $\ttv^2\to\ttn$. Furthermore, let
  $\gamma(x,x')$ be a guard.
  \begin{enumerate}
  \item
    There is a guarded r-expression $\logic{itadd}(x)$ such that  for all structures $A$, assignments $\Ca$
    over $A$, and $a\in V(A)$ we have
      \[
    \num{\logic{itadd}}^{(A,\Ca)}(a)=
      \sum_{(a',\vec b)}\num{\vec Z}^{(A,\Ca)}(a,a',\vec b),
    \]
    where the sum ranges over all $(a',\vec b)\in V(A)\times\Nat^\ell$
    such that $A\models\gamma(a,a')$ and $(a,a',\vec b)\in\Ca(X)$ and
    $\vec b=(b_1,\ldots, b_\ell)\in\Nat^\ell$ with $b_i<\Ca(V)(a,a')$
    for all $i\in[\ell]$.
  \item
    There are guarded r-expressions $\logic{max}(x)$ and $\logic{min}(x)$ such that  for all structures $A$, assignments $\Ca$
    over $A$, and $a\in V(A)$ we have
    \begin{align*}
      \num{\logic{max}}^{(A,\Ca)}(a)&=
      \max_{(a',\vec b)}\num{\vec Z}^{(A,\Ca)}(a,a',\vec b),
      \\
      \num{\logic{min}}^{(A,\Ca)}(a)&=
      \min_{(a',\vec b)}\num{\vec Z}^{(A,\Ca)}(a,a',\vec b),
    \end{align*}
    where $\max$ and $\min$ range over all $(a',\vec b)\in V(A)\times\Nat^\ell$
    such that $A\models\gamma(a,a')$ and $(a,a',\vec b)\in\Ca(X)$ and
    $\vec b=(b_1,\ldots, b_\ell)\in\Nat^\ell$ with $b_i<\Ca(V)(a,a')$
    for all $i\in[\ell]$.
    \end{enumerate}
\end{lemma}

\begin{proof}
  The proof is an easy adaptation of the proof of Lemma~\ref{lem:ar5},
  arguing as in the proof of Lemma~\ref{lem:ar8} to make sure that we
  obtains guarded formulas and terms.
\end{proof}

\begin{myremark}\label{rem:mc2}Lemmas~\ref{lem:ar8} and \ref{lem:ar7} have analogues for the modal
  fragment $\MC$. For  Lemma~\ref{lem:ar8}, we let $X$ be a relation
  variable of type  $\tta\ttv\ttn^\ell\ttz$
  and $U,V$ function variables of types
  $\ttv\ttn^\ell\to\ttn$, $\ttv\to\ttn$, respectively. Then in the
  assertions, we always omit $a$ from the argument lists. For example,
  statement (1) becomes:
  {\itshape
    There is an $\MC$-term $\logic{u-itadd}(x)$ such that for all
    structures $A$, assignments $\Ca$ over $A$, and $a\in V(A)$,
    \[
      \sem{\logic{u-itadd}}^{(A,\Ca)}(a)=\sum_{(a',\vec
        b)}\Ca(U)(a',\vec b),
    \]
    where the sum ranges over all $(a',\vec b)\in V(A)\times\Nat^\ell$
    such that $A\models\gamma(a,a')$ and $(a',\vec b)\in\Ca(X)$ and $\vec b=(b_1,\ldots, b_\ell)\in\Nat^\ell$ with
    $b_i<\Ca(V)(a')$ for all $i\in[\ell]$.
  }
  The modeification for assertion (2) is similar.

  For Lemma~\ref{lem:ar7}, we let $\vec Z$ be an r-schema of type
  $\rtp{\ttv\ttn^\ell}$, $X$ a relation variable of type $\tta\ttv\ttn^\ell\ttz$,
  and $V$ be a function variable of type
  $\ttv\to\ttn$. The modification of the assertions is similar to the
  previous lemma.

  For both lemmas, the adaptation of the proof is straightforward.
  \uend
\end{myremark}

\section{Graph Neural Networks}
\label{sec:gnn}
We will work with standard message passing graph neural networks
(GNNs\footnote{We use the abbreviation GNN, but MPNN is also very common.})
\cite{GilmerSRVD17}.
A GNN consists of a finite sequence of layers.
A \emph{GNN layer} of input dimension $p$ and output dimension $q$  is
a triple $\FL=(\msg,\agg,\comb)$ of functions: a \emph{message function}
$\msg:\Real^{2p}\to\Real^{p'}$, an \emph{aggregation function}
$\agg$ mapping finite multisets of vectors in $\Real^{p'}$ to vectors
in $\Real^{p''}$, and a \emph{combination function}
$\comb:\Real^{p+p''}\to\Real^q$.
A \emph{GNN} is a tuple $\FN=(\FL^{(1)},\ldots,\FL^{(d)})$ of GNN layers, where the output dimension $q^{(i)}$ of $\FL^{(i)}$
matches the input dimension $p^{(i+1)}$ of $\FL^{(i+1)}$. We call
$q^{(0)}\coloneqq p^{(1)}$ the input dimension of $\FN$ and
$q^{(d)}$ the output dimension. 

To define the semantics, let $\FL=(\msg,\agg,\comb)$ be a GNN layer of input dimension $p$
and output dimension $q$. It computes a function $\FL\colon
\CGS_p\to\CGS_q$ (as for circuits and feedforward neural networks, we
use the same letter to denote the network and the function it
computes) defined by $\FL(G,\Cx)\coloneqq(G,\Cy)$, where
$\Cy:V(G)\to\Real^q$ is defined by 
\begin{equation}
  \label{eq:15}
    \Cy(v)\coloneqq \comb\Bigg(\Cx(v),\agg\Big(\Biglmulti \msg\big(\Cx(v),\Cx(w)\big)\Bigmid w\in N_G(v)\Bigrmulti\Big)\Bigg).
  \end{equation}
A GNN $\FN=(\FL^{(1)},\ldots,\FL^{(d)})$ composes the transformations
computed by its layers $\FL^{(i)}$, that is, it computes the function
$\FN\colon \CGS_{q^{(0)}}\to\CGS_{q^{(d)}}$ defined by
\[
  \FN(G,\Cx)\coloneqq \FL^{(d)}\circ \FL^{(d-1)}\circ\ldots\circ
  \FL^{(1)}.
\]
It will be convenient to also define $\tilde{\FN}$ as the function
mapping $(G,\Cx)$ to the signal $\Cx'\in\CS_{q^{(d)}}(G)$ such that 
$\FN(G,\Cx)=(G,\Cx')$, so $\FN(G,\Cx)=(G,\tilde\FN(G,\Cx))$, and
similarly $\tilde{\FL}$ for a single layer $\FL$.

\begin{myremark}\label{rem:ac_vs_mpnn}\label{rem:mc3}Our version of GNNs corresponds to the \emph{message passing neural
    networks} due to \cite{GilmerSRVD17}. In \cite{GroheR24}, these
  GNNs are called \emph{2-GNNs}, to distinguish them from another
  common version of message passing graph neural networks, cleanly
  formalised as the \emph{aggregate-combine GNNs} in
  \cite{BarceloKM0RS20}, and called \emph{1-GNNs} in
  \cite{GroheR24}. In these 1-GNNs, messages
  only depend on the vertex they are sent from, so the update rule
  \eqref{eq:15} becomes
  \begin{equation}
    \label{eq:101}
    \Cy(v)\coloneqq \comb\Bigg(\Cx(v),\agg\Big(\Biglmulti \msg\big(\Cx(w)\big)\Bigmid w\in N_G(v)\Bigrmulti\Big)\Bigg).
  \end{equation}
  Whenever we want to make the distinction between the two versions
  explicit, we use the 1-GNN / 2-GNN terminology, but most of the time
  we will just work with 2-GNNs and simply call them GNNs.  The reason
  I decided to focus on (2-)GNNs is that in practical work we also
  found it beneficial to use them. 

  However, we will explain how to adapt our results to 1-GNNs in a
  series of remarks leading to Theorem~\ref{theo:converse-modal}. On the logical side, 2-GNNs correspond to the
  guarded fragment $\GC$ and $1$-GNNs to the modal fragement $\MC$ of
  first-order logic with counting
  (see Remark~\ref{rem:guarded-variant}).

  The relation between 1-GNNs
  and 2-GNNs is more complicated than one might think; we refer the
  reader to \cite{GroheR24}.
  \uend
\end{myremark}

So far, we have defined GNNs as an abstract computation model
computing transformations between graph signals. To turn them into
deep learning models, we represent the functions that specify the
layers by feedforward neural networks. More precisely, we assume that
the message functions $\msg$ and the combination functions $\comb$ of
all GNN layers are specified by FNNs $\FF_\msg$ and
$\FF_{\comb}$. Furthermore, we assume that the aggregation function
$\agg$ is summation $\SUM$, arithmetic mean $\MEAN$, or
maximum $\MAX$. Note that this means that the aggregation function
does not change the dimension, that is, we always have $p'=p''$
(referring to the description of GNN layers above).
To be able to deal with isolated nodes as well, we
define $\SUM(\emptyset)\coloneqq\MEAN(\emptyset)\coloneqq\MAX(\emptyset)\coloneqq\vec0$.

If the FNNs $\FF_\msg$ and $\FF_{\comb}$ on all layers are (rational)
piecewise linear, we call the GNN \emph{(rational) piecewise
  linear}. Similarly, if they are rpl-approximable, we call the GNN
\emph{rpl-approximable}.

We mention a few extensions of the basic GNN model. Most importantly,
in a GNN with \emph{global readout} \cite{BarceloKM0RS20} (or,
equivalently, a GNN with a \emph{virtual node} \cite{GilmerSRVD17}) in
each round the nodes also obtain the aggregation of the states of all
nodes in addition to the messages they receive from their
neighbours. So the state update rule \eqref{eq:15} becomes
\begin{align*}
  \Cy(v)\coloneqq \comb\bigg(\Cx(v),
  &\,\agg\Big(\Biglmulti
    \msg\big(\Cx(v),\Cx(w)\big)\Bigmid w\in N_G(v)\Bigrmulti\Big),
  \\
  &\,\agg'\Big(\Biglmulti \Cx(w)\Bigmid w\in
  V(G)\Bigrmulti\Big)\bigg).
\end{align*}
We could also apply some function $\msg'$ to the $\Cx(w)$ before
aggregating, but this would not change the expressiveness, because we
can integrate this into the combination function.

To adapt GNNs to directed graphs, it is easiest to use separate
message functions for in-neighbours and out-neighbours and to
aggregate them separately and then combine both in the combination
function. In graphs with edge labels, or \emph{edge signals} $\Cy\colon E(G)\to\Real^k$ for
some $k$, we can give these signals as a third
argument to the message function, so the message function becomes
$\msg\big(\Cx(v),\Cx(w),\Cy(v,w)\big)$. We can also adapt this to directed
edge-labeled graphs and hence to arbitrary binary relational
structures. 
Often, we want to use GNNs to compute \emph{graph-level} functions
$\CGS_p\to\Real^q$ rather than \emph{node-level} functions
$\CGS_p\to\CGS_q$. For this, we aggregate the values of the output
signal at the nodes to a single value. A \emph{graph-level GNN}
is a triple $\FG=(\FN,\agg,\ro)$ consisting of a GNN $\FN$, say with input dimension $p$ and output
dimension $p'$, an \emph{aggregate function} $\agg$, which we assume
to be either $\SUM$ or $\MEAN$ or $\MAX$, and a
\emph{readout function} $\ro:\Real^{p'}\to\Real^q$, which we assume to
be computed by an FNN $\FF_{\ro}$.

\emph{All of our results have straightforward extensions to all these
  variants of the basic model.} Since the article is lengthy and
technical as it is, I decided to focus just on the basic model
here. Occasionally, I comment on some of the extensions, pointing
out which modifications on the logical side need to be made.

\subsection{Useful Bounds}

\begin{lemma}\label{lem:gnnbound1}
  Let $\FL$ be a GNN layer of input dimension $p$. Then there is a
   $\gamma\in\PNat$ such that for
  all graphs $G$, all signals $\Cx\in\CS_p(G)$, and all vertices $v\in V(G)$ we have
  \begin{align}
     \label{eq:16}
    \inorm{\tilde\FL(G,\Cx)(v)}&\le\gamma\cdot\Big(\inorm{\Cx|_{N[v]}}+1\Big)\max\{\deg(v),1\}\\
    \label{eq:17}
                               &\le\gamma\cdot\big(\inorm{\Cx}+1\big)|G|.
  \end{align}
\end{lemma}

Recall that $\Cx|_{N[v]}$ denotes the restriction of a signal
$\Cx$ to the closed neighbourhood $N[v]$ of $v$, and we have
$\inorm{\Cx|_{N[v]}}=\max_{w\in N[v]}\inorm{\Cx(w)}$.
The bound \eqref{eq:16} is \emph{local}, it only depends on the
neighbourhood of $v$. The \emph{global} bound
\eqref{eq:17} is simpler, but  a bit weaker.

\begin{proof}
  Clearly, \eqref{eq:16} implies \eqref{eq:17}, so we only have to
  prove the local bound \eqref{eq:16}.  Let $\msg$, $\agg$, $\comb$ be the
  message, aggregation, and combination functions of $\FL$. Let
  $\FF_{\msg}$ and $\FF_{\comb}$ be FNNs computing $\msg$ and $\comb$,
  respectively, and let $\gamma_{\msg}\coloneqq\gamma(\FF_{\msg})$ and
  $\gamma_{\comb}\coloneqq\gamma(\FF_{\comb})$ be the constants of
  Lemma~\ref{lem:fnngrowth}(2).
  
  Let $G$ be a graph, $\Cx\in\CS_p(G)$, and $v\in V(G)$. Then for all $w\in N_G(v)$ we have
  \[
    \inorm{\msg\big(\Cx(v),\Cx(w)\big)}\le\gamma_\msg\cdot\Big(\inorm{\big(\Cx(v),\Cx(w)\big)}+1\Big)
    \le\gamma_{\msg}\cdot\big(\inorm{\Cx|_{N[v]}}+1\big)
    .
  \]
  Since for every multiset $M$ we have $\agg(M)\le |M|\cdot m$, where
  $m$ is the the maximum absolute value of the entries of $M$, it
  follows that
  \[
    \Cz(v)\coloneqq\agg\Big(\Biglmulti
    \msg\cdot\big(\Cx(v),\Cx(w)\big)\Bigmid w\in N^G(v)\Bigrmulti\Big)\le
    \gamma_\msg\big(\inorm{\Cx|_{N[v]}}+1\big)\deg(v).
    \]
    Since $\gamma_\msg\ge 1$, this implies
    \[
      \inorm{\big(\Cx(v),\Cz(v)\big)}\le
      \gamma_\msg\cdot\big(\inorm{\Cx|_{N[v]}}+1\big)\max\{\deg(v),1\}.
    \]
    Hence
    \begin{align*}
      \inorm{\tilde{\FL}(G,\Cx)(v)}&=\inorm{\comb\Big(\big(\Cx(v),\Cz(v)\big)\Big)}\\
                    &\le \gamma_{\comb}\cdot\big(\inorm{\big(\Cx(v),\Cz(v)\big)}+1\big)\\
                    &\le \gamma_{\comb}\cdot\big(\gamma_\msg\big(\inorm{\Cx|_{N[v]}}+1\big)\max\{\deg(v),1\}+1\big)\\
      &\le 2\gamma_{\comb}\gamma_{\msg}\cdot(\inorm{\Cx|_{N[v]}}+1) \max\{\deg(v),1\}.
    \end{align*}
    We let $\gamma\coloneqq 2\gamma_{\comb}\gamma_{\msg}$.
\end{proof}

\begin{lemma}\label{lem:gnnbound2}
  Let $\FL$ be a GNN layer of input dimension $p$. Then there is a $\lambda\in\PNat$ such that for all
  graphs $G$, all signals $\Cx,\Cx'\in\CS_p(G)$, and all vertices $v\in V(G)$
  we have
  \begin{align}
    \label{eq:18}
    \inorm{\tilde\FL(G,\Cx)(v)-\tilde\FL(G,\Cx')(v)}&\le
                                                      \lambda\inorm{\Cx|_{N[v]}-\Cx'|_{N[v]}}\max\{\deg(v),1\}\\
    \label{eq:19}
    &\le \lambda\inorm{\Cx-\Cx'}|G|.
  \end{align}
\end{lemma}

\begin{proof}
  Again, the local bound \eqref{eq:18} implies the global bound
  \eqref{eq:19}. So we only need to prove \eqref{eq:18}.
  Let
  $\msg,\agg,\comb$ be the message, aggregation, and combination
  functions of $\FL$. Let $\FF_{\msg}$ and $\FF_{\comb}$ be FNNs
  computing $\msg$ and $\comb$, respectively, and let
  $\lambda_{\msg}\coloneqq\lambda(\FF_{\msg})$ and
  $\lambda_{\comb}\coloneq\lambda(\FF_{\comb})$ be their Lipschitz constants
  (from Lemma~\ref{lem:fnngrowth}(1)).

    Let $G$ be a graph, $\Cx,\Cx'\in\CS_p(G)$, and
  $\Cy\coloneqq \tilde\FL(G,\Cx)
  $,
  $\Cy'\coloneqq \tilde\FL(G,\Cx')$. Let $v\in V(G)$. 
  For all $w\in N(v)$ we have
  \begin{align*}
    \inorm{\msg(\Cx(v),\Cx(w))-\msg(\Cx'(v),\Cx'(w))}&
    \le\lambda_\msg\inorm{(\Cx(v),\Cx(w))-(\Cx'(v),\Cx'(w))}.
  \end{align*}
  Thus for
  \begin{align*}
    \Cz(v)&\coloneqq\agg\Big(\biglmulti \msg(\Cx(v),\Cx(w))
            \bigmid w\in N_G(v)\bigrmulti\Big),\\
    \Cz'(v)&\coloneqq\agg\Big(\biglmulti \msg(\Cx'(v),\Cx'(w))
             \bigmid w\in N_G(v)\bigrmulti\Big)
  \end{align*}
  we have
  \begin{equation}
    \label{eq:20}
    \inorm{\Cz(v)-\Cz'(v)}\le \lambda_\msg \inorm{\Cx|_{N[v]}-\Cx'|_{N[v]}}\deg(v).
  \end{equation}
  It follows that
  \begin{align}
    \notag
    \inorm{\Cy(v)-\Cy'(v)}&=\inorm{\comb(\Cx(v),\Cz(v))-\comb(\Cx(v),\Cz'(v)}\\
    \notag
                          &\le\lambda_{\comb}\inorm{(\Cx(v),\Cz(v))-(\Cx'(v),\Cz'(v))}\\
    \notag
    &\le
      \lambda_{\comb}\max\Big\{\inorm{\Cx(v)-\Cx'(v)},\inorm{\Cz(v)-\Cz'(v)}\Big\}\\
    \label{eq:21}
    &\le \lambda_{\comb}\lambda_{\msg}\inorm{\Cx|_{N[v]}-\Cx'|_{N[v]}}\max\{\deg(v),1\}.
  \end{align}
  This implies the assertion of the lemma for
  $\lambda\coloneqq\lambda_\comb\lambda_\msg$.
\end{proof}

\section{The Uniform Case: GNNs with Rational Weights}
\label{sec:uniform}
In this section, we study the descriptive complexity of rational
piecewise linear GNNs. The following theorem, which is the main result
of this section, states that the signal transformations computed by
rational piecewise linear GNNs can be approximated arbitrarily closely by \GC-formulas and
terms.

\begin{theorem}\label{theo:uniform}
  Let $\FN$ be a rational piecewise linear GNN of input dimension $p$
  and output dimension $q$. Let $\vec X_1,\ldots,\vec X_p$ be
  r-schemas of type $\rtp{\ttv}$, and let $W$ be a function variable
  of type $\ttv\to\ttn$.  
Then there are guarded
  r-expressions $\logic{gnn-eval}_1(x),\ldots, \logic{gnn-eval}_q(x)$ such that the following holds for all graphs $G$ and assignments
  $\Ca$ over $G$. Let $\Cx\in\CS_p(G)$ be the signal defined by
  \begin{equation}
    \label{eq:22}
    \Cx(v)\coloneqq\Big(\num{\vec X_1}^{(G,\Ca)}(v),\ldots,
    \num{\vec X_p}^{(G,\Ca)}(v)\Big),
  \end{equation}
  and let $\Cy=\tilde\FN(G,\Cx)$. Then for
  all $v\in V(G)$, 
  \begin{equation}
    \label{eq:23}
     \inorm{\Cy(v)
    -\big(\num{\logic{gnn-eval}_1}^{(G,\Ca)}(v),\ldots, \num{\logic{gnn-eval}_q}^{(G,\Ca)}(v)\big)}
    \le 2^{-\Ca(W)(v)}.
  \end{equation}
\end{theorem}

The main step in the proof of the theorem is the following lemma,
which is the analogue of the theorem for a single GNN layer.

\begin{lemma}\label{lem:uniform}
    Let $\FL$ be a rational piecewise linear GNN layer of input dimension $p$
  and output dimension $q$. Let $\vec X_1,\ldots,\vec X_p$ be
  r-schemas of type $\rtp{\ttv}$, and let $W$ be a function variable
  of type $\ttv\to\ttn$.  Then there are guarded
  r-expressions $\logic{l-eval}_1(x),\ldots, \logic{l-eval}_q(x)$ such
  that the following holds for all graphs $G$ and assignments
  $\Ca$ over $G$.
  Let $\Cx\in\CS_p(G)$ be the signal defined by
  \begin{equation}
    \label{eq:24}    \Cx(v)\coloneqq\Big(\num{\vec X_1}^{(G,\Ca)}(v),\ldots,
    \num{\vec X_p}^{(G,\Ca)}(v)\Big)
  \end{equation}
  and let $\Cy\coloneqq\tilde\FL(G,\Cx)$. Then for
  all $v\in V(G)$, 
  \begin{equation}
    \label{eq:25}
     \inorm{\Cy(v)
    -\big(\num{\logic{l-eval}_1}^{(G,\Ca)}(v),\ldots, \num{\logic{l-eval}_q}^{(G,\Ca)}(v)\big)}
    \le2^{-\Ca(W)(v)} 
  \end{equation}
\end{lemma}

\begin{proof}
  For the presentation of the proof it will be easiest to fix a graph
  $G$ and an assignment $\Ca$ over $G$, though of course the formulas
  and terms we shall define will not depend on this graph and
  assignment.  Let $\Cx\in\CS_p(G)$ be the signal defined in
  \eqref{eq:24}, and let $\Cy\coloneqq\tilde\FL(G,\Cx)$. Let
  $\msg:\Real^{2p}\to\Real^r$, $\agg$, and
  $\comb:\Real^{p+r}\to\Real^q$ be the message, aggregation, and
  combination functions of $\FL$. Let $\FF_{\msg}$ and $\FF_{\comb}$
  be rational piecewise linear FNNs computing $\msg,\comb$,
  respectively, and recall that $\agg$ is either $\SUM$, $\MEAN$, or
  $\MAX$ defined on finite multisets of vectors in $\Real^{r}$. Let
  $\lambda\in\PNat$ be a Lipschitz constant for $\comb$.

  \begin{techclaim}\label{claim:c:1}
    There are guarded r-expressions
    $\boldsymbol{\mu}_1(x,x'),\ldots,\boldsymbol{\mu}_r(x,x')$ such that for all $v,v'\in V(G)$ 
    \begin{align*}
      &\msg\Big(\num{\vec X_1}^{(G,\Ca)}(v),\ldots, \num{\vec
        X_{p}}^{(G,\Ca)}(v), \num{\vec X_1}^{(G,\Ca)}(v'),\ldots, \num{\vec
        X_{p}}^{(G,\Ca)}(v')\Big)\\
      &\hspace{6.5cm}=\Big(\num{\boldsymbol{\mu}_1}^{(G,\Ca)}(v,v'),\ldots,
      \num{\boldsymbol{\mu}_r}^{(G,\Ca)}(v,v')\Big)
    \end{align*}

   \end{techclaim}
\begin{subproof}
    This follows from Corollary~\ref{cor:ar6a} applied to
    $\FF_{\msg}$. The expressions we obtain are guarded, because we
    just substitute atoms containing the variables $x$ or $x'$ in the
    arithmetical formulas we obtain from Corollary~\ref{cor:ar6a}, but
    never quantify over vertex variables.
    \uend
  \end{subproof}

  \begin{techclaim}\label{claim:c:2}
    Let $\vec Z_1,\ldots,\vec Z_{r}$ be r-schemas of type
    $\rtp{\ttv}$. Then there  are guarded r-expressions
    $\boldsymbol{\gamma}_1(x),\ldots, \boldsymbol{\gamma}_q(x)$ such that for all assignments $\Ca'$
    over $G$,
    \begin{align*}
      &\comb\Big(\num{\vec X_1}^{(G,\Ca)}(v),\ldots, \num{\vec
        X_{p}}^{(G,\Ca)}(v), \num{\vec Z_1}^{(G,\Ca)}(v),\ldots, \num{\vec
        Z_{r}}^{(G,\Ca)}(v)\Big)\\
      &\hspace{6.5cm}=\Big(\num{\boldsymbol{\gamma}_1}^{(G,\Ca)}(v),\ldots,
      \num{\boldsymbol{\gamma}_q}^{(G,\Ca)}(v)\Big)
    \end{align*}

   \end{techclaim}
\begin{subproof}
    Again, this follows from Corollary~\ref{cor:ar6a}.
    \uend
  \end{subproof}

  To complete the proof, we need to distinguish between the different
  aggregation functions. \MEAN-aggregation is most problematic,
  because it involves a division, which we can only approximate in our
  logic.
  
 \begin{cs}
   \case1 $\agg=\SUM$.\\
   We substitute the r-expressions
   $\boldsymbol{\mu}_i(x,x')$ of Claim~\ref{claim:c:1} for the r-schema $\vec Z$ in the
   r-expression $\logic{itadd}$ of Lemma~\ref{lem:ar7} to obtain guarded r-expressions
   $\boldsymbol{\sigma}_i$, for $i\in[r]$, such that for all $v\in V(G)$
   we have
   \begin{equation}
     \label{eq:26}
     \num{\boldsymbol{\sigma}_i}^{(G,\Ca)}(v)=\sum_{v'\in 
       N(v)}\num{\boldsymbol{\mu}_i}^{(G,\Ca)}(v,v'). 
   \end{equation}
   Then we substitute the r-expressions $\boldsymbol{\sigma}_i(x)$ for
   the variables $\vec Z_i$ in the formulas $\boldsymbol{\gamma}_j$ of
   Claim~\ref{claim:c:2} and obtain the desired r-expressions $\logic{l-eval}_j(x)$
   such that
   \begin{align*}
     &\Big(\num{\logic{l-eval}_1}^{(G,\Ca)}(v),\ldots,
       \num{\logic{l-eval}_q}^{(G,\Ca)}(v)\Big)\\
     &\hspace{1cm}=\comb\Big(\num{\vec X_1}^{(G,\Ca)}(v),\ldots, \num{\vec
        X_{p}}^{(G,\Ca)}(v),
          \num{\boldsymbol{\sigma}_1}^{(G,\Ca)}(v),\ldots,
          \num{\boldsymbol{\sigma}_{r}}^{(G,\Ca)}(v)\Big)\\
     &\hspace{1cm}=\Cy(v).
   \end{align*}
   Thus in this case, the r-expressions $\logic{l-eval}_j(x)$ even
   define $\Cy$ exactly. Of course this implies that they satisfy
   \eqref{eq:25}.
   \case2 $\agg=\MAX$.\\
   We can argue as in Case~1, using the r-expression $\logic{max}$ of
   Lemma~\ref{lem:ar7} instead of $\logic{itadd}$. Again, we obtain
   r-expressions $\logic{l-eval}_j(x)$ that define $\Cy$ exactly.

   \case3 $\agg=\MEAN$.\\
   The proof is similar to Case~1 and Case~2, but we need
   to be careful. We cannot define the mean of a family of
   numbers exactly, but only approximately, because of the division it
   involves.

  Exactly as in Case~1 we define r-expressions
  $\boldsymbol{\sigma}_i(x)$ satisfying \eqref{eq:26}. Recall that
  $\lambda$ is a Lipschitz constant for $\comb$. Let $\boldsymbol{\delta}(x)\coloneqq\#
  x'.E(x,x')$ be a term defining the degree of a vertex. Using
  Lemma~\ref{lem:ar6} we can construct an r-expression
  $\boldsymbol{\nu}_i$ such that 
  \[
    \left|\frac{\num{\boldsymbol{\sigma}_i}^{(G,\Ca)}(v)}{\num{\boldsymbol{\delta}}^{(G,\Ca)}(v)}
      -
      \num{\boldsymbol{\nu}_i}^{(G,\Ca)}(v)
    \right|
    < 2^{-\Ca(W)-\lambda}\le\lambda^{-1}2^{-\Ca(W)}
  \]
  if $\num{\boldsymbol{\delta}}^{(G,\Ca)}(v)=\deg(v)\neq 0$ and $
  \num{\boldsymbol{\nu}_i}^{(G,\Ca)}(v)=0$ otherwise. Thus, letting
  \begin{align*}
    \Cz(v)&\coloneqq \MEAN\big(\biglmulti
    \msg(\Cx(v),\Cx(v'))\bigmid v'\in N(v)\bigrmulti\big)\\
    &=\begin{cases}
      0&\text{if }\deg(v)=0,\\
      \left(\frac{\num{\boldsymbol{\sigma}_1}^{(G,\Ca)}(v)}{\num{\boldsymbol{\delta}}^{(G,\Ca)}(v)},\ldots, \frac{\num{\boldsymbol{\sigma}_r}^{(G,\Ca)}(v)}{\num{\boldsymbol{\delta}}^{(G,\Ca)}(v)}\right)&\text{otherwise}
 \end{cases}
  \end{align*}
  we have
  \[
    \inorm{\Cz(v)-\Big(\num{\boldsymbol{\nu}_1}^{(G,\Ca)}(v),\ldots, \num{\boldsymbol{\nu}_r}^{(G,\Ca)}(v)\Big)}\le\lambda^{-1}2^{-\Ca(W)}
  \]
  for all $v\in V(G)$. By the Lipschitz continuity of $\comb$, this
  implies
  \begin{equation}
    \label{eq:27}
    \inorm{\comb\big(\Cx(v),\Cz(v)\big)-\comb\Big(\Cx(v),
      \num{\boldsymbol{\nu}_1}^{(G,\Ca)}(v),\ldots,
      \num{\boldsymbol{\nu}_r}^{(G,\Ca)}(v)\Big)}\le 2^{-\Ca(W)}.
  \end{equation}
  We substitute the r-expressions $\boldsymbol{\nu}_i(x)$ for
   the variables $\vec Z_i$ in the formulas $\boldsymbol{\gamma}_j$ of
   Claim~\ref{claim:c:2} and obtain r-expressions $\logic{l-eval}_j(x)$
   such that
   \begin{align*}
     &\Big(\num{\logic{l-eval}_1}^{(G,\Ca)}(v),\ldots,
       \num{\logic{l-eval}_q}^{(G,\Ca)}(v)\Big)\\
     &\hspace{5mm}=\comb\Big(\num{\vec X_1}^{(G,\Ca)}(v),\ldots, \num{\vec
        X_{p}}^{(G,\Ca)}(v),
          \num{\boldsymbol{\nu}_1}^{(G,\Ca)}(v),\ldots,
          \num{\boldsymbol{\nu}_{r}}^{(G,\Ca)}(v)\Big)\\
     &\hspace{5mm}=\comb\Big(\Cx(v),
          \num{\boldsymbol{\nu}_1}^{(G,\Ca)}(v),\ldots,
          \num{\boldsymbol{\nu}_{r}}^{(G,\Ca)}(v)\Big).
   \end{align*}
   Since $\Cy(v)=\comb\big(\Cx(v),\Cz(v)\big)$, the assertion
   \eqref{eq:25} follows from \eqref{eq:27}.
   \qedhere
 \end{cs}
\end{proof}

\begin{proof}[Proof of Theorem~\ref{theo:uniform}]
  We fix a graph $G$ and assignment $\Ca$ over $G$ for the
  presentation of the proof; as usual the formulas we shall define will not
  depend on this graph and assignment.  Let $\Cx\in\CS_p(G)$ be the
  signal defined in \eqref{eq:22}.

  Suppose that $\FN=(\FL^{(1)},\ldots,\FL^{(d)})$. Let $p^{(i-1)}$ be
  the input dimension of $\FL^{(i)}$, and let $p^{(i)}$ be the output
  dimension. Then $p=p^{(0)}$ and $q=p^{(d)}$.
  Moreover, let
  $\Cx^{(0)}\coloneqq\Cx$ and
  $\Cx^{(i)}\coloneqq\tilde\FL^{(i)}(G,\Cx^{(i-1)})$ for
  $i\in[d]$. Note that $\Cx^{(d)}=\Cy$.

  For every $i\in[d]$, let $\lambda^{(i)}\coloneqq\lambda(\FL^{(i)})$
  be the constant of Lemma~\ref{lem:gnnbound2}. We inductively define
  a sequence of $\GC$-terms $\logic{err}^{(i)}(x)$, which will give us the
  desired error bounds. We let $\logic{err}^{(d)}(x)\coloneqq W(x)$. To
  define $\logic{err}^{(i)}(x)$ for $0\le i<d$, we first note that by
  Lemma~\ref{lem:ar8}, for every $\GC$-term $\theta(x)$ there is a $\GC$-term $\logic{maxN}_\theta(x)$ such that for every $v\in V(G)$ we have
  \[
    \sem{\logic{maxN}_\theta}^{(G,\Ca)}(v)=\max\Big\{\sem{\theta}^{(G,\Ca)}(w)\Bigmid
    w\in N_G[v]\Big\}.
  \]
  We let $\logic{dg}(x)\coloneqq(\# x'.E(x,x'))+\one$ and
  \[
    \logic{err}^{(i)}(x)\coloneqq 
    \logic{maxN}_{\logic{err}^{(i+1)}}(x)+\lambda^{(i+1)}\cdot\logic{maxN}_\logic{dg}(x)+\one.
  \]
  Letting
  \[
    k^{(i)}(v)\coloneqq\sem{\logic{err}^{(i)}}^{(G,\Ca)}(v),
  \]
  for every 
  $v\in V(G)$ and $0\le i< d$, we have
  \begin{equation}
    \label{eq:28}
    k^{(i)}(v)=\max\big\{k^{(i+1)}(w)\bigmid w\in
    N_G[v]\big\}+\lambda^{(i+1)}\max\big\{\deg_G(w)+1\bigmid w\in
    N_G[v]\big\}+1.
  \end{equation}
  Furthermore, $k^{(d)}(v)=\Ca(W)(v)$.

  \medskip
  Now for $i\in[d]$ and $j\in[p^{(i)}]$ we shall define guarded
  r-expressions $\rexp_j^{(i)}(x)$ such that for all $v\in
  V(G)$, with
  \[
    \Cz^{(i)}(v)\coloneqq\Big(\num{\rexp_1^{(i)}}(v),\ldots,
    \num{\rexp_{p^{(i)}}^{(i)}}(v)\Big)
  \]
  we have 
  \begin{equation}
    \label{eq:29}
    \inorm{\Cx^{(i)}(v)-\Cz^{(i)}(v)}\le2^{-k^{(i)}(v)}.
  \end{equation}
  For $i=d$ and with $\logic{gnn-eval}_j\coloneqq\rexp^{(d)}_j$, this
  implies \eqref{eq:23} and hence the assertion of the theorem.

  To define $\rexp_j^{(1)}(x)$, we apply
  Lemma~\ref{lem:uniform} to the first layer $\FL^{(1)}$ and
  substitute $\logic{err}^{(0)}$ for $W$ in the
  resulting-expression. Then \eqref{eq:29} for $i=1$
  follows directly from Lemma~\ref{lem:uniform} and the fact that
  $k^{(0)}(v)\ge k^{(1)}(v)$ for all $v$.

  For the inductive step, let $2\le i\le d$ and suppose that we have
  defined $\rexp^{(i-1)}_j(x)$ for all $j\in[p^{(i-1)}]$. To define
  $\rexp_j^{(i)}(x)$, we apply Lemma~\ref{lem:uniform} to the $i$th
  layer $\FL^{(i)}$ and substitute
  $\rexp_1^{(i-1)},\ldots, \rexp_{p^{(i-1)}}^{(i-1)}$ for
  $\vec X_1,\ldots,\vec X_{p^{(i-1)}}$ and $\logic{err}^{(i-1)}$ for $W$ in
  the resulting r-expression. Then by Lemma~\ref{lem:uniform}, for
  all $v\in V(G)$ we have
  \begin{equation}
    \label{eq:30}
    \inorm{\hat\FL^{(i)}(G,\Cz^{(i-1)})(v)-\Cz^{(i)}(v)}\le
    2^{-k^{(i-1)}(v)}.
  \end{equation}
  Moreover, by Lemma~\ref{lem:gnnbound2} applied to
  $\FL^{(i)}$ and
  $\Cx\coloneqq\Cx^{(i-1)}$, $\Cx'\coloneqq\Cz^{(i-1)}$ we have
  \begin{align*}
    \notag
    &\inorm{\Cx^{(i)}(v)-\hat\FL^{(i)}(G,\Cz^{(i-1)})(v)}\\
    \notag
    \le\;&
    \lambda^{(i)}\max\left\{\left.\inorm{\Cx^{(i-1)}(w)-\Cz^{(i-1)}(w)}\;\right|\;
           w\in N_G[v]\right\}\big(\deg_G(v)+1\big)\\
    \notag
    \le\;&
    \lambda^{(i)}\max\Big\{2^{-k^{(i-1)}(w)}\Bigmid
      w\in N_G[v]\Big\}\big(\deg_G(v)+1\big)\\
\le\;&
    \max\Big\{2^{-k^{(i-1)}(w)+\lambda^{(i)}\cdot(\deg_G(v)+1)}\Bigmid
      w\in N_G[v]\Big\}.
  \end{align*}
  We choose a $w\in N_G[v]$ minimising $k^{(i-1)}(w)$. Then
  \begin{equation}
    \label{eq:31}
    \Big\|\Cx^{(i)}(v)-\hat\FL^{(i)}(G,\Cz^{(i-1)})(v)\Big\|\le 2^{-(k^{(i-1)}(w)-\lambda^{(i)}\cdot(\deg_G(v)+1))}.
  \end{equation}
  Combining \eqref{eq:30} and \eqref{eq:31} by the triangle
  inequality, we get
   \begin{equation*}
    \big\|\Cx^{(i)}(v)-\Cz^{(i)}(v)\big\|\le2^{-k^{(i-1)}(v)}+2^{-(k^{(i-1)}(w)-\lambda^{(i)}\cdot(\deg_G(v)+1))}.
  \end{equation*}
  Observe that $k^{(i-1)}(v)\ge k^{(i)}(v)+1$ and
  \begin{align*}
    &k^{(i-1)}(w)-\lambda^{(i)}\cdot(\deg_G(v)+1)\\
    &=\max\big\{k^{(i)}(v')\bigmid v'\in
      N_G[w]\big\}+\lambda^{(i)}\max\big\{\deg_G(v')+1\bigmid v'\in
      N_G[w]\big\}+1\\
    &\hspace{2cm}-\lambda^{(i)}\cdot(\deg_G(v)+1)\\
    &\ge
      k^{(i)}(v)+\lambda^{(i)}\big(\deg_G(v)+1\big)+1
    -\lambda^{(i)}\cdot(\deg_G(v)+1)\hspace{2cm}\text{since
    }v\in N_G[w]\\
    &= k^{(i)}(v)+1.
  \end{align*}
  Thus
  \[
    \big\|\Cy^{(i)}(v)-\Cz^{(i)}(v)\big\|\le
    2^{-k^{(i)}(v)-1}+2^{-k^{(i)}(v)-1}=2^{-k^{(i)}(v)},
  \]
  which proves \eqref{eq:29} and hence the theorem.
\end{proof}

\begin{myremark}
  It is worth mentioning that the approximation error in
  Theorem~\ref{theo:uniform} is due to the division involved in MEAN
  aggregations. If we only consider rational piecewise linear GNNs
  with SUM and MAX aggregation, we obtain an exact simulation by
  $\GC$-formulas, that is, we can
  replace inequality \eqref{eq:23} by the equality
  \[
    \Cy(v)=\big(\num{\logic{gnn-eval}_1}^{(G,\Ca)}(v),\ldots,
    \num{\logic{gnn-eval}_q}^{(G,\Ca)}(v)\big).
  \]
  This follows easily from the proof.

  However, for MEAN aggregation, we cannot achieve such an exact
  result, at least not if we work with dyadic rationals. If we try to
  use arbitrary rationals, we run into other difficulties, because the
  denominators of the fractions we obtain can get very large.
  \uend
\end{myremark}

Since logics define queries, when comparing the expressiveness of graph neural networks with that of
logics, it is best to focus on queries. Recall from
Section~\ref{sec:graph} that we identified
$\ell$-labelled graphs with graphs carrying an $\ell$-dimensional
Boolean signal. A unary query on the class $\CGS^\bool_\ell$ is an equivariant signal transformations from
$\CGS_\ell^\bool$ to $\CGS_1^\bool$.

We say that a GNN $\FN$ \emph{computes} a unary query ${\CQ}:
\CGS_\ell^\bool\to\CGS_1^\bool$ if for all $(G,\Cb)\in
\CGS_\ell^\bool$ and all $v\in V(G)$ it holds that
\begin{equation}
  \label{eq:34}
    \begin{cases}
    \tilde\FN(G,\Cb)(v)\ge \frac{3}{4}&\text{if }{\CQ}(G,\Cb)(v)=1,\\
    \tilde\FN(G,\Cb)(v)\le \frac{1}{4}&\text{if }{\CQ}(G,\Cb)(v)=0.
  \end{cases}
\end{equation}
Observe that if we allow $\lsig$ or $\relu$ activations, we can
replace the $\ge\frac{3}{4}$ and $\le\frac{1}{4}$ in \eqref{eq:34} by
$=1$ and $=0$ and thus require
$\tilde\FN(G,\Cb)(v)={\CQ}(G,\Cb)(v)$. We simply apply the
transformation $\lsig(2x-\frac{1}{2})$ to the output. It maps the interval
$(-\infty,\frac{1}{4}]$ to $0$ and the interval $[\frac{3}{4},\infty)$
to $1$. With other activations such as the logistic function, this is
not possible, which is why we chose our more flexible definition.

\begin{corollary}\label{cor:uniform}
  Every unary query on $\CGS_\ell$ that is computable by a rational
  piecewise linear GNN is
  definable in $\GC$.
\end{corollary}

Note that this Corollary is  Theorem~\ref{theo:main3} stated
in the introduction.  

The reader may wonder if the converse of the previous corollary holds,
that is, if every query definable in $\GC$ is computable by a rational
piecewise linear GNN. It is not; we refer the reader to
Remark~\ref{rem:converse}. 

  As mentioned earlier, there are versions of the theorem for all the
extensions of basic GNNs that we discussed in
Section~\ref{sec:gnn}. For later reference, we state the version for GNNs with global
readout.

\begin{theorem}\label{theo:uniform-gc}
  Let $\FN$ be a rational piecewise linear GNN with global readout of input dimension $p$
  and output dimension $q$. Let $\vec X_1,\ldots,\vec X_p$ be
  r-schemas of type $\rtp{\ttv}$, and let $W$ be a function variable
  of type $\ttv\to\ttn$.  
Then there are 
  r-expressions $\logic{gnn-eval}_1(x)$, $\ldots$, $\logic{gnn-eval}_q(x)$
  in $\GCgc$ such that the following holds for all graphs $G$ and assignments
  $\Ca$ over $G$. Let $\Cx\in\CS_p(G)$ be the signal defined by
  \begin{equation}
    \label{eq:32}
    \Cx(v)\coloneqq\Big(\num{\vec X_1}^{(G,\Ca)}(v),\ldots,
    \num{\vec X_p}^{(G,\Ca)}(v)\Big),
  \end{equation}
  and let $\Cy=\tilde\FN(G,\Cx)$. Then for
  all $v\in V(G)$, 
  \begin{equation}
    \label{eq:33}
     \inorm{\Cy(v)
    -\big(\num{\logic{gnn-eval}_1}^{(G,\Ca)}(v),\ldots, \num{\logic{gnn-eval}_q}^{(G,\Ca)}(v)\big)}
    \le 2^{-\Ca(W)(v)}.
  \end{equation}
\end{theorem}

The proof of the this theorem is completely analogous to the proof of
Theorem~\ref{theo:uniform}. However, when expressing the global
aggregations (as in \eqref{eq:26}), we need to use unguarded counting
terms of the shape \eqref{eq:6} and hence end up with $\GCgc$-expressions.

\begin{myremark}\label{rem:mc4}Theorem~\ref{theo:uniform} and Corollary~\ref{cor:uniform} also have
  versions for 1-GNNs and the modal fragment $\MC$. In
  Theorem~\ref{theo:uniform}, if $\FN$ is a 1-GNN, we can obtain
  \emph{modal} r-expressions $\logic{gnn-eval}_i(x)$. To prove this,
  we need a corresponding version of Lemma~\ref{lem:uniform}. In the
  proof of this lemma, in Claim~\ref{claim:c:1} we have to only need to simulate a
  message function
  \[
  \msg\Big(\num{\vec X_1}^{(G,\Ca)}(v'),\ldots, \num{\vec
    X_{p}}^{(G,\Ca)}(v')\Big),
    \]  
    dropping    the arguments depending on
  $v$, and hence we only need to construct r-expressions
  $\mu_i(x')$. Then in the aggregations, we can apply the modal
  version of Lemma~\ref{lem:ar7} to obtain modal r-expressions
  $\boldsymbol{\sigma}_i$ satisfying \eqref{eq:26} (in Case~1, and the
  corresponding equalities in Cases~2 and 3).

  The rest of the proofs of Lemma~\ref{lem:uniform} and
  Theorem~\ref{theo:uniform} go through without any changes. The modal
  version of Corollary~\ref{cor:uniform} then follows; it reads as:
  {\itshape Every unary query on $\CGS_\ell$ that is computable by a
    rational piecewise linear 1-GNN is definable in $\MC$.}
  \uend
\end{myremark}

\section{The Non-Uniform Case:
  GNNs with Arbitrary Weights and Families of GNNs}
\label{sec:nonuniform}

Now we consider the general case where the weights in the neural
networks are arbitrary real numbers. We also drop the assumption that
the activation functions be piecewise linear, only requiring rpl
approximability. The price we pay is a non-uniformity on the side of
the logic and a slightly weaker approximation
guarantee as well as a boundedness assumption on the input signal.

\begin{theorem}\label{theo:nonuniform}
  Let $\FN$ be an rpl-approximable GNN of input dimension $p$
  and output dimension $q$. Let $\vec X_1,\ldots,\vec X_p$ be
  r-schemas of type $\rtp{\ttv}$, and let $W,W'$ be function variables
  of type $\emptyset\to\ttn$.

  Then there are
  r-expressions $\logic{gnn-eval}_1(x)$, $\ldots$, $\logic{gnn-eval}_q(x)$
  in $\GCnu$ such
  that the following holds for all graphs $G$ and assignments
  $\Ca$ over $G$.
 Let $\Cx\in\CS_p(G)$ be the signal defined by
  \begin{equation}
    \label{eq:35}
    \Cx(v)\coloneqq\Big(\num{\vec X_1}^{(G,\Ca)}(v),\ldots,
    \num{\vec X_p}^{(G,\Ca)}(v)\Big),
  \end{equation}
  and let $\Cy=\tilde\FN(G,\Cx)$. Assume that
  $\inorm{\Cx}\le\Ca(W)$ and that $\Ca(W')\neq0$. Then for
  all $v\in V(G)$, 
  \begin{equation}
    \label{eq:36}
     \inorm{\Cy(v)
    -\big(\num{\logic{gnn-eval}_1}^{(G,\Ca)}(v),\ldots, \num{\logic{gnn-eval}_q}^{(G,\Ca)}(v)\big)}
    \le\frac{1}{\Ca(W')}.
\end{equation}  
\end{theorem}

Let us comment on the role of the two 0-ary functions (that is,
constants) $W,W'$. We introduce them to add flexibility in the
bounds. Their values depend on the assignment $\Ca$, which means that
we can freely choose them. For example, we can let
$\Ca(W)=\Ca(W')=n\coloneqq|G|$. Then we get an approximation error of
$1/n$ for input signals bounded by $n$. Or we could let $\Ca(W)=1$ and
$\Ca(W')=100$. Then we get an approximation error of 1\% for Boolean
input signals.

Since we move to a non-uniform regime anyway, to obtain the most
general results we may as well go all
the way to a non-uniform GNN model where we have different GNNs for
every size of the input graphs.

We need additional terminology. We define the \emph{bitsize}
$\bsize(\FF)$ of a rational piecewise linear FNN $\FF$ to be the sum
of the bitsizes of its skeleton, all its weights and biases, and all
its activations. We define the \emph{weight} of an arbitrary FNN
$\FF=(V,E,(\Fa_v)_{v\in V},\vec w,\vec b)$ to be
\[
  \wt(\FF)\coloneqq|V|+[E|+\inorm{\vec w}+\inorm{\vec b}+\max_{v\in V}\big(\lambda(\Fa_v)+\Fa_v(0)\big)
\]
Here $\lambda(\Fa_v)$ denotes the least integer that is a Lipschitz
constant for $\Fa_v$.
The \emph{size} $\size(\FF)$ of a rational piecewise linear FNN $\FF$
is the maximum of its bitsize and its weight. 
The \emph{depth} $\depth(\FF)$ of an FNN $\FF$ is the depth of its
skeleton, that is, the length of a longest path
from an input node to an output node of $\FF$.

The \emph{weight} $\wt(\FN)$ of a GNN $\FN$ is the sum of the weights of the FNNs for the message and
combination functions of all layers of $\FN$.
The \emph{bitsize} $\bsize(\FN)$ and the \emph{size} $\size(\FN)$ of a rational piecewise linear GNN
$\FN$ is the sum of the (bit)sizes of all its FNNs. The
\emph{skeleton} of a GNN $\FN$ consists of the directed acyclic graphs
underlying the FNNs for the message and combination functions of all
layers of $\FN$.
Thus if two GNNs have the same skeleton they have the same number of layers and the same input
and output dimensions on all layers, but they
may have different activation functions and different weights. The \emph{depth} $\depth(\FN)$ of a
GNN $\FN$ is the number of layers of $\FN$ times the
maximum depth of all its FNNs.

Let
$\CN=(\FN^{(n)})_{n\in\Nat}$ be a family of GNNs. Suppose that the input dimension of $\FN^{(n)}$ is $p^{(n)}$ and the
output dimension is $q^{(n)}$. It will be convenient to call $(p^{(n)})_{n\in\Nat}$ the
\emph{input dimension} of $\CN$ and $(q^{(n)})_{n\in\Nat}$ the
\emph{output dimension}.
Then for every graph $G$ of order $n$ and every $\Cx\in\CS_{p(n)}(G)$
we let 
$\CN(G,\Cx)\coloneqq \FN^{(n)}(G,\Cx)$ and
$\tilde\CN(G,\Cx)\coloneqq\tilde\FN^{(n)}(G,\Cx)$. Thus $\CN$ computes
a generalised form of signal transformation where the input and output
dimension depend on the order of the input graph.

We say that $\CN$ is of \emph{polynomial weight} if there is a
polynomial $\pi(X)$ such that $\wt(\FN^{(n)})\le \pi(n)$ for all
$n$. \emph{Polynomial (bit)size} is defined similarly. The family $\CN$ is of
\emph{bounded depth} if there is a $d\in\Nat$ such that
$\depth(\FN^{(n)})\le d$ for all $n$. The family
$\CN$ is \emph{rpl approximable} if there is a polynomial $\pi'(X,Y)$
such that for all $n\in\PNat$ and all $\epsilon>0$, every activation
function of $\FN^{(n)}$ is $\epsilon$-approximable by a rational
piecewise linear function of bitsize at most
$\pi'(\epsilon^{-1},n)$.

\begin{theorem}\label{theo:snonuniform}
  Let $\CN$ be an rpl-approximable polynomial-weight, bounded-depth
  family of GNNs of input dimension $(p^{(n)})_{n\in\Nat}$ and
output dimension $(q^{(n)})_{n\in\Nat}$. Let $\vec X$ be an r-schema of type
  $\rtp{\ttv\ttn}$, and let $W,W'$ be function variables of type
  $\emptyset\to\ttn$.

  Then there is 
  an r-expression $\logic{gnn-eval}(x,y)$ in ${\GCnu}$ such that the following
  holds for all graphs $G$ and assignments $\Ca$ over $G$.  Let
  $n\coloneqq|G|$, and let 
  $\Cx\in\CS_{p^{(n)}}(G)$ be the signal defined by
  \begin{equation}
    \label{eq:37}
    \Cx(v)\coloneqq\Big(\num{\vec X}^{(G,\Ca)}(v,0),\ldots,
    \num{\vec X}^{(G,\Ca)}(v,p^{(n)}-1)\Big).
  \end{equation}
 Assume that
  $\inorm{\Cx}\le\Ca(W)$ and that $\Ca(W')\neq0$. Let  $\Cy=\tilde\CN(G,\Cx)\in\CS_{q^{(n)}}(G)$. Then for
  all $v\in V(G)$, 
  \begin{equation}
    \label{eq:38}
     \inorm{\Cy(v)
    -\big(\num{\logic{gnn-eval}}^{(G,\Ca)}(v,0),\ldots, \num{\logic{gnn-eval}}^{(G,\Ca)}(v,q^{(n)}-1)\big)}
    \le\frac{1}{\Ca(W')}.
\end{equation}  
\end{theorem}

Observe that Theorem~\ref{theo:snonuniform} implies
Theorem~\ref{theo:nonuniform}, because we can simply let $\CN$ be the
family consisting of the same GNN for every $n$. So we only need to
prove Theorem~\ref{theo:snonuniform}.
The basic idea of the proof is simple. We exploit the continuity of the functions computed by FNNs
and GNNs not only in terms of the input signals but also in terms of
the weights and the biases. This allows us to approximate the
functions computed by GNNs with arbitrary real weights by GNNs with
rational weights. However, the bitsize of the rationals we need to get
a sufficiently precise approximation depends on the size of the input
graph, and this leads to the non-uniformity.

Before we delve into the proof, let us state one important corollary.
Extending the definition for single GNNs in the obvious way, we say
that a family $\CN=(\FN^{(n)})_{n\in\Nat}$ of GNNs \emph{computes} a
unary query ${\CQ}:\CGS^\bool_p\to\CGS^\bool_1$ on $p$-labelled graphs if
for all $n\in\Nat$, all $(G,\Cb)\in \CGS_p^\bool$ with $n=|G|$, and
all $v\in V(G)$ it holds that
\[
  \begin{cases}
    \tilde\FN^{(n)}(G,\Cb)(v)\ge \frac{3}{4}&\text{if }{\CQ}(G,\Cb)(v)=1,\\
    \tilde\FN^{(n)}(G,\Cb)(v)\le \frac{1}{4}&\text{if }{\CQ}(G,\Cb)(v)=0.
  \end{cases}
\]

\begin{corollary}\label{cor:snonuniform}
  Every unary query on $\CGS^\bool_p$ that is computable by an
  rpl-approximable polynomial-weight bounded-depth family of GNNs is
    definable in ${\GCnu}$.
\end{corollary}

The exact analogues of Theorems~\ref{theo:nonuniform} and
  \ref{theo:snonuniform} hold for GNNs with global readout and the
  logic $\GCgc$, with only small modifications of the proof. For later
  reference, we state the analogue of Theorem~\ref{theo:snonuniform}.
\begin{theorem}\label{theo:snonuniform-gc}
  Let $\CN$ be an rpl-approximable polynomial-weight, bounded-depth
  family of GNNs with global readout of input dimension $(p^{(n)})_{n\in\Nat}$ and
output dimension $(q^{(n)})_{n\in\Nat}$. Let $\vec X$ be an r-schema of type
  $\rtp{\ttv\ttn}$, and let $W,W'$ be function variables of type
  $\emptyset\to\ttn$.

  Then there is 
  an r-expression $\logic{gnn-eval}(x,y)$ in ${\GCgcnu}$ such that the following
  holds for all graphs $G$ and assignments $\Ca$ over $G$.  Let
  $n\coloneqq|G|$, and let 
  $\Cx\in\CS_{p^{(n)}}(G)$ be the signal defined by
  \begin{equation}
    \label{eq:137}
    \Cx(v)\coloneqq\Big(\num{\vec X}^{(G,\Ca)}(v,0),\ldots,
    \num{\vec X}^{(G,\Ca)}(v,p^{(n)}-1)\Big).
  \end{equation}
 Assume that
  $\inorm{\Cx}\le\Ca(W)$ and that $\Ca(W')\neq0$. Let  $\Cy=\tilde\CN(G,\Cx)\in\CS_{q^{(n)}}(G)$. Then for
  all $v\in V(G)$, 
  \begin{equation}
    \label{eq:138}
     \inorm{\Cy(v)
    -\big(\num{\logic{gnn-eval}}^{(G,\Ca)}(v,0),\ldots, \num{\logic{gnn-eval}}^{(G,\Ca)}(v,q^{(n)}-1)\big)}
    \le\frac{1}{\Ca(W')}.
\end{equation}  
\end{theorem}

\begin{myremark}\label{rem:mc5}As the uniform simulation results, the nonuniform
  Theorems~\ref{theo:nonuniform} and \ref{theo:snonuniform} as well as Corollary~\ref{cor:snonuniform} have
  versions for 1-GNNs and the modal fragment $\MC$. In
  Theorem~\ref{theo:nonuniform}, if $\FN$ is a 1-GNN then we obtain
  r-expressions $\logic{gnn-eval}_i(x)$ in $\MCnu$, the modal fragment
  with built-in numerical relations. Similarly, in
  Theorem~\ref{theo:snonuniform}, if $\CN$ is a family of 1-GNN then
  we obtain an
  r-expression $\logic{gnn-eval}(x,y)$ in $\MCnu$. And the modified
  version of Corollary~\ref{cor:snonuniform} states that queries
  computable by 
  rpl-approximable polynomial-weight bounded-depth families of 1-GNNs
  are definable in $\MCnu$.

  As the modified versions of both Theorem~\ref{theo:nonuniform} and
  Corollary~\ref{cor:snonuniform} follow easily from the modified
  version of Theorem~\ref{theo:snonuniform}, we only need to adapt the
  proof of Theorem~\ref{theo:snonuniform}. Within this long proof, the
  only place where the exact messaging mechanism plays a role is in
  the proof of Claim~\ref{claim:g:3} (on page~\pageref{page:claim3}). The changes we
  need to make there are analogous to the changes we needed to make in
  the proof of the modified version of Lemma~\ref{lem:uniform} (see
  Remark~\ref{rem:mc4}).
  \uend
\end{myremark}

\subsection{Bounds and Approximations for FNNs}

In this section, we shall prove that we can approximate rpl
approximable FNNs by rational piecewise linear FNNs whose size is
bounded in terms of the approximation ratio. For this, we first need
to establish bounds on the Lipschitz constant and growth of
an FNN in terms of its structure, its activation functions, and its
parameters.

Throughout this section, we let $\FA=\big(V,E,(\Fa_v)_{v\in V}\big)$
be an FNN architecture of input dimension $p$ and output dimension $q$. We let $d$ be the depth and $\Delta$ the
maximum in-degree of the directed graph $(V,E)$. Without loss of
generality, we assume $\Delta\ge 1$ and thus $d\ge 1$. If $\Delta=0$,
we simply add a dummy edge of weight $0$ to the network.
Moreover, we let
$\lambda\in\PNat$ be a Lipschitz constant for all
activation function $\Fa_v$ for $v\in V$, and we let
\[
  \mu\coloneqq\max\big\{\ceil{|\Fa_v(0)|}\bigmid v\in V\big\}.
\]
For vectors $\vec x\in\Real^p$, $\vec w\in\Real^E$, $\vec
b\in\Real^V$, we assume that $\vec x=(x_1,\ldots,x_p)$, 
$\vec w=(w_e)_{e\in E}$, and $\vec b=(b_v)_{v\in V}$.

In the first two lemmas we analyse the dependence of the growth and variation of
the functions $f_{\FA,v}(\vec x,\vec w,\vec b)$ and 
$\FA(\vec x,\vec w,\vec b)$ on the constants
$d,\Delta,\lambda,\mu$ and $\inorm{\vec x},\inorm{\vec w},\inorm{\vec
  b}$ (more precisely than in Lemma~\ref{lem:fnngrowth}).

\begin{lemma}\label{lem:ub1}
  Let $\gamma\coloneqq2\Delta\lambda\max\{\lambda,\mu\}$. Then for all $\vec x\in\Real^p$, $\vec b\in\Real^{V}$,
  and $\vec w\in\Real^E$, and all $v\in V$ of depth $t$ we have
  \begin{equation}\label{eq:39}
    \big|f_{\FA,v}(\vec x,\vec b,\vec
    w)\big|\le \gamma^t(\inorm{\vec w}+1)^t\big(\inorm{\vec x}+\inorm{\vec b}+1).
  \end{equation}
\end{lemma}

\begin{proof}
  Note that for all $x\in\Real$ we have 
  \begin{equation}
    \label{eq:40}
    |\Fa_v(x)|\le \lambda|x|+\mu.
  \end{equation}
  For all input nodes
  $X_i$ we have
  \begin{equation}
    \label{eq:41}
    \big|f_{\FA,X_i}(\vec x,\vec b,\vec
    w)\big|=|x_i|\le\inorm{\vec x}.
  \end{equation}
  This implies \eqref{eq:39} for $t=0$.

  \begin{techclaim}\label{claim:d:1}
    For all nodes $v\in V$ of depth $t\ge 1$ we have
    \begin{equation}
      \label{eq:42}
      \big|f_{\FA,v}(\vec x,\vec b,\vec
      w)\big|\le \big(\Delta\lambda\inorm{\vec w})^t\inorm{\vec x}+
      \sum_{s=0}^{t-1}\big(\Delta\lambda\inorm{\vec w})^{s}\big(\lambda\inorm{\vec b}+\mu\big).
    \end{equation}
  
  \end{techclaim}
\begin{subproof}
  We prove
  \eqref{eq:42} by induction on $t\ge 1$.
  Suppose that $v\in V$ is a node of depth $t$, and let
  $v_1,\ldots,v_k$ be its in-neighbours. Let $b\coloneqq b_v$ and
  $w_i\coloneqq w_{v_iv}$ for $i\in[k]$. Moreover, let
  $y_i\coloneqq f_{\FA,v_i}(\vec x,\vec b,\vec w)$ and
  $\vec y=(y_1,\ldots,y_k)$. If $t=1$, by \eqref{eq:41} we have
  \begin{equation}
    \label{eq:43}
    \inorm{\vec y}\le\inorm{\vec x}
  \end{equation}
  If $t>1$, by the induction hypothesis we have
  \begin{equation}
    \label{eq:44}
    \inorm{\vec y}\le \big(\Delta\lambda\inorm{\vec w})^{t-1}\inorm{\vec x}+ \sum_{s=0}^{t-2}\big(\Delta\lambda\inorm{\vec w})^{s}\big(\lambda\inorm{\vec b}+\mu\big).
  \end{equation}
  Thus
  \begin{align*}
    \big|f_{\FA,v}(\vec x,\vec b,\vec
    w)\big|
    &=\Big|\Fa_v\Big(b+\sum_{i=1}^kw_iy_i\Big)\Big|\\
    &\le \Bigg|\lambda\Big|b+\sum_{i=1}^kw_iy_i\Big|+\mu\Bigg|&\text{by
                                                        }\eqref{eq:40}\\
    &\le \lambda\sum_{i=1}^k|w_i|\cdot|y_i|+\lambda|b|+\mu\\
    &\le \lambda \Delta\inorm{\vec w} \inorm{\vec y}+\lambda\inorm{\vec b}+\mu.
  \end{align*}
  Now if $t=1$, assertion \eqref{eq:42} follows immediately from
  \eqref{eq:43}. If $t>1$, by \eqref{eq:44}
  we obtain
  \begin{align*}
    &\big|f_{\FA,v}(\vec x,\vec b,\vec
    w)\big|\\
    &\le \Delta\lambda\inorm{\vec w} \inorm{\vec y}+\lambda\inorm{\vec b}+\mu\\
    &\le \Delta\lambda\inorm{\vec w} \Big(\big(\Delta\lambda\inorm{\vec w})^{t-1}\inorm{\vec x}+ \sum_{s=0}^{t-2}\big(\Delta\lambda\inorm{\vec w})^{s}\big(\lambda\inorm{\vec b}+\mu\big)\Big)+\lambda\inorm{\vec b}+\mu\\
    &=\big(\Delta\lambda\inorm{\vec w})^{t}\inorm{\vec x}+\sum_{s=0}^{t-1}\big(\Delta\lambda\inorm{\vec w})^{s}\big(\lambda\inorm{\vec b}+\mu\big).
  \end{align*}
  This proves the claim.
  \uend
    \end{subproof}

    It remains to prove that the claim yields \eqref{eq:39} for $t\ge
    1$. Since
    $\Delta\lambda\ge 1$, we have
    \[
      \sum_{s=0}^{t-1}\big(\Delta\lambda\inorm{\vec w})^{s}\le
      \big(2\Delta\lambda(\inorm{\vec w}+1))^{t}.
    \]
    Thus by Claim~\ref{claim:d:1} we have
    \begin{align*}
      \big|f_{\FA,v}(\vec x,\vec b,\vec
      w)\big|
&\le \big(2\Delta\lambda(\inorm{\vec w}+1)\big)^t\big(\inorm{\vec x}+\lambda\inorm{\vec b}+\mu\big)\\
&\le  \big(2\Delta\lambda(\inorm{\vec w}+1)\big)^t\max\{\lambda,\mu\}\big(\inorm{\vec x}+\inorm{\vec b}+1\big)\\
&\le\gamma^t \big(\inorm{\vec w}+1\big)^t\big(\inorm{\vec x}+\inorm{\vec b}+1\big).
\qedhere
\end{align*}
\end{proof}

\begin{lemma}\label{lem:ub1a}
  For all $\vec x,\vec x'\in\Real^p$, $\vec b\in\Real^{V}$,
  and $\vec w\in\Real^E$ it holds that
  \[
    \inorm{\FA(\vec x,\vec w,\vec b)-\FA(\vec x',\vec w,\vec
      b)}\le(\lambda\Delta)^d\inorm{\vec w}^d\inorm{\vec x-\vec x'}
  \]
\end{lemma}

\begin{proof}
  Let $\vec x,\vec x'\in\Real^p$, $\vec b\in\Real^{V}$,
  and $\vec w\in\Real^E$.
    We shall prove by induction on $t$ that for all nodes $v\in V$ of
  depth $t$ we have
  \begin{equation}
    \label{eq:45}
    \inorm{f_{\FA,v}(\vec x,\vec w,\vec b)-f_{\FA,v}(\vec x',\vec w,\vec
      b)}\le(\lambda\Delta\inorm{\vec w})^t\inorm{\vec x-\vec x'}.
  \end{equation}
  Nodes of depth $t=0$ are input nodes, and we have
  \begin{equation*}
    \big|f_{\FA,X_i}(\vec x,\vec b,\vec
    w)-f_{\FA,X_i}(\vec x',\vec b,\vec
    w)\big|=|x_i-x_i'|\le\inorm{\vec x-\vec x'}.
  \end{equation*}
  For the inductive step, let $v\in V$ be a node of depth $t>0$, and
  let $v_1,\ldots,v_k$ be its in-neighbours. Let $b\coloneqq
  b_v,b'\coloneqq b'_v$ and $w_i\coloneqq w_{v_iv},w'_i\coloneqq
  w'_{v_iv}$ for $i\in[k]$. Moreover, let  $\vec
  y=(y_1,\ldots,y_k)$ and $\vec
  y'=(y_1',\ldots,y_k')$ with
  $y_i\coloneqq f_{\FA,v_i}(\vec x,\vec b,\vec w)$ and
  $y_i'\coloneqq f_{\FA,v_i}(\vec x',\vec b,\vec w)$. Then
  \begin{align*}
    \big|f_{\FA,v}(\vec x,\vec b,\vec
    w)-f_{\FA,v}(\vec x,\vec b',\vec
    w')\big|&=\Big|\Fa_v\Big(b+\sum_{i=1}^kw_iy_i\Big)-\Fa_v\Big(b+\sum_{i=1}^kw_iy'_i\Big)\Big|\\
            &\le \lambda\Big(\sum_{i=1}^kw_i|y_i-y_i'|\Big)\\
            &\le \lambda\Delta\inorm{\vec w}\inorm{\vec y-\vec y'}.
  \end{align*}
  Since by the induction hypothesis we have $\inorm{\vec y-\vec y'}\le
  (\lambda\Delta\inorm{\vec w})^{t-1}\inorm{\vec x-\vec x'}$, the
  assertion \eqref{eq:45} follows.
\end{proof}

\begin{lemma}\label{lem:ub2}
  Let $\nu\coloneqq (4\Delta\lambda\gamma)^d$, where $\gamma$ is the
  constant of Lemma~\ref{lem:ub1}. Then for
  all $\epsilon\in\Real$, $\vec x\in\Real^p$, $\vec b,\vec b'\in\Real^{V}$,
  and $\vec w,\vec w'\in\Real^E$ with
  \begin{equation}
    \label{eq:46}
    0\le \max\big\{\|\vec b-\vec
  b'\|_\infty,\|\vec w-\vec
  w'\|_\infty\big\}\le\epsilon\le 1
  \end{equation}
we have
  \[
    \big\|\FA(\vec x,\vec b,\vec
    w)-\FA(\vec x,\vec b',\vec
    w')\big\|_\infty\le
\nu\big(\inorm{\vec w}+1\big)^{d}\big(\inorm{\vec x}+\inorm{\vec b}+1\big)\epsilon.
  \]
\end{lemma}

\begin{proof}
  Let $\vec x\in\Real^p$ and
  $\epsilon\in[0,1]$, $\vec b,\vec b'\in\Real^{V}$,
  $\vec w,\vec w'\in\Real^E$ such that \eqref{eq:46} holds.
 
  We shall prove by induction on $t$ that for all nodes $v\in V$ of
  depth $t$ we have
  \begin{equation}
    \label{eq:47}
    \big|f_{\FA,v}(\vec x,\vec b,\vec
    w)-f_{\FA,v}(\vec x,\vec b',\vec
    w')\big|\le \big(4\Delta\lambda\gamma (\inorm{\vec w}+1)\big)^{t}\big(\inorm{\vec x}+\inorm{\vec b}+1\big)\epsilon.
  \end{equation}
  Applied to the output nodes $v=Y_i$ of depth $\le d$, this yields the assertion of
  the lemma.

  Nodes of depth $t=0$ are input nodes, and we have
  \begin{equation}
    \label{eq:48}
    \big|f_{\FA,X_i}(\vec x,\vec b,\vec
    w)-f_{\FA,X_i}(\vec x,\vec b',\vec
    w')\big|=|x_i-x_i|=0.
  \end{equation}
  For the inductive step, let $v\in V$ be a node of depth $t>0$, and
  let $v_1,\ldots,v_k$ be its in-neighbours. Let $b\coloneqq
  b_v,b'\coloneqq b'_v$ and $w_i\coloneqq w_{v_iv},w'_i\coloneqq
  w'_{v_iv}$ for $i\in[k]$. Moreover, let  $\vec
  y=(y_1,\ldots,y_k)$ and $\vec
  y'=(y_1',\ldots,y_k')$ with
  $y_i\coloneqq f_{\FA,v_i}(\vec x,\vec b,\vec w)$ and
  $y_i'\coloneqq f_{\FA,v_i}(\vec x,\vec b',\vec w')$. 
  
  \begin{techclaim}\label{claim:e:1}
    \[
      \big|f_{\FA,v}(\vec x,\vec b,\vec
      w)-f_{\FA,v}(\vec x,\vec b',\vec
      w')\big|\le
\Delta\lambda \inorm{\vec y-\vec y'}\inorm{\vec w}+\Delta\lambda\big(\inorm{\vec y}+\|\vec y-\vec
              y'\|_\infty+1\big)\epsilon
   \]
  
  \end{techclaim}
\begin{subproof}
  By the definition of $f_{\FA,v}$ and the Lipschitz continuity of the
  activation functions we have
  \begin{align*}
    \big|f_{\FA,v}(\vec x,\vec b,\vec
    w)-f_{\FA,v}(\vec x,\vec b',\vec
    w')\big|&=\Big|\Fa_v\Big(b+\sum_{i=1}^kw_iy_i\Big)-\Fa_v\Big(b'+\sum_{i=1}^kw'_iy'_i\Big)\Big|\\
            &\le \lambda\cdot\Big(|b-b'|+\sum_{i=1}^k|w_iy_i-w_i'y_i'|\Big)
  \end{align*}
  Observe that $|b-b'|\le\inorm{\vec b-\vec b'}\le\epsilon$ and 
  \begin{align*}
    y_iw_i-y_i'w_i'&=(y_i-y_i')w_i+y_i'(w_i-w_i')\\
                   &=(y_i-y_i')w_i+(y_i'-y_i)(w_i-w_i')+y_i(w_i-w_i')\\
    &\le\inorm{\vec y-\vec y'}\inorm{\vec w}+\epsilon\|\vec y-\vec
      y'\|_\infty+\epsilon\|\vec y\|_\infty.
  \end{align*}
  Hence
  \begin{align*}
    \big|f_{\FA,v}(\vec x,\vec b,\vec
    w)-f_{\FA,v}(\vec x,\vec b',\vec
    w')\big|&\le\lambda\Big(\epsilon+\Delta\big(\inorm{\vec y-\vec y'}\inorm{\vec w}+\|\vec y-\vec
              y'\|_\infty \epsilon+\|\vec y\|_\infty\epsilon\big)\Big)\\
    &\le\Delta\lambda \inorm{\vec y-\vec y'}\inorm{\vec w}+\Delta\lambda\big(\inorm{\vec y}+\|\vec y-\vec
              y'\|_\infty+1\big)\epsilon.
  \end{align*}
  This proves the claim.
  \uend
  \end{subproof}
  
  By the inductive hypothesis \eqref{eq:47}, we have
  \begin{equation}
    \label{eq:49}
    \inorm{\vec y-\vec y'}\le  \big(4\Delta\lambda\gamma (\inorm{\vec w}+1)\big)^{t-1}\big(\inorm{\vec x}+\inorm{\vec b}+1\big)\epsilon.
  \end{equation}
  Thus 
  \begin{equation}
    \label{eq:50}
    \Delta\lambda \inorm{\vec y-\vec y'}\inorm{\vec w}\le 4^{t-1}\big(\Delta\lambda\gamma (\inorm{\vec w}+1)\big)^{t}\big(\inorm{\vec x}+\inorm{\vec b}+1\big)\epsilon
  \end{equation}
  and
  \begin{align}
    \notag
    \Delta\lambda \|\vec y-\vec
    y'\|_\infty\epsilon&\le
    \Delta\lambda\big(4\Delta\lambda\gamma (\inorm{\vec w}+1)\big)^{t-1}\big(\inorm{\vec x}+\inorm{\vec b}+1\big)\epsilon^2\\
    \label{eq:51}
    &\le 4^{t-1}\big(\Delta\lambda\gamma (\inorm{\vec w}+1)\big)^{t}\big(\inorm{\vec x}+\inorm{\vec b}+1\big)\epsilon&\text{because }\epsilon\le 1.
  \end{align}
  By Lemma~\ref{lem:ub1} we have $\inorm{\vec y}\le  \gamma^{t-1}\big(\inorm{\vec w}+1\big)^{t-1}\big(\inorm{\vec x}+\inorm{\vec b}+1)$
  and thus
  \begin{align}
\notag
    \Delta\lambda\inorm{\vec y}\epsilon
    &\le
    \Delta\lambda\gamma^{t-1}\big(\inorm{\vec w}+1\big)^{t-1}\big(\inorm{\vec x}+\inorm{\vec b}+1\big)\epsilon\\
    \label{eq:52}
    &\le
    \big(\Delta\lambda\gamma(\inorm{\vec w}+1)\big)^{t}\big(\inorm{\vec x}+\inorm{\vec b}+1\big)\epsilon
  \end{align}
  Plugging \eqref{eq:50}, \eqref{eq:51}, and \eqref{eq:52} into
  Claim~\ref{claim:e:1}, we get
  \begin{align*}
    \big|f_{\FA,v}(\vec x,\vec b,\vec
    w)-f_{\FA,v}(\vec x,\vec b',\vec
    w')\big|\le\;& 4^{t-1}\big(\Delta\lambda\gamma( \inorm{\vec w}+1)\big)^{t}\big(\inorm{\vec x}+\inorm{\vec b}+1\big)\epsilon\\
    &+(\Delta\lambda\gamma( \inorm{\vec w}+1))^{t}\big(\inorm{\vec x}+\inorm{\vec b}+1)\epsilon\\
              &+4^{t-1}\big(\Delta\lambda\gamma( \inorm{\vec w}+1)\big)^{t}\big(\inorm{\vec x}+\inorm{\vec b}+1\big)\epsilon\\
                 &+\Delta\lambda\epsilon\\
    \le\;&4^{t}\big(\Delta\lambda\gamma( \inorm{\vec w}+1)\big)^{t}\big(\inorm{\vec x}+\inorm{\vec b}+1\big)\epsilon.
  \end{align*}
This proves \eqref{eq:47} and thus completes the proof of the lemma.
\end{proof}

\begin{lemma}\label{lem:ub3}
  Let $\beta\coloneqq(4\Delta\lambda\gamma)^d$, where $\gamma$ is the
  constant of Lemma~\ref{lem:ub1}. Let $\epsilon>0$. For all $v\in V$, let $\Fa_v':\Real\to\Real$ be an
  $\epsilon$-approximation of
  $\Fa_v$ that is Lipschitz continuous with constant $2\lambda$, and
  let $\FA'\coloneqq(V,E,(\Fa_v')_{v\in V})$. Then for
  all $\vec x\in\Real^p$, $\vec b\in\Real^{V}$,
  and $\vec w\in\Real^E$,
  \[
    \big\|\FA(\vec x,\vec b,\vec
    w)-\FA'(\vec x,\vec b,\vec
    w)\big\|_\infty\le
\beta \big(\inorm{\vec w}+1\big)^d\big(\inorm{\vec x}+\inorm{\vec b}+1\big)\epsilon.
  \]
\end{lemma}

\begin{proof}
  We shall prove by induction on $t$ that for all nodes $v\in V$ of
  depth $t$ we have
  \begin{equation}
    \label{eq:53}
    \big|f_{\FA,v}(\vec x,\vec b,\vec
    w)-f_{\FA',v}(\vec x,\vec b,\vec
    w)\big|\le (4\Delta\lambda\gamma)^t\big(\inorm{\vec w}+1\big)^t\big(\inorm{\vec x}+\inorm{\vec b}+1\big)\epsilon. 
  \end{equation}
 This yields the assertion of
  the lemma.

  Nodes of depth $t=0$ are input nodes, and we have
  \begin{equation*}
    \big|f_{\FA,X_i}(\vec x,\vec b,\vec
    w)-f_{\FA',X_i}(\vec x,\vec b,\vec
    w)\big|=|x_i-x_i|=0.
  \end{equation*}
  For the inductive step, let $v\in V$ be a node of depth $t>0$, and
  let $v_1,\ldots,v_k$ be its in-neighbours. Let $b\coloneqq
  b_v$ and $w_i\coloneqq w_{v_iv}$ for $i\in[k]$. Moreover, let  $\vec
  y=(y_1,\ldots,y_k)$ and $\vec
  y'=(y_1',\ldots,y_k')$ with
  $y_i\coloneqq f_{\FA,v_i}(\vec x,\vec b,\vec w)$ and
  $y_i'\coloneqq f_{\FA',v_i}(\vec x,\vec b,\vec w)$.

  \begin{techclaim}\label{claim:f:1}
    \[
      \big|f_{\FA,v}(\vec x,\vec b,\vec
    w)-f_{\FA',v}(\vec x,\vec b,\vec
    w)\big|\le 2
     \gamma^t\big(\|\vec w\|+1\big)^t\big(\|\vec x\|_\infty+\inorm{\vec b}+1)\epsilon
    +2\Delta\lambda\inorm{\vec w}\inorm{\vec y-\vec y'}
   \]
  
  \end{techclaim}
\begin{subproof}
   We have
   \begin{align}
     \notag
    \big|f_{\FA,v}(\vec x,\vec b,\vec
    w)-f_{\FA',v}(\vec x,\vec b,\vec
     w)\big|&=\Big|\Fa_v\Big(b+\sum_{i=1}^kw_iy_i\Big)-\Fa'_v\Big(b+\sum_{i=1}^kw_iy'_i\Big)\Big|\\
     \label{eq:54}
            &\le\Big|\Fa_v\Big(b+\sum_{i=1}^kw_iy_i\Big)-\Fa'_v\Big(b+\sum_{i=1}^kw_iy_i\Big)\Big|\\
     \notag
    &\hspace{\widthof{=}}+\Big|\Fa_v'\Big(b+\sum_{i=1}^kw_iy_i\Big)-\Fa'_v\Big(b+\sum_{i=1}^kw_iy'_i\Big)\Big|
  \end{align}
  By Lemma~\ref{lem:ub1} we have
   \begin{equation}
     \label{eq:55}
     \Big|\Fa_v\Big(b+\sum_{i=1}^kw_iy_i\Big)\Big|=\big|f_{\FA,v}(\vec x,\vec b,\vec
    w)\big|\le \gamma^t\big(\|\vec w\|+1\big)^t\big(\|\vec x\|_\infty+\inorm{\vec b}+1).
   \end{equation}
   Since $\Fa'_v$ $\epsilon$-approximates $\Fa_v$, this implies
   \begin{align}
     \notag
     \Big|\Fa_v\Big(b+\sum_{i=1}^kw_iy_i\Big)-\Fa'_v\Big(b+\sum_{i=1}^kw_iy_i\Big)\Big|&\le\epsilon
     \gamma^t\big(\|\vec w\|+1\big)^t\big(\|\vec x\|_\infty+\inorm{\vec b}+1)+\epsilon\\
     \label{eq:56}
     &\le 2
     \gamma^t\big(\|\vec w\|+1\big)^t\big(\|\vec x\|_\infty+\inorm{\vec b}+1)\epsilon.
   \end{align}
   Furthermore, by the Lipschitz continuity of $\Fa'_v$ we have
   \begin{equation}
      \label{eq:57}
     \Big|\Fa_v'\Big(b+\sum_{i=1}^kw_iy_i\Big)-\Fa'_v\Big(b+\sum_{i=1}^kw_iy'_i\Big)\Big|\le2\lambda\sum_{i=1}^k|w_i|\cdot|y_i-y_i'|
    \le 2\Delta\lambda\inorm{\vec w}\inorm{\vec y-\vec y'}.
   \end{equation}
  The assertion of the claim from \eqref{eq:54}, \eqref{eq:56}, and \eqref{eq:57}.
  \uend
  \end{subproof}

  By the inductive hypothesis \eqref{eq:53}, we have
  \begin{equation*}
    \inorm{\vec y-\vec y'}\le
 (4\Delta\lambda\gamma)^{t-1}\big(\|\vec
 w\|_\infty+1\big)^{t-1}\big(\|\vec x\|_\infty+\inorm{\vec b}+1\big)\epsilon. 
  \end{equation*}
  Thus by the claim,
  \begin{align*}
    \big|f_{\FA,v}(\vec x,\vec b,\vec
    w)-f_{\FA',v}(\vec x,\vec b,\vec
    w)\big|
    &\le 2
     \gamma^t\big(\|\vec w\|+1\big)^t\big(\|\vec x\|_\infty+\inorm{\vec b}+1)\epsilon
      +2\Delta\lambda\inorm{\vec w}\inorm{\vec y-\vec y'}\\
    &\le 2\gamma^t\big(\|\vec w\|+1\big)^t\big(\|\vec
      x\|_\infty+\|\vec b\|_\infty+1)\epsilon\\
    &\hspace{\widthof{$\le$}}
      + 2\cdot4^{t-1}(\Delta\lambda\gamma)^t \big(\inorm{\vec w}+1\big)^{t}\big(\inorm{\vec x}+\inorm{\vec b}+1\big)\epsilon\\
    &\le (4\gamma\Delta\lambda)^t \big(\inorm{\vec w}+1\big)^{t}\big(\inorm{\vec x}+\inorm{\vec b}+1\big)\epsilon.
  \end{align*}
  This proves \eqref{eq:53} and hence the lemma.
\end{proof}

\begin{lemma}\label{lem:lip-app}
  Let $f:\Real\to\Real$ be Lipschitz continuous with constant
  $\lambda>0$ such that $f$ is rpl approximable. Then for every
  $\epsilon>0$ there is a rational piecewise linear function $L$ of
  bitsize polynomial in $\epsilon^{-1}$ such that $L$ is an
  $\epsilon$-approximation of $f$ and $L$ is Lipschitz continuous with
  constant $(1+\epsilon)\lambda$.
\end{lemma}

\begin{proof}
  Let $0<\epsilon\le 1$ and $\epsilon'\coloneqq\frac{\epsilon}{10}$. Let $L'$ be a piecewise linear
  $\epsilon'$-approximation of $f$ of bitsize polynomial in
  $\epsilon^{-1}$. Let $t_1<\ldots<t_{n}$ be the thresholds of $L'$, and
  let $a_0,\ldots,a_{n}$ and $b_0,\ldots,b_{n}$ be its slopes
  and constants. Then $|a_0|\le(1+\epsilon)\lambda$; otherwise the
  slope of the linear function $a_0x+b_0$ would be too large (in
  absolute value) to approximate the function $f$ whose slope is
  bounded by $\lambda$. For the same reason,
  $|a_n'|\le(1+\epsilon)\lambda$.

  Let $s\coloneqq t_1$ and $s'\coloneqq t_n$. We subdivide the
  interval $[s,s']$ into sufficiently small subintervals (of length at
  most $\epsilon'\cdot\lambda^{-1}$). Within each such interval, $f$
  does not change much, because it is Lipschitz continuous, and we can
  approximate it sufficiently closely by a linear function with
  parameters whose bitsize is polynomially bounded in
  $\epsilon^{-1}$. The slope of theses linear functions will not be
  significantly larger than $\lambda$, because the slope of $f$ is
  bounded by $\lambda$. We can combine all these linear pieces with
  the linear functions $a_0x+b_0$ for the interval $(-\infty,s]$ and
  $a_nx+b_n$ for the interval $[s',\infty)$ to obtain the desired
  piecewise linear approximation of $f$.
\end{proof}

Now we are ready to prove the main result of this subsection.

\begin{lemma}\label{lem:ub4}
  For every $d\in\PNat$ there is a polynomial $\pi'(X,Y)$ such that the
  following holds.
  Let $\FF=\big(V,E,(\Fa_v)_{v\in V},\vec w,\vec b\big)$ be an
  rpl-approximable FNN
  architecture of depth $d$. Let
  $\epsilon>0$. Then there exists a rational piecewise-linear FNN
  $\FF'=\big(V,E,(\Fa'_v)_{v\in V},\vec w',\vec b'\big)$ of size at most $\pi'\big(\epsilon^{-1},\wt(\FF)\big)$
  such that for all $v\in V$ it holds that $\lambda(\Fa'_v)\le2\lambda(\Fa_v)$ and for all $\vec x\in\Real^p$ it holds that
  \[
    \inorm{\FF(\vec x)-\FF'(\vec x)}\le \big(\inorm{\vec
      x}+1\big)\epsilon.
  \]
\end{lemma}

Note that $\FF'$ has the same skeleton as $\FF$.

\begin{proof}
  Without loss of generality we assume $\epsilon\le 1$.
  Let $\FA\coloneqq\big(V,E,(\Fa_v)_{v\in
    V}\big)$. Define the parameters $\Delta,\lambda,\mu$ with
  respect to $\FA$ as before. Note that
  $\Delta,\lambda,\mu\le\wt(\FF)$.
  Choose the constants $\gamma$ according to Lemma~\ref{lem:ub1}, $\nu$ according to
  Lemma~\ref{lem:ub2}, and $\beta$ according to
  Lemma~\ref{lem:ub3} and note that for fixed $d$ they depend
  polynomially on $\Delta,\lambda,\mu$ and hence on $\wt(\FF)$.

  Let
  $\alpha\coloneqq 2\nu\big(\inorm{\vec w}+1\big)^d \big(\inorm{\vec
    b}+1\big)$.  Let $\vec w'\in\DRat^E$, $\vec b'\in\DRat^V$ such
  that $\|\vec w-\vec w'\|_\infty\le\frac{\epsilon}{\alpha}$ and
  $\|\vec b-\vec b'\|_\infty\le\frac{\epsilon}{\alpha}$. Clearly, we
  can choose such $\vec w'=(w_e')_{e\in E}$ and
  $\vec b'=(b'_v)_{v\in V}$ such that all their entries have bitsize
  polynomial in $\frac{\alpha}{\epsilon}$, which is polynomial in
  $\epsilon^{-1}$ and in $\wt(\FF)$.  Then by Lemma~\ref{lem:ub2},
  for all $\vec x\in\Real^p$ we have
  \[
    \inorm{\FA(\vec x,\vec b,\vec
    w)-\FA(\vec x,\vec b',\vec
    w')}\le
    \nu \big(\inorm{\vec
    w}+1\big)^d\big(\inorm{\vec x}+\inorm{\vec
    b}+1\big)\frac{\epsilon}{\alpha}\le\big(\inorm{\vec
    x}+1\big)\frac{\epsilon}{2},
  \]
  Let $\alpha'\coloneqq 2\beta \big(\inorm{\vec w'}+1\big)^d
  \big(\inorm{\vec b'}+1\big)$.
  For every $v\in V$, we let $\Fa'_v$ be a rational piecwise-linear
  function of bitsize polynomial
  in in $\epsilon^{-1}$ that is an
  $\frac{\epsilon}{\alpha'}$-approximation of $\Fa_v$ and Lipschitz
  continuous with constant $2\lambda$. Such an $\Fa'_v$ exists by
  Lemma~\ref{lem:lip-app}, because $\Fa_v$
  is rpl-approximable and Lipschitz continuous with constant $\lambda$. Let $\FA'\coloneqq\big(V,E,(\Fa'_v)_{v\in
    V}\big)$. By Lemma~\ref{lem:ub3}, for all $\vec x\in\Real^p$  we have
  \[
    \inorm{\FA(\vec x,\vec b',\vec
    w')-\FA'(\vec x,\vec b',\vec
    w')}\le
    \beta \big(\inorm{\vec
    w'}+1\big)^d\big(\inorm{\vec x}+\inorm{\vec
    b'}+1\big)\frac{\epsilon}{\alpha'}\le
  \big(\inorm{\vec x}+1\big)\frac{\epsilon}{2}.
\]
Overall,
  \begin{align*}
    \inorm{\FF(\vec x)-\FF'(\vec x)}
    &= \inorm{\FA(\vec x,\vec b,\vec w)-\FA'(\vec x,\vec b',\vec w')}\\
    &\le \inorm{\FA(\vec x,\vec b,\vec
    w)-\FA(\vec x,\vec b',\vec
    w')}+\inorm{\FA(\vec x,\vec b',\vec
    w')-\FA'(\vec x,\vec b',\vec
      w')}\\
    &\le \big(\inorm{\vec x}+1\big)\epsilon.
  \end{align*}
    
\end{proof}

\subsection{Bounds and Approximations for GNNs}

We start with a more explicit version of Lemmas \ref{lem:gnnbound1} and~\ref{lem:gnnbound2},
the growth bounds for GNN layers. Recall that the depth of a GNN layer
is the maximum of the depths of the FNNs for the combination and the
message function. 

\begin{lemma}\label{lem:gnnbound1a}
    For every $d\in\PNat$ there is a polynomial $\pi(X)$ such that the following
  holds. Let $\FL$ be a GNN layer of depth $d$, and let $p$
  be the input dimension of $\FL$. Then for all
  graphs $G$ and all signals $\Cx\in\CS_p(G)$
  we have
  \begin{align}
    \label{eq:58}
    \inorm{\tilde\FL(G,\Cx)}\le \pi\big(\wt(\FL)\big)\big(\inorm{\Cx}+1\big)|G|.
  \end{align}
\end{lemma}

\begin{proof}
  The proof of Lemma~\ref{lem:gnnbound1} yields
  \[
    \inorm{\tilde\FL(G,\Cx)(v)}\le
    2\gamma_\msg\gamma_\comb\big(\inorm{\Cx}+1\big)|G|,
  \]
  where $\gamma_\msg$,
  $\gamma_\comb$ are growth bounds for the message and
  combination functions of $\FL$. It follows from Lemma~\ref{lem:ub1}
  that $\gamma_\msg$,
  $\gamma_\comb$ can be chosen polynomial in the weight of $\FL$.
\end{proof}

\begin{lemma}\label{lem:gnnbound2a}
  For every $d\in\PNat$ there is a polynomial $\pi(X)$ such that the following
  holds. Let $\FL$ be a GNN layer of depth $d$, and let $p$
  be the input dimension of $\FL$. Then for all
  graphs $G$ and all signals $\Cx,\Cx'\in\CS_p(G)$
  we have
  \begin{align}
    \label{eq:59}
    \inorm{\tilde\FL(G,\Cx)-\tilde\FL(G,\Cx')}\le \pi\big(\wt(\FL)\big)\inorm{\Cx-\Cx'}|G|.
  \end{align}
\end{lemma}

\begin{proof}
  The proof of Lemma~\ref{lem:gnnbound2} yields
  \[
    \inorm{\tilde\FL(G,\Cx)(v)-\tilde\FL(G,\Cx')(v)}\le
    \lambda_\msg\lambda_\comb\inorm{\Cx-\Cx'}|G|,
  \]
  where $\lambda_\msg$,
  $\lambda_\comb$ are Lipschitz constants for the message and
  combination functions of $\FL$. It follows from Lemma~\ref{lem:ub1a}
  that $\lambda_\msg$,
  $\lambda_\comb$ can be chosen polynomial in the weight of $\FL$.
\end{proof}

\begin{corollary}\label{cor:gnnbound2a}
  For every $d\in\PNat$ there is a polynomial $\pi(X)$ such that the following
  holds. Let $\FN$ be a GNN of depth $d$, and let $p$
  be the input dimension of $\FN$. Then for all
  graphs $G$ and all signals $\Cx,\Cx'\in\CS_p(G)$
  we have
  \begin{align*}
    \inorm{\tilde\FN(G,\Cx)-\tilde\FN(G,\Cx')}\le
    \pi\big(\wt(\FL)\big)\inorm{\Cx-\Cx'}|G|^d.
  \end{align*}
\end{corollary}

\begin{lemma}\label{lem:gnnbound3}
  For every $d\in\PNat$ there exist polynomials $\pi(X)$ and
  $\pi'(X,Y)$ such that the following holds.  Let $\FL$ be an
  rpl-approximable GNN layer of depth $d$, and let $p$ be the input
  dimension of $\FL$. Then for all $\epsilon>0$ there exists a
  rational piecewise-linear GNN layer $\FL'$ of size at most
  $\pi'(\epsilon^{-1},\wt(\FL))$ with the same skeleton as $\FL$ such
  that the Lipschitz constants of the activation functions of $\FL'$
  are at most twice the Lipschitz constants of the corresponding
  activation functions in $\FL$ and for all graphs $G$, all signals
  $\Cx,\Cx'\in\CS_p(G)$, and all vertices $v\in V(G)$ it holds that
  \begin{align*}
    &\inorm{\tilde\FL(G,\Cx)-\tilde\FL'(G,\Cx')}
    \le
    \pi\big(\wt(\FL)\big)|G|\Big(\inorm{\Cx-\Cx'}
    +
    \big(\inorm{\Cx}+1\big)\epsilon\Big).
  \end{align*}
\end{lemma}

\begin{proof}
  Let $\msg:\Real^{2p}\to\Real^r$ and $\comb:\Real^{p+r}\to\Real^q$ be
  the message and combination functions of $\FL$, and let $\FF_\msg$
  and $\FF_\comb$ be the FNNs for these functions. By
  Lemma~\ref{lem:ub4} there are rational piecewise linear FNNs
  $\FF'_{\msg}$ and $\FF'_{\comb}$ of size polynomial in $\wt(\FF_\msg),\wt(\FF_\comb)\le\wt(\FL)$
  with activation functions of Lipschitz constants at most twice the Lipschitz constants of the
  corresponding activation functions in $\FF_\msg$, $\FF_\comb$ such that for all
  $\vec x,\vec x'\in\Real^p$ and $\vec z\in\Real^r$ we have
  \begin{align}
    \label{eq:60}
    \inorm{\msg(\vec x,\vec x')-\msg'(\vec x,\vec
    x')}&\le\big(\inorm{(\vec x,\vec x')}+1\big)\epsilon,\\
    \label{eq:61}
    \inorm{\comb(\vec x,\vec z)-\comb'(\vec x,\vec
    z)}&\le\big(\inorm{(\vec x,\vec z)}+1\big)\epsilon.
  \end{align}
  Let $\FL'$ be the GNN layer with message function $\msg'$,
  combination function $\comb'$, and the same aggregation function
  $\agg$ as $\FL$. By Lemma~\ref{lem:ub1} there is an $\alpha\in\PNat$
  that is polynomial in $\wt(\FF_\msg)$ and hence polynomial in
  $\wt(\FL)$ such that
  for all $\vec x,\vec x'\in\Real^p$ it holds that
  \begin{equation}
    \label{eq:62}
    \inorm{\msg(\vec x,\vec x')}\le\alpha\big(\max\big\{\inorm{\vec
      x},\inorm{\vec x'}\big\}+1\big).
  \end{equation}
  By Lemma~\ref{lem:ub1a} there is an $\alpha'\in\PNat$ polynomial in $\wt(\FF_\comb')$ and hence polynomial in
  $\wt(\FL)$ such that
  for all $\vec x,\vec x'\in\Real^{p},\vec z,\vec z'\in\Real^r$ it holds that
  \begin{equation}
    \label{eq:63}
    \inorm{\comb'(\vec x,\vec z)-\comb'(\vec x',\vec z')}\le\alpha'\max\big\{\inorm{\vec
      x-\vec x'},\inorm{\vec z-\vec z'}\big\}.
  \end{equation}
  By Lemma~\ref{lem:gnnbound2a} there is an $\alpha''\in\PNat$ polynomial in
  $\wt(\FL')$ and thus polynomial in $\wt(\FL)$ such
  that for all $G$ and $\Cx\in\CS_p(G)$,
  \begin{equation}
    \label{eq:64}
    \inorm{\tilde{\FL}'(G,\vec x)-\tilde{\FL}'(G,\vec x')}\le\alpha''\inorm{\Cx-\Cx'}|G|. 
  \end{equation}

  Let $G$ be a graph of order $n\coloneqq|G|$ and $v\in
  V(G)$. Furthermore, let $\Cx,\Cx'\in\CS_p(G)$ and $\Cy\coloneqq\tilde\FL(G,\Cx)$,
  $\Cy'\coloneqq\tilde\FL'(G,\Cx)$,
  $\Cy''\coloneqq\tilde\FL'(G,\Cx')$.  Then
  \[
    \inorm{\tilde\FL(G,\Cx)(v)-\tilde\FL'(G,\Cx')(v)}\le
    \inorm{\Cy(v)-\Cy'(v)}+\inorm{\Cy'(v)-\Cy''(v)}.
  \]
  By \eqref{eq:64} we have
  \begin{equation}
    \label{eq:65}
    \inorm{\Cy'(v)-\Cy''(v)}\le\alpha'' n\inorm{\Cx-\Cx'}.
  \end{equation}
  Thus we need to
  bound $\inorm{\Cy(v)-\Cy'(v)}$. Let
  \begin{align*}
    \Cz(v)&\coloneqq\agg\Big(\biglmulti \msg(\Cx(v),\Cx(w))
            \bigmid w\in N_G(v)\bigrmulti\Big),\\
    \Cz'(v)&\coloneqq\agg\Big(\biglmulti \msg'(\Cx(v),\Cx(w))
             \bigmid w\in N_G(v)\bigrmulti\Big).
  \end{align*}
  Then by \eqref{eq:60}, we have
  \begin{equation}
    \label{eq:66}
    \inorm{\Cz(v)-\Cz'(v)}\le n\left(\inorm{\Cx}+1\right)\epsilon.
  \end{equation}
  Furthermore, by \eqref{eq:62}, for all $w\in N(v)$ we have
  \[
    \inorm{\msg\big(\Cx(v),\Cx(w)\big)}
    \le\alpha\Big(\inorm{\Cx}+1\Big)
  \]
  Thus, since $\alpha\ge 1$ and $n\ge 1$,
  \begin{equation}
    \label{eq:67}
    \inorm{\big(\Cx(v),\Cz(v)\big)}=\max\Big\{\inorm{\Cx(v)},
    \inorm{\Cz(v)}\Big\}\le \alpha n\Big(\inorm{\Cx}+1\Big).
  \end{equation}
  Putting things together, we get
  \begin{align*}
    \inorm{\Cy(v)-\Cy'(v)}&=\inorm{\comb\big(\Cx(v),\Cz(v)\big)-\comb'\big(\Cx(v),\Cz'(v)\big)}\\
    &\le
      \inorm{\comb\big(\Cx(v),\Cz(v)\big)-\comb'\big(\Cx(v),\Cz(v)\big)}\\
    &\hspace{\widthof{$\le$ }}
      +\inorm{\comb'\big(\Cx(v),\Cz(v)\big)-\comb'\big(\Cx(v),\Cz'(v)\big)}\\
                          &\le
                            \big(\inorm{(\Cx(v),\Cz(v))}+1)\epsilon&\text{by
    \eqref{eq:61}}\\
     &\hspace{\widthof{$\le$
       }}+\alpha'\inorm{\Cz(v)-\Cz'(v)}&\text{by \eqref{eq:63}}\\
    &\le
      2\alpha n\Big(\inorm{\Cx}+1\Big)\epsilon+\alpha' n\big(\inorm{\Cx}+1\big)\epsilon&\text{by
                                                    \eqref{eq:67} and
                                                                                 \eqref{eq:66}}\\
    &\le
      (2\alpha+\alpha')n\big(\inorm{\Cx}+1\big)\epsilon
  \end{align*}
  Combined with \eqref{eq:65}, this yields the assertion of the lemma.
  \end{proof}

Now we are ready to prove the main lemma of this section.

\begin{lemma}\label{lem:gnnbound4}
  For every $d\in\PNat$ there exist a polynomial $\pi(X,Y)$ such that
  the following holds.  Let $\FN$ be an rpl-approximable GNN of depth
  $d$, and let $p$ be the input dimension of $\FN$. Then for all
  $\epsilon>0$ there exists a rational piecewise-linear GNN $\FN'$ of
  size at most $\pi(\epsilon^{-1},\wt(\FN))$ with the same skeleton
  as $\FN$ such that the Lipschitz constants of the activation
  functions of $\FN'$ are at most twice the Lipschitz constants of the
  corresponding activation functions in $\FN$ and for all graphs $G$
  and all signals $\Cx\in\CS_p(G)$ it holds that
  \[
    \inorm{\tilde\FN(G,\Cx)-\tilde\FN'(G,\Cx)}\le |G|^d\big(\inorm{\Cx}+1\big)\epsilon.
  \]
\end{lemma}

\begin{proof}
  Suppose that $\FN=(\FL_1,\ldots,\FL_d)$. For every $t\in[d]$, let
  $p_{t-1}$ be the input dimension of $\FL_t$. Then $p=p_0$.
  By Lemma~\ref{lem:gnnbound1a} there is an
  $\alpha$ polynomial in $\wt(\FN)$ such that for all $t\in[d]$, $G$, and
  $\Cx\in\CS_{p_{t-1}}(G)$ we have
  \begin{equation}
    \label{eq:68}
    \inorm{\tilde\FL_t(G,\Cx)}\le\alpha |G| \big(\inorm{\Cx}+1\big).
  \end{equation}
  Let $\pi'(X)$ be the polynomial of Lemma~\ref{lem:gnnbound3},
  \[
    \alpha'\coloneqq\max_{t\in[d]}\pi'\big(\wt(\FL_t)\big),
  \]
  and 
  \begin{align*}
    \beta&\coloneqq 3\max\{\alpha,\alpha'\},\\
    \epsilon'&\coloneqq\frac{\epsilon}{\beta^d}.
  \end{align*}
  Note that $\epsilon'$ is polynomial in $\wt(\FN)$. For every $t\in[d]$, we apply
  Lemma~\ref{lem:gnnbound3} to $\FL_t$ and $\epsilon'$ and obtain a
  rational piecewise linear GNN layer 
  $\FL_t'$ such that for all graphs $G$ and all
  signals $\Cx,\Cx'\in\CS_{p_{t-1}}(G)$ we have
  \begin{equation}
    \label{eq:69}
    \inorm{\tilde\FL_t(G,\Cx)-\tilde\FL_t'(G,\Cx')}\le\alpha'|G|
    \Big(\inorm{\Cx-\Cx'}+\big(\inorm{\Cx}+1\big)\epsilon'\Big).
  \end{equation}

  Let $G$ be a graph of order $n\coloneqq|G|$, and
  $\Cx\in\CS_p(G)$. 
  Let $\Cx_0\coloneqq\Cx_0'\coloneqq\Cx$, and for $t\in[d]$, let
  $\Cx_t\coloneqq\tilde\FL_t(G,\Cx_{t-1})$ and
  $\Cx_t'\coloneqq\tilde\FL'_i(G,\Cx_{t-1}')$. We shall prove that
  for all $t\in\{0,\ldots,d\}$ we have
  \begin{align}
    \label{eq:70}
    \inorm{\Cx_t}&\le\beta^tn^t\big(\inorm{\Cx}+1\big),\\
    \label{eq:71}
    \inorm{\Cx_t-\Cx_t'}&\le
    \beta^tn^t\big(\inorm{\Cx}+1\big)\epsilon'.
  \end{align}
  Since
  $\beta^d\epsilon'=\epsilon$ and $\Cx_d=\tilde\FN(G,\Cx)$,
  $\Cx_d'\coloneqq\tilde\FN'(G,\Cx)$, \eqref{eq:71} implies the
  assertion of the lemma.
                   
  We prove \eqref{eq:70} and \eqref{eq:71} by induction on $t$. The
  base step $t=0$ is trivial, because $\Cx_0=\Cx$ and
  $\Cx_0=\Cx'$. For the inductive step, let $t\ge 1$.
  By \eqref{eq:68} and the induction hypothesis we have
  \begin{align*}
    \inorm{\Cx_t}&\le\alpha n\big(\inorm{\Cx_{t-1}}+1\big)\\
    &\le \alpha
      n\Big(\beta^{t-1}n^{t-1}\big(\inorm{\Cx}+1\big)+1\Big)\\
    &\le \beta^tn^t \big(\inorm{\Cx}+1\big),
  \end{align*}
  where the last inequality holds because $2\alpha\le \beta$.
  This proves \eqref{eq:70}.
  
  By \eqref{eq:69} we have
  \begin{equation}
    \label{eq:72}
    \inorm{\Cx_t-\Cx_t'}\le\alpha' n 
                          \Big(\inorm{\Cx_{t-1}-\Cx_{t-1}'}+\big(\inorm{\Cx_{t-1}}+1\big)\epsilon'\Big).
  \end{equation}
  By induction hypothesis \eqref{eq:71},
  \begin{align}
    \notag
    \alpha'n
    \inorm{\Cx_{t-1}-\Cx_{t-1}'}&\le\alpha'n\beta^{t-1}n^{t-1}\big(\inorm{\Cx}+1\big)\epsilon'\\
\label{eq:73}
    &\le \frac{1}{3}\beta^{t}n^t \big(\inorm{\Cx}+1\big)\epsilon'.
  \end{align}
  By induction hypothesis \eqref{eq:70},
  \begin{align}
    \notag
    \alpha'n \big(\inorm{\Cx_{t-1}}+1\big)\epsilon'&\le
    \alpha'n
                                                     \Big(\beta^{t-1}n^{t-1}\big(\inorm{\Cx}+1\big)+1\Big)\epsilon'\\
    \label{eq:74}
    &\le \frac{2}{3}\beta^{t}n^t \big(\inorm{\Cx}+1\big)\epsilon'
  \end{align}
  Plugging \eqref{eq:73} and \eqref{eq:74} into \eqref{eq:72}, we obtain
  the desired inequality \eqref{eq:71}.
\end{proof}

\subsection{Proof of Theorem~\ref{theo:snonuniform}}

Let us first remark that we cannot directly apply
Theorem~\ref{theo:uniform} (the ``uniform theorem'') to a family of rational
piecewise linear GNN approximating the GNNs in our family $\CN$. The reason is that in Theorem~\ref{theo:uniform} the GNN is
``hardwired'' in the formula, whereas in our non-uniform setting we
obtain a different GNN for every input size. Instead, we encode the
sequence of rational piecewise linear GNNs approximating the GNNs in
the original family into the numerical
built-in relations. Then our formula evaluates these GNNs directly on
the numerical side of the structures.

\begin{proof}[Proof of Theorem~\ref{theo:snonuniform}]
  Let $\CN=(\FN^{(n)})_{n\in\Nat}$. Furthermore, let $d$ be an upper
  bound on the depth of all the $\FN^{(n)}$. Without loss of
  generality we assume that every $\FN^{(n)}$ has exactly $d$ layers
  $\FL^{(n)}_1,\ldots,\FL^{(n)}_d$. For $t\in[d]$, let $p^{(n)}_{t-1}$
  and $p^{(n)}_t$ be
  the input and output dimension of $\FL_t^{(n)}$. Then
  $p^{(n)}\coloneqq p^{(n)}_0$ is the input dimension of $\FN^{(n)}$
  and $q^{(n)}\coloneqq p_d^{(n)}$ is the output dimension. By the definition
  of the weight of a GNN, the
  $p^{(n)}_t$ are polynomially bounded in $n$.
  Let
  $\lambda^{(n)}\in\Nat$ be a Lipschitz constant for all activation
  functions in $\FN^{(n)}$. By the definition of the weight of an FNN
  and GNN we may choose $\lambda^{(n)}$ polynomial in $\wt(\FN^{(n)})$
  and thus in $n$.  

  For all $k,n\in\PNat$, we let $\FN^{(n,k)}$ be a rational piecewise
  linear GNN with the same skeleton as $\FN^{(n)}$ of size polynomial
  in $\wt(\FN^{(n)})$, hence also polynomial in $n$, and $k$ such that
  all activation functions of $\FN^{(n,k)}$ have Lipschitz constant at
  most $2\lambda^{(n)}$, and
  for all graphs $G$ of order $n$ and all
  signals $\Cx\in\CS_p(G)$ it holds that
  \begin{equation}
    \label{eq:75}
    \inorm{\tilde\FN^{(n)}(G,\Cx)-\tilde\FN^{(n,k)}(G,\Cx)}\le
    \frac{n^d}{k}\big(\inorm{\Cx}+1\big). 
  \end{equation}
  Such an $\FN^{(n,k)}$ exists by Lemma~\ref{lem:gnnbound4}. Let
  $\FL^{(n,k)}_1,\ldots,\FL^{(n,k)}_d$ be the layers of $\FN^{(n,k)}$.

  We want to describe the GNNs $\FN^{(n,k)}$ with built-in relations,
  using an encoding similar to F-schemes. We
  cannot just use the same encoding as for the F-schemes because the
  non-uniform logic $\FOC_{nu}$ does not allow for
  built-in numerical functions.\footnote{The reader may wonder why
    we do not simply allow for built-in functions to avoid
    this difficulty. The reason is that then we could built terms
    whose growth is no longer polynomially bounded in the size of the
    input structure, which would make the logic too powerful. In
    particular, the logic would no longer be contained in $\TC^0$.}

  Let $t\in[d]$. In the following, we define the built-in relations
  that describe the $t$th layer $\FL_t^{(n,k)}$ of all the
  $\FN^{(n,k)}$. We need to describe
  FNNs $\FF^{(n,k)}_\msg$ and $\FF^{(n,k)}_{\comb}$ for the message and combination
  functions of the layers, and in addition we need to describe the
  aggregation function. For the aggregation functions we
  use three relations
  $A^{\SUM}_t,A^{\MEAN}_t,A^{\MAX}_t\subseteq\Nat^2$,
  where
  \[
    A^{\SUM}_t\coloneqq\big\{(n,k)\in\Nat^2\bigmid\Fa^{(n,k)}_t=\SUM\big\},
  \]
  and $A^{\MEAN}_t$ and $A^{\MAX}_t$ are defined similarly.
  For each of the two FNNs
  $\FF^{(n,k)}_\msg$ and $\FF^{(n,k)}_{\comb}$ we use $18$ relations.
  We only describe the encoding of $\FF^{(n,k)}_\msg$ using relations
  $M_t^1,\ldots,M_t^{18}$. The encoding of
  $\FF^{(n,k)}_\comb$ is analogous using a fresh set of 18 relations $C_t^1,\ldots,C_t^{18}$.

  Say,
  \[
    \FF^{(n,k)}_\msg=\Big(V^{(n,k)},E^{(n,k)},(\Fa_v^{(n,k)})_{v\in
      V^{(n,k)}}, (w_e^{(n,k)})_{e\in E^{(n,k)}}, (b_v^{(n,k)})_{v\in
      V^{(n,k)}}\Big),
  \]
  where without loss of generality we assume that
  $V^{(n,k)}$ is an initial segment of $\Nat$.
  We use relation $M_t^{1}\subseteq\Nat^3$ and $M_t^2\subseteq\Nat^4$
  to describe the vertex set $V$ and
  the edge set $E$, letting
  \begin{align*}
    M_t^1&\coloneqq\big\{(n,k,v)\bigmid v\in V^{(n,k)}\big\},\\
    M_t^2&\coloneqq\big\{(n,k,u,v)\bigmid (u,v)\in E^{(n,k)}\big\}.
  \end{align*}
  As the bitsize of the skeleton of $\FN^{(n)}$ and hence
  $|V^{(n,k)}|$ is polynomially bounded in $n$, there is an
  arithmetical term $\theta_V(y,y')$ such that for all graphs $G$ and
  all assignments $\Ca$ we have
  \[
    \sem{\theta_V}^{(G,\Ca)}(n,k)=|V^{(n,k)}|
  \]
 This term uses the constant $\ord$ as well as the built-in relation $M_t^1$. It
  does not depend on the graph $G$ or the assignment $\Ca$.

  For the weights, we use the relations
  $M_t^3,M_t^4,M_t^5\subseteq\Nat^5$, letting
  \begin{align*}
    M_t^3&\coloneqq\big\{(n,k,u,v,r)\bigmid
                 e\coloneqq (u,v)\in E^{(n,k)}\text{ with
                 }w_e^{(n,k)}=(-1)^{r}2^{-s}m\big\},\\
    M_t^4&\coloneqq\big\{(n,k,u,v,s)\bigmid
                 e\coloneqq(u,v)\in E^{(n,k)}\text{ with
                 }w_e^{(n,k)}=(-1)^{r}2^{-s}m\big\},\\
    M_t^5&\coloneqq\big\{(n,k,u,v,i)\bigmid
                 e\coloneqq (u,v)\in E^{(n,k)}\text{ with
                 }w_e^{(n,k)}=(-1)^{r}2^{-s}m\text{ and }\Bit(i,m)=1\big\}.                 
  \end{align*}
  We always assume that $w_e^{(n,k)}=(-1)^{r}2^{-s}m$ is the
  \emph{canonical representation} of $w_e^{(n,k)}$ with $r=0,m=0,s=0$
  or $r\in\{0,1\}$ and $s=0$ and $m\neq 0$ or $r\in\{0,1\}$ and $s\neq0$ and $m\neq 0$ odd. As the bitsize
  of the numbers $w_e^{(n,k)}$ is polynomial in $k$,
  we can easily construct an arithmetical r-expression
  $\rexp_w(y,y',z,z')$ such that for all graphs
  $G$ and assignments $\Ca$ and $n,k,u,v\in\Nat$ such that $e\coloneqq
  (u,v)\in
  E^{(n,k)}$, 
  \[
    \num{\rexp_w}^{(G,\Ca)}(n,k,u,v)=w_e^{(n,k)}.
  \]
  This r-expression depends on the built-in relations $M_t^{i}$, but not on
  $G$ or $\Ca$.

  Similarly, we define the relations  $M_t^{6},M_t^7,M_t^8\subseteq\Nat^4$ for the biases and
  an r-expression
  $\rexp_b(y,y',z)$ such that for all graphs
  $G$ and assignments $\Ca$ and $n,k,v\in\Nat$ with $v\in
  V^{(n,k)}$, 
  \[
    \num{\rexp_b}^{(G,\Ca)}(n,k,v)=b_v^{(n,k)}.
  \] 
  To store the activation
  functions $\Fa^{(n,k)}_v$ we use the remaining ten relations
  $M_t^{9}\subseteq\Nat^4$, $M_t^{10},\ldots,M_t^{18}\subseteq\Nat^5$.
  The relation $M_t^{9}$ is used to store the number
  $m^{(n,k)}_v$ of
  thresholds of $\Fa^{(n,k)}_v$:
  \[
    M_t^{9}\coloneqq\{(n,k,v,m^{(n,k)}_v)\bigmid v\in V^{(n,k)}\big\}.
  \]
  As the bitsize of $\Fa_v^{(n,k)}$ is polynomial in $k$, the number
  $m^{(n,k)}_v$ is bounded by a polynomial in $k$. Thus we can
  construct an arithmetical term $\theta_{\Fa}(y,y',z)$ such that for
  all graphs $G$, all assignments $\Ca$, all $n,k\in\Nat$, and all
  $v\in V^{(n,k)}$ we have
  \[
    \sem{\theta_\Fa}^{(G,\Ca)}(n,k,v)=m^{(n,k)}_v
  \]
  Of course this term needs to use the built-in relation $M_t^9$. It
  does not depend on the graph $G$ or the assignment $\Ca$.
  
  The relations $M_t^{10},N_t^{11},
  N_t^{12}$ are used to store thresholds. Say, the
  thresholds of $\Fa_v^{(n,k)}$ are
  $t_{v,1}^{(n,k)}<\ldots<t_{v,m}^{(n,k)}$, where $m=m^{(n,k)}_v$. We let
  \begin{align*}
    N_t^{10}&\coloneqq\big\{(n,k,v,i,r)\bigmid
                 v\in V^{(n,k)},1\le i\le m^{(n,k)}_v\text{ with
                 }t_{v,i}^{(n,k)}=(-1)^{r}2^{-s}m\big\},\\
    N_t^{11}&\coloneqq\big\{(n,k,v,i,s)\bigmid
                 v\in V^{(n,k)},1\le i\le m^{(n,k)}_v\text{ with
                 }t_{v,i}^{(n,k)}=(-1)^{r}2^{-s}m\big\},\\
    N_t^{12}&\coloneqq\big\{(n,k,v,i,j)\bigmid
                 v\in V^{(n,k)},1\le i\le m^{(n,k)}_v\text{ with
                 }t_{v,i}^{(n,k)} =(-1)^{r}2^{-s}m\text{ and }\Bit(j,m)=1\}.                 
  \end{align*}
  We always assume that $t_{v,i}^{(n,k)} =(-1)^{r}2^{-s}m$ is the
  canonical representation of $t_{v,i}^{(n,k)}$.
As the bitsize
  of $\Fa_v^{(n,k)}$ is polynomial in $k$,
  we can construct an arithmetical r-expression
  $\rexp_{t}(y,y',z,z')$  such that for all graphs
  $G$ and assignments $\Ca$ and $n,k,v,i\in\Nat$ with $v\in
  V^{(n,k)}$, $i\in[m^{(n,k)}_v]$, 
  \[
    \num{\rexp_t}^{(G,\Ca)}(n,k,v,i)=t_{v,i}^{(n,k)}.
  \] 
  Similarly, we use the relations $N_t^{13},N_t^{14},
  N_t^{15}$ to represent the slopes of $\Fa_v^{(n,k)}$
  and the relations $N_t^{16},N_t^{17},
  N_t^{18}$ to represent the constants. Furthermore, we construct arithmetical r-expressions
  $\rexp_{s}(y,y',z,z')$ and $\rexp_{c}(y,y',z,z')$ to access them.
We can combine the term $\theta_\Fa$ and the r-expressions
$\rexp_{t}(y,y',z,z')$, $\rexp_{s}(y,y',z,z')$, $\rexp_{c}(y,y',z,z')$
 to an L-expression $\boldsymbol{\chi}(y,y',z)$ such that for all graphs
  $G$, 
  all assignments $\Ca$, all $n,k\in\Nat$, and all $v\in V^{(n,k)}$ we
  have
  \[
    \Lin{\boldsymbol{\chi}}^{(G,\Ca)}(n,k,v)=\Fa_v^{(n,k)}.
  \]
  Then we can combine the term $\theta_V$, the relation $N^E_t$, the
  r-expressions $\rexp_w,\rexp_b$, and the L-expression
  $\boldsymbol{\chi}$ to an F-expression $\boldsymbol{\phi}^\msg_t(y,y')$ such
  that for all graphs
  $G$, 
  all assignments $\Ca$, and all $n,k\in\Nat$ we
  have
  \[
    \Fin{\boldsymbol{\phi}^\msg_t}^{(G,\Ca)}(n,k)=\FF_\msg^{(n,k)}.
  \]
  Similarly, we obtain an F-expression
  $\boldsymbol{\phi}^\comb_t(y,y')$ for the combination function which
  uses the relation $C_t^1,\ldots,C_t^{18}$ that represent
  $\FF_\comb^{(n,k)}$.

  In addition to the $\FN^{(n,k)}$, we also need access to the
  Lipschitz constants $\lambda^{(n)}$. We use one more built-in
  relation
  \[
    L\coloneqq\{(n,\lambda^{(n)})\bigmid n\in\Nat\}\subseteq\Nat^2.
  \]
  As $\lambda^{(n)}$ is polynomially bounded in $n$, there is a term
  $\theta_\lambda(y)$ using the built-in relation $L$ such that for all
  $G$, $\Ca$ and all $n$,
  \[
    \sem{\theta_\lambda}^{(G,\Ca)}(n)=\lambda^{(n)}.
  \]

  \begin{techclaim}\label{claim:g:1}
    For all $t\in[d]$, there is an arithmetical term $\eta_t(y,y')$ such that for all
    $n,k\in\Nat$ and all $G$, $\Ca$, the value $\sem{\eta_t}^{(G,a)}(n,k)$ is a
    Lipschitz constant for the combination function $\comb^{(n,k)}_t$
    of $\FL^{(n,k)}_t$.

   \end{techclaim}
\begin{subproof}
    Using Lemma~\ref{lem:ub1a} we can easily obtain such a Lipschitz
    constant using the fact that $2\lambda^{(n)}$ is a Lipschitz
    constant for all activation functions of
    $\comb^{(n,k)}_t$, and using the term $\theta_\lambda(y)$ to access
    $\lambda^{(n)}$ and the F-expression
    $\boldsymbol{\phi}^\comb_t(y,y')$ to access the FNN for
    $\comb^{(n,k)}_t$.
    \uend
  \end{subproof}

  \begin{techclaim}\label{claim:g:2}
    \label{page:claim3}
    For all $t\in[d]$, there is an arithmetical term $\zeta_t(y,y')$
    such that for all $n,k\in\Nat$, all $G$, $\Ca$, and all
    $\Cx,\Cx'\in \CS_{p_{t-1}}(G)$,
    \[
      \inorm{\tilde{\FL}^{(n,k)}_t(G,\vec
        x)-\tilde{\FL}^{(n,k)}_t(G,\vec x)}\le\sem{\zeta_t}^{(G,\Ca)}(n,k)n\inorm{\Cx-\Cx'}
    \]

    \end{techclaim}
\begin{subproof}
    To construct this term we use Lemma~\ref{lem:gnnbound2a}.
    \uend
  \end{subproof}

  Note that the bound $\sem{\zeta_t}^{(G,\Ca)}(n,k)$ is independent of
  the graph $G$ and the assignment $\Ca$ and only depends on
  $\tilde{\FL}^{(n,k)}_t$ and $\lambda^{(n)}$.

  The following claim is an analogue of Lemma~\ref{lem:uniform}.

  \begin{techclaim}\label{claim:g:3}
    Let $t\in[d]$. Let $\vec X$ be an r-schema of type
    $\rtp{\ttv\ttn}$, and let $W$ be a function variable of type
    $\emptytuple\to\ttn$. Then there is a guarded r-expression
    $\logic{l-eval}_{t}(y,y',x,y'')$ such that
    the following holds for all $n,k\in\Nat$, all graphs $G$, and all
    assignments $\Ca$ over $G$.  Let $\Cx\in\CS_{p^{(n)}_{t-1}}(G)$ be the signal
    defined by
    \begin{equation*}
      \Cx(v)\coloneqq\Big(\num{\vec X}^{(G,\Ca)}(v,0),\ldots,
      \num{\vec X}^{(G,\Ca)}(v,p^{(n)}_{t-1}-1)\Big)
    \end{equation*}
    and let $\Cy\coloneqq\tilde\FL^{(n,k)}_t(G,\Cx)\in\CS_{p^{(n)}_t}(G)$. Then for
    all $v\in V(G)$, 
    \begin{equation*}
      \inorm{\Cy(v)
        -\big(\num{\logic{l-eval}}^{(G,\Ca)}(n,k,v,0),\ldots, \num{\logic{l-eval}}^{(G,\Ca)}(n,k,v,p^{(n)}_t-1)\big)}
      \le2^{-\Ca(W)} 
    \end{equation*}
    
   \end{techclaim}
\begin{subproof} The proof is completely analogous to the proof of
    Lemma~\ref{lem:uniform}, except that in the proof of the analogues
    of Claims~\ref{claim:c:1} and
    \ref{claim:c:2} we use Lemma~\ref{lem:ar6a} instead of Corollary~\ref{cor:ar6a}
    to evaluate the FNNs computing the message function and
    combination function. We substitute suitable instantiations of the
    r-expression $\boldsymbol{\chi}(y,y')$ and the L-expression
    $\boldsymbol{\psi}(y,y')$ for the r-schemas $\vec Z_v,\vec Z_e$
    representing the parameters of the FNN and the L-schemas
    $\vec Y_v$ representing the activation functions in
    Lemma~\ref{lem:ar6a}.

    Furthermore, in Case~3 of the proof of Lemma~\ref{lem:uniform}
    (handling $\MEAN$-aggregation) we need a term that defines a
    Lipschitz constant for the combination function of
    $\FL^{(n,k)}_t$. (In the proof of Lemma~\ref{lem:uniform}, this is
    the constant $\lambda$.) We can use the term $\eta_t(y,y')$ of Claim~\ref{claim:g:1}.
    \uend
  \end{subproof}

The next claim is the analogue of Theorem~\ref{theo:uniform} for our
setting with built-in relations.

  \begin{techclaim}\label{claim:g:4}
     Let $\vec X$ be an r-schema of type
    $\rtp{\ttv\ttn}$, and let $W$ be a function variable of type
    $\emptytuple\to\ttn$. Then there is a guarded
  r-expression $\logic{n-eval}(y,y',x,y'')$ such
  that the following holds for all $n,k\in\Nat$, all graphs $G$, and all
  assignments
  $\Ca$ over $G$.
  Let $\Cx\in\CS_{p^{(n)}}(G)$ be the signal defined by
  \begin{equation*}
    \Cx(v)\coloneqq\Big(\num{\vec X}^{(G,\Ca)}(v,0),\ldots,
    \num{\vec X}^{(G,\Ca)}(v,p^{(n)}-1)\Big),
  \end{equation*}
  and let $\Cy\coloneqq\tilde\FN^{(n,k)}(G,\Cx)\in\CS_{q^{(n)}}(G)$. Then for
  all $v\in V(G)$, 
  \begin{equation}
    \label{eq:76}
     \inorm{\Cy(v)
    -\big(\num{\logic{n-eval}}^{(G,\Ca)}(n,k,v,0),\ldots, \num{\logic{n-eval}}^{(G,\Ca)}(n,k,v,q^{(n)}-1)\big)}
    \le2^{-\Ca(W)} 
  \end{equation}

 \end{techclaim}
\begin{subproof}
  The proof is completely analogous to the proof of
  Theorem~\ref{theo:uniform}, using Claim~\ref{claim:g:3} instead of
  Lemma~\ref{lem:uniform}. In the proof of Theorem~\ref{theo:uniform}
  we need access to a term that defines a constant $\lambda^{(t)}$ that
  bounds the growth of the $t$th layer of $\FN^{(n,k)}$. We use the
  term $\zeta_t(y,y')$ of Claim~\ref{claim:g:2}. 
  \uend
\end{subproof}

What remains to be done is choose the right $k$ to achieve the
desired approximation error in
\eqref{eq:38}. We will define $k$ using a closed term $\logic{err}$ that
depends on $W,W'$ as well as the order of the input graph. We let
\[
  \logic{err}\coloneqq2\ord\cdot(W+1)\cdot W'.
\]
In the following, let us assume that $G$ is a graph of order $n$ and
$\Ca$ is an assignment over $G$ satisfying the two assumptions
$\Ca(W')\neq0$ and $\inorm{\Cx}\le\Ca(W)$ for the signal
$\Cx\in\CS_{p^{(n)}}(G)$ defined by $\vec X$ as in \eqref{eq:37}. Let
\[
  k\coloneqq
\sem{\logic{err}}^{(G,\Ca)}=2n(\Ca(W)+1)\Ca(W')\ge2n\big(\inorm{\Cx}+1\big)\Ca(W')
\]
and thus
\begin{equation*}
  n\big(\inorm{\Cx}+1\big)\frac{1}{k}\le\frac{1}{2\Ca(W')}.
\end{equation*}
Thus by \eqref{eq:75},
\begin{equation}
  \label{eq:77}
  \inorm{\tilde\FN^{(n)}(G,\Cx)-\tilde\FN^{(n,k)}(G,\Cx)}\le 
  \frac{1}{2\Ca(W')}. 
\end{equation}
Now we let $\logic{gnn-eval}(x,y)$ be the formula obtained from the
formula $\logic{gnn-eval}(y,y',x,y'')$ of Claim~\ref{claim:g:4} by substituting
$\ord$ for $y$, $\logic{err}$
for $y'$, $y$ for $y''$, and $W'$ for $W$. Then by \eqref{eq:76},
we have
\begin{align*}
  &\inorm{\tilde\FN^{(n,k)}(G,\Cx)(v)
    -\big(\num{\logic{gnn-eval}}^{(G,\Ca)}(v,0),\ldots, \num{\logic{gnn-eval}}^{(G,\Ca)}(v,q^{(n)}-1)\big)}\\
  &\le2^{-\Ca(W')}\le \frac{1}{2\Ca(W')} 
\end{align*}
 Combined with \eqref{eq:77}, this yields
\[
  \inorm{\tilde\FN(G,\Cx)(v)
    -\big(\num{\logic{gnn-eval}}^{(G,\Ca)}(v,0),\ldots, \num{\logic{gnn-eval}}^{(G,\Ca)}(v,q^{(n)}-1)\big)}\le
  \frac{1}{\Ca(W')}, 
\]
that is, the desired inequality \eqref{eq:38}.
\end{proof}

\section{A Converse}
\label{sec:converse}

The main result of this section is a converse of Corollary~\ref{cor:snonuniform}.
For later reference, we
prove a slightly more general lemma that not only applies to queries over
labelled graphs, that is, graphs with Boolean signals, but actually to
queries over graphs with integer
signals within some range. 

\begin{lemma}\label{lem:converse}
  Let $U_1,\ldots,U_p$ be function variables of type $\ttv\to\ttn$, 
  and let $\phi(x)$ be a ${\GCnu}$-formula that contains no relation or
  function variables except possibly the $U_i$. Then there is a polynomial-size bounded-depth family $\CN$
  of rational piecewise-linear GNNs of input dimension $p$ such that
  for all graphs $G$ of order $n$ and all assignments $\Ca$ over $G$ the following
  holds. Let $\Cu\in\CS_p(G)$ be the signal defined by
  \begin{equation}
    \label{eq:78}
    \Cu(v)\coloneqq\big(\Ca(U_1)(v),\ldots,\Ca(U_p)(v)\big). 
  \end{equation}
  Assume that $\Ca(U_i)(v)< n$ for all $i\in[p]$ and $v\in V(G)$. Then for all $v\in V(G)$, $\tilde\CN(G,\Cu)\in\{0,1\}$ and
  \[
    \tilde\CN(G,\Cu)(v)=1\iff (G,\Ca)\models\phi(v).
  \]
  Furthermore, all GNNs in $\CN$ only use $\lsig$-activations and
  $\SUM$-aggregation. 
\end{lemma}

In the following, we will use the following more suggestive notation
for the setting of the lemma: for a signal $\Cu:V(G)\to\{0,\ldots,n-1\}^p$, we write
\[
  (G,\Cu)\models\phi(v)
\]
if
$(G,\Ca)\models\phi(v)$ for some assignments $\Ca$
with $\Ca(U_i)(v)=\Cu(v)_i$ for all $i\in[p]$ and $v\in V(G)$.
This is reasonable because $\phi$ only
depends on the assignment to $U_i$ and to
$x$. Thus if $(G,\Cu)\models\phi(v)$ then $(G,\Ca)\models\phi$ for all assignments $\Ca$
with with $\Ca(U_i)(v)=\Cu(v)_i$ and $\Ca(x)=v$.

To prove Lemma~\ref{lem:converse}, we need to construct an FNN that
transforms a tuple of nonnegative integers into a ``one-hot encoding''
of these integers mapping $i$ to the $\{0,1\}$-vector with a $1$ in
the $i$th position and $0$s everywhere else.

\begin{lemma}\label{lem:decode}
  Let $m,n\in\Nat$. Then there is an FNN\/ $\FF$\/ of input
  dimension $m$ and output dimension $m\cdot n$ such that for all
  $\vec x=(x_0,\ldots,x_{m-1})\in\{0,\ldots,n-1\}^m$ the following
  holds. Suppose that $\FF(\vec x)=\vec y=(y_0,\ldots,y_{mn-1})$. Then
  for $k=in+j$, $0\le i<m,0\le j<n$, 
  \[
    y_k=
    \begin{cases}
      1&\text{if }j=x_i,\\
      0&\text{otherwise}.
    \end{cases}
  \]
  Furthermore, $\FF$ has size $O(mn)$, depth $2$, and it only uses $\lsig$
  activations.
\end{lemma}

\begin{example}
  Suppose that $m=3,n=4$, and $\vec x=(1,3,0)$. Then we want $\FF(\vec
  x)$ to be
  \[
    (0,1,0,0,\;0,0,0,1,\;1,0,0,0).\uende
  \]
\end{example}

\begin{proof}[Proof of Lemma~\ref{lem:decode}]
  Observe that the function $f(x)\coloneqq\lsig(x+1)-\lsig(x)$ satisfies
  $f(0)=1$ and $f(k)=0$ for all integers $k\neq 0$.
  We design the network such that
  \[
    y_{im+j}=f(x_i-j)=\lsig(x_i-j+1)-\lsig(x_i-j).
  \]
  On the middle layer we compute the values $\lsig(x_i-j+1)$ and
  $\lsig(x_i-j)$ for all $i,j$.
\end{proof}

\begin{proof}[Proof of Lemma~\ref{lem:converse}]
  Let us fix an $n\in\Nat$.
  We need to define a GNN $\FN$ of size polynomial
  in $n$ and of depth only depending  on $\phi$, but not on $n$,
  such that for every graph $G$ of order $n$, every signal
  $\Cu:V(G)\to\{0,\ldots,n-1\}^p$, and every $v\in V(G)$ it holds that
  $\tilde\FN(G,\Cu)(v)\in\{0,1\}$ and
  \begin{equation}
    \label{eq:79}
      (G,\Cu)\models\phi(v)\iff \FN(G,\Cu)(v)=1.
    \end{equation}
    We will prove \eqref{eq:79} by induction on the formula
    $\phi(x)$. Let us understand the structure of this formula. The formula uses two vertex variables
    $x_1,x_2$. We usually refer to them as $x,x'$, with the
    understanding that if $x=x_i$ then $x'=x_{3-i}$. In addition, the
    formula uses an arbitrary finite set of number
    variables, say, $\{y_1,\ldots,y_\ell\}$. When
    we want to refer to any of these variables without specifying a
    particular $y_i$, we use the notations like $y,y',y^j$.

    By a slight extension of Lemma~\ref{lem:termbound} to a setting
    where we allow function variables, but require them to take values
    smaller than the order of the input structure, all subterms of
    $\phi(x)$ are polynomially bounded in the order $n$ of the input
    graph. Thus in every counting subterm
    $\#(x',y^1<\theta_1,\ldots,y^k<\theta_k).\big(E(x,x')\wedge
    \psi\big)$ or $\#(y^1<\theta_1,\ldots,y^k<\theta_k).  \psi$,
    we may replace the $\theta_i$ by a fixed closed term
    $\theta\coloneqq(\ord+1)^r$ for some constant $r\in\Nat$ and
    rewrite the counting terms as
    $\#(x',y^1<\theta,\ldots,y^k<\theta).\big(E(x,x')\wedge
    y^1<\theta_1\wedge\ldots \wedge y^k<\theta_k\wedge
    \psi\big)$ and $\#(y^1<\theta,\ldots,y^k<\theta).\big(y^1<\theta_1\wedge\ldots \wedge y^k<\theta_k\wedge
    \psi\big)$, respectively. We fix $\theta$ for the rest of the
    proof and assume that it is used as the bound in all counting
    terms appearing in $\phi$.

    Arguably the most important building blocks of the formula $\phi$
    are subterms of the form
    \begin{equation}
      \label{eq:80}
      \#(x',y^1<\theta,\ldots,y^k<\theta).
      \big(E(x,x')\wedge \psi\big). 
    \end{equation}
    Let us call these the \emph{neighbourhood terms} of $\phi$. Note
    that the only guards available in our vocabulary of graphs are
    $E(x,x')$ and $E(x',x)$, and since we are dealing with undirected
    graphs these two are equivalent. This is why we always assume that
    the guards $\gamma$ of subterms
    $\#(x',y_1<\theta,\ldots,y_k<\theta).(\gamma\wedge\ldots)$
    are of the form $E(x,x')$. Furthermore, we may assume that atoms
    $E(x,x')$ only appear as guards of neighbourhood terms. This
    assumption is justified by the observation that atoms $E(x,x')$
    must always occur within some neighbourhood term (otherwise both
    $x$ and $x'$ would occur freely in $\phi$), and it would make no
    sense to have either $E(x,x')$ or its negation in the subformula
    $\psi$ in \eqref{eq:80} unless it appeared within a neighbourhood
    term of $\psi$. We may also
    assume that $\phi$ has no atomic subformulas $E(x,x)$, because
    they always evaluate to false (in the undirected simple graphs we
    consider), or $x=x$, because they always evaluate to true.

    In the following, we will use the term ``expression'' to refer to
    subformulas and subterms of $\phi$, and we denote expressions by
    $\xi$. An \emph{vertex-free expression} is an expression
    with no free vertex variables. A
    \emph{vertex expression} is an expression with exactly one
    free vertex variable $x$ (so $\phi=\phi(x)$ itself is a vertex expression), and
    an \emph{edge expression} is an expression with two free
    vertex variables. This terminology is justified by the observation
    that edge expressions must be guarded, so the two
    free variables must be interpreted by the endpoints of an
    edge. The most important edge formulas are the formulas
    $E(x,x')\wedge\psi$ appearing within the neighbourhood terms.

   To simplify the presentation, let us a fix a graph $G$ of order $n$
    and a signal $\Cu:V(G)\to\{0,\ldots,n-1\}^p$ in the
    following. Of course the GNN we shall define will not depend on
    $G$ or $\Cu$; it will only depend on $n$ and $\phi$.

    Recall that all subterms of $\phi$ take values polynomially
    bounded in $n$. Let $m\in\Nat$ be polynomial in $n$
    such that all subterms of $\psi$ take values strictly smaller than
    $m$. Let $M\coloneqq\{0,\ldots,m-1\}$. When evaluating $\phi$, we
    only need to consider assignments that map number variables
    $y_1,\ldots,y_{\ell}$ appearing in $\phi$ to
    values in $M$. Then every vertex-free formula $\psi$ defines a
    relation $S_\psi\subseteq M^{\ell+1}$ consisting of all tuples
    $(b,a_1,\ldots,a_\ell)\in M^{\ell+1}$ such that
    $(G,\Ca)\models\psi$ for all assignments $\Ca$ over $G$ with
    $\Ca(y_j)=a_j$ for all $j\in[\ell]$. (The first coordinate $b$ is
    just a dummy coordinate that will allow us to work with
    $\ell+1$-ary relations throughout.) Every vertex-free term
    $\theta$ defines a relation $S_\theta\subseteq M^{\ell+1}$
    consisting of all tuples $(b,a_1,\ldots,a_\ell)\in M^{\ell+1}$
    such that $\sem{\theta}^{(G,\Ca)}=b$ for all assignments $\Ca$
    over $G$ with $\Ca(y_j)=a_j$ for all $j\in[\ell]$. Similarly,
    every vertex expression $\xi$ defines a relation
    $S_\xi(v)\subseteq M^{\ell+1}$ for every $v\in V(G)$ and every
    edge expression $\xi$ defines a relation
    $S_\xi(v,v')\subseteq M^{\ell+1}$ for every pair
    $(v,v')\in V(G)^2$.

    Let $\tm\coloneqq m^{\ell+1}$ and
    $\tM\coloneqq\{0,\ldots,\tm-1\}$.  Let
    $\langle\cdot\rangle:M^{\ell+1}\to\tM$ be the bijection defined by
    $\big\langle(a_0,\ldots,a_{\ell})\big\rangle=\sum_{i=0}^\ell
    a_im^i$. Note that $\langle\cdot\rangle$ maps relations
    $R\subseteq M^{\ell+1}$ to subsets $\langle R\rangle\subseteq\tM$,
    which we may also view as vectors in $\{0,1\}^\tm$: for $i\in\tm$
    we let $\langle R\rangle_i=1$ if $i\in\langle R\rangle$ and
    $\langle R\rangle_i=0$ otherwise.

    Thus every vertex-free expression $\xi$ defines a vector
    $\vec x_\xi\coloneqq\langle S_\xi\rangle\in \{0,1\}^\tm$. Every
    vertex expression defines an $\tm$-ary Boolean signal $\Cx_\xi$
    defined by
    \[
      \Cx_\xi(v)\coloneqq \angles{S_\xi(v)}.
    \]
    Every edge expression $\xi$ defines
    an ``edge signal'' $\Cy_{\xi }$ defined by
    \[
      \Cy_\xi(v,v')\coloneqq \angles{S_\xi(v,v')}.
    \]
    Note that, formally, $\Cy_\xi(v,v')$ is defined for all pairs
    $v,v'$ and not only for the pairs of endpoints of edges.

    Assume that the vertex expressions are
    $\xi^{(1)},\ldots,\xi^{(d)}$, ordered in such a way that if
    $\xi^{(i)}$ is a subexpression of $\xi^{(j)}$ then
    $i<j$. Then $\xi^{(d)}=\phi$.  Our GNN $\FN$ will have $d+1$ layers
    $\FL^{(1)},\ldots,\FL^{(d+1)}$. For $t\in[d]$, the layer $\FL^{(t)}$ will have
    input dimension $p_{t-1}\coloneqq p+(t-1)\tm$ and output
    dimension $p_t\coloneqq p+t\tm$. Layer
    $\FL^{(d+1)}$ will have input dimension $p_d$ and output dimension
    $p_{d+1}\coloneqq1$. Note that $p_0=p$. So our GNN
    $\FN=(\FL^{(1)},\ldots,\FL^{(d+1)})$ will have input dimension $p$
    and output dimension $1$, which is exactly what we need.

    Let $\vec x^{(0)}\coloneqq\Cu$ be the input
    signal, and for every $t\in[d+1]$, let
    $\Cx^{(t)}\coloneqq\tilde\FL^{(t)}(G,\Cx^{(t-1)})$. We shall define
    the layers in such a way that for all $t\in [d]$, all $v\in
    V(G)$ we have
    \begin{equation}
      \label{eq:81}
      \Cx^{(t)}(v)=\Cu(v)\Cx_{\xi^{(1)}}(v)\ldots \Cx_{\xi^{(t)}}(v).
    \end{equation}
    Furthermore,
    \begin{equation}
      \label{eq:82}
      \Cx^{(d+1)}(v)=
                        \begin{cases}
                          1&\text{if }G\models\phi(v),\\
                          0&\text{otherwise}.
                        \end{cases}
    \end{equation}
    So the GNN will take care of the vertex expressions. We also need
    to take care of the vertex-free expressions and the edge
    expressions. For vertex-free expressions, this is easy. Observe
    that a vertex-free expression contains no vertex variables at all,
    free or bound, because once we introduce a vertex variable the
    only way to bind it is by a neighbourhood term, and such a term
    always leaves one vertex variable free. Thus the value of a
    vertex-free expression does not depend on the input graph $G$, but
    only on its order $n$, the built-in numerical relations, and the
    integer arithmetic that is part of the logic. This means that for
    a vertex-free expression $\xi$ the relation $S_\xi$ is ``constant'' and can be
    treated like a built-in numerical relation that can be hardwired
    into the GNN. Dealing with edge expressions is more
    difficult. We will handle them when dealing with the neighbourhood
    terms.

    Let us turn to the vertex expressions. Let $t\in[d]$, and let $\xi\coloneqq\xi^{(t)}$. We distinguish
    between several cases depending on the shape of $\xi$.
    \begin{cs}
      \case1 $\xi=U_i(x)$ for some $i\in[p]$.
      
      Observe that for $v\in V(G)$ we have
      \[
        S_{\xi}(v)=\{\Cu(v)_i\}\times M^{\ell}.
      \]
      Thus $\Cx_\xi(v)_k=1$ if $k=\Cu(v)_i+\sum_{j=1}^\ell a_j m^j$
      for some $(a_1,\ldots,a_\ell)\in M^\ell$ and $\Cx_\xi(v)_k=0$
      otherwise.

      Using Lemma~\ref{lem:decode} in the special case $m=1$, we can
      design an FNN $\FF_1$ of input dimension $1$ and output
      dimension $m$ that maps $\Cu(v)_i$ to
      $(0,\ldots,0,1,0,\ldots,0)\in\{0,1\}^m$ with a $1$ at index
      $\Cu(v)_i$. We put another layer of $\tm$ output nodes on top of
      this and connect the node $k=\sum_{j=0}^\ell a_jm^j$ on this
      layer to the output node of $\FF_1$ with index $a_0$ by an edge
      of weight $1$. All nodes have bias $0$ and use $\lsig$
      activations. The resulting FNN $\FF_2$ has input dimension $1$
      and output dimension $\tm$, and it maps an input $x$ to the
      vector $\angles{\{x\}\times M^\ell}\in\{0,1\}^{\tm}$. Thus in
      particular, it maps $\Cu(v)_i$ to $\angles{S_{\xi}(v)}$. Padding
      this FNN with additional input and output nodes, we obtain an
      FNN $\FF_3$ of input dimension $p_{t-1}+1$ and output dimension
      $p_t=p_{t-1}+\tm$ that map $(x_1,\ldots,x_{p_{t-1}+1})$ to
      $(x_1,\ldots,x_{p_{t-1}},\FF_2(x_i))$.
      
      We use $\FF_3$ to define the combination functions
      $\comb^{(t)}: \Real^{p_{t-1}+1}\to\Real^{p_t}$ of the
      layer $\FL^{(t)}$. We define the message function by
      $\msg^{(t)}:\Real^{2p_{t-1}}\to \Real$ by $\msg^{(t)}(\vec
      x)\coloneqq 0$ for all $\vec x$, and we use sum aggregation.
      Then clearly, $\FL^{(t)}$ computes the transformation
      $(G,\Cx^{(t-1)})\mapsto (G,\Cx^{(t)})$ satisfying \eqref{eq:81}.

      \case2 $\xi=\xi'*\xi''$, where $\xi',\xi''$ are vertex
      expressions and either $*\in\{+,\cdot,\le\}$ and $\xi,\xi',\xi'$ are
      terms or $*=\wedge$ and
      $\xi,\xi',\xi''$ are formulas.
      
      We could easily (though tediously) handle this case by
      explicitly constructing the appropriate FNNs, as we did in
      Case~1. However, we will give a general argument that will help
      us through the following cases as well.

      As the encoding $R\subseteq M^{\ell+1}\mapsto\langle R\rangle
      \subseteq\tM$ and the decoding $\langle R\rangle
      \subseteq\tM\mapsto R\subseteq M^{\ell+1}$ are definable by
      arithmetical $\FOC$-formulas, the transformation
      \[
        \big(\langle S_{\xi'}(v)\rangle, \langle S_{\xi''}(v)\rangle\big)\mapsto
        \langle S_{\xi'*\xi''}(v)\rangle,
      \]
      which can be decomposed as
      \[
        \big(\langle S_{\xi'}(v)\rangle, \langle S_{\xi''}(v)\rangle\big)
        \mapsto
        \big(S_{\xi'}(v),S_{\xi'}(v)\big)\mapsto S_{\xi*\xi''}(v)\mapsto
        \langle S_{\xi'*\xi''}(v)\rangle,
      \]
      is also definable by an arithmetical $\FOC$-formula, using
      Lemmas~\ref{lem:ar1} and \ref{lem:ar2a} for the main step
      $\big(S_{\xi'}(v),S_{\xi'}(v)\big)\mapsto S_{\xi*\xi''}(v)$.
      Hence by Corollary~\ref{cor:nuTC}, it is computable by a
      threshold circuit $\FC^*$ of bounded depth (only depending on
      $*$) and polynomial size. Hence by Lemma~\ref{lem:c2fnn} it is
      computable by an FNN $\FF^*$ of bounded depth and polynomial
      size. Note that for every vertex $v$, this FNN $\FF^*$ maps
      $\big(\Cx_{\xi'}(v),\Cx_{\xi'}(v)\big)$ to $\Cx_\xi(v)$.

      We have $\xi'=\xi^{(t')}$ and $\xi''=\xi^{(t'')}$ for some
      $t',t''<t$. 
      Based on $\FF^*$ we construct an
      FNN $\FF$ of input dimension $p_{t-1}+1$ and output dimension $p_t=p_{t-1}+\tm$
      such that for $\vec u\in\{0,\ldots,n-1\}^p$ and $\vec x_1,\ldots,\vec
      x_{t-1}\in \Real^{\tm}$, $x\in\Real$
      \[
        \FF(\vec u\vec x_1\ldots\vec x_{t-1}x)=
        \vec u\vec x_1\ldots\vec x_{t-1}\FF^*(\vec x_{t'},\vec x_{t''}).
      \]
      Continuing as in Case~1, we use $\FF$ to define the combination functions
      $\comb^{(t)}$ of the
      layer $\FL^{(t)}$, and again we use a trivial message function.

      \case3 $\xi=\neg\xi'$, where $\xi'$ is vertex
      formula.
      
      This case can be handled as Case~2.

      \case4 $\xi=\xi'*\xi''$, where $\xi'$ is a vertex expression,
      $\xi''$ is a vertex-free expression, and either
      $*\in\{+,\cdot,=,\le\}$ and $\xi,\xi',\xi''$ are terms or $*\in\wedge$ and
      $\xi,\xi',\xi''$ are formulas.
      
      As in Case~2, we construct a threshold circuit  $\FC^*$ of bounded depth and polynomial
      size that computes the mapping $\big(\langle S_{\xi'}(v)\rangle, \langle S_{\xi''}\rangle\big)\mapsto
        \langle S_{\xi'*\xi''}(v)\rangle$. As we argued above, the
        relation $S_\xi''$ for the vertex-free expression $\xi''$ and
        hence the vector $\angles{S_{\xi''}}$ do not depend on the
        input graph. Hence we can simply hardwire the
        $\angles{S_{\xi''}}$ into $\FC^*$, which gives us a circuit
        that computes the mapping $\langle S_{\xi'}(v)\rangle\mapsto
        \langle S_{\xi'*\xi''}(v)\rangle$. From this circuit we obtain an
        FNN $\FF^*$ that computes the same mapping, and we can
        continue as in Case~2.

      \case5 $\xi=\#(y^1<\theta,\ldots,y^k<\theta).
      \psi$, where $\psi$ is a vertex formula.

We argue as in Cases~2-4.
      We construct a threshold circuit that computes the mapping
      \[
        \langle S_{\psi}(v)\rangle\mapsto
        \langle S_{\xi}(v)\rangle.
      \]
      Turning this circuit into an FNN $\FF^*$ that computes the same
      mapping, we can continue as in Case~2.

      \case 6
      $\xi=\#(x',y^1<\theta,\ldots,y^k<\theta).
      \big(E(x,x')\wedge\psi\big)$, where $\psi$ is an edge
      formula or a vertex formula.

      Without loss of generality, we assume that $\psi$ is an edge
      formula. If it is not, instead of $\psi$ we take the conjunction
      of $\psi$ with $U_1(x)+U_1(x')\ge 0$, which is always true,
      instead. Moreover, we may assume that $\psi$
      contains no equality atoms $x=x'$, because the guard $E(x,x')$
      forces $x$ and $x'$ to be distinct.  Thus $\psi$ is constructed
      from vertex expressions and vertex-free expressions using
      $+,\cdot,\le$ to combine terms, $\neg,\wedge$ to combine formulas,
      and counting terms
      $\#((y')^1<\theta,\ldots,(y')^{k'}<\theta).\psi'$. Let
      $\xi^{(t_1)},\ldots,\xi^{(t_s)}$ be the maximal (with respect to the
      inclusion order on expressions) 
      vertex expressions
      occurring in $\psi$. Assume that $x$ is the free vertex variable
      of $\xi^{(t_1)},\ldots,\xi^{(t_r)}$ and $x'$ is the free vertex variable
      of $\xi^{(t_{r+1})},\ldots,\xi^{(t_s)}$. Arguing
      as in Case~2 and Case~5, we can construct a threshold circuit
      $\FC$ of bounded depth and polynomial size that computes the
      mapping
     \[
        \big(\langle S_{\xi^{(t_1)}}(v)\rangle, \ldots, \langle S_{\xi^{(t_r)}}(v)\rangle, \langle S_{\xi^{(t_{r+1})}}(v')\rangle, \ldots, \langle S_{\xi_{i_{s}}}(v')\rangle\big)\mapsto
        \langle S_{\psi}(v,v')\rangle
      \]
      To simplify the notation, let us assume that
      $y^i=y_{i}$. Thus the free variables of $\xi$ are among
      $x,y_{k+1},\ldots,y_{\ell}$, and we may write
      $\xi(x,y_{k+1},\ldots,y_{\ell})$. Let
      \[
        \zeta(x,x',y_{k+1},\ldots,y_\ell)\coloneqq\#(y_1<\theta,\ldots,y_k<\theta).\psi
      \]
      and observe that for all $v\in V(G)$ and
      $a_{k+1},\ldots,a_\ell\in M$ we have
      \[
        \sem{\xi}^G(v, a_{k+1},\ldots,a_\ell)=\sum_{v'\in
          N(v)}\sem{\zeta}^G(v,v',a_{k+1},\ldots,a_\ell).
      \]
      For $v,v'\in V(G)$, let
      \[
        R_\zeta(v,v')\coloneqq\{(a_{0},\ldots,a_{\ell-k-1},b)\bigmid
        b<\sem{\zeta}^G(v,v',a_{0},\ldots,a_{\ell-k-1})\}\subseteq
        M^{\ell-k+1},
      \]
      and slightly abusing notation, let
      \[
        \angles{R_\zeta(v,v)}=\left\{\left.\sum_{i=0}^{\ell-k}a_im^i\;\right|\;(a_0,\ldots,a_{\ell-k})\in
        R_\zeta(v)\right\}\subseteq\big\{0,\ldots,m^{\ell-k+1}-1\big\}
      \]
      which we may also view as a vector in $\{0,1\}^{m^{\ell-k+1}}$. 
      Again arguing via $\FOC$, we can construct a threshold circuit
      $\FC'$ that computes the transformation
      \[
        \big(\langle S_{\xi^{(t_1)}}(v)\rangle, \ldots, \langle S_{\xi^{(t_r)}}(v)\rangle, \langle S_{\xi^{(t_{r+1})}}(v')\rangle, \ldots, \langle S_{\xi_{i_{s}}}(v')\rangle\big)\mapsto
        \langle R_\zeta(v,v')\rangle.
      \]
      Using Lemma~\ref{lem:c2fnn}, we can turn $\FC'$ into an FNN
      $\FF'$ that computes the same Boolean function. Let
      $\vec
      c(v,v')=(c_0,\ldots,c_{m^{\ell-k}-1})\in M^{m^{\ell-k}}$ be
      the vector defined by as follows: for
      $(a_0,\ldots,a_{\ell-k-1})\in M^{\ell-k}$ and 
      $j=\sum_{i=0}^{\ell-k-1}a_im^i$ we let
      \[
        c_j\coloneqq\sem{\zeta}^G(v,v',a_{0},\ldots,a_{\ell-k-1}).
      \]
      Then
      \begin{align*}
        c_j&=\big|\big\{ b\bigmid
             b<\sem{\zeta}^G(v,v',a_{0},\ldots,a_{\ell-k-1})\big\}\big|\\
        &=\big|\big\{ b\bigmid (a_0,\ldots,a_{\ell-k-1},b)\in
          R_\zeta(v,v')\big\}\big|\\
        &=\sum_{b\in M}\langle R_\zeta(v,v')\rangle_{j+bm^{\ell-k}}.
      \end{align*}
      Thus an FNN of depth $1$ with input
      dimension $m^{\ell-k+1}$ and output dimension $m^{\ell-k}$ can
      transform $\langle R_\theta(v,v')\rangle$ into $\frac{1}{m}\vec
      c(v,v')$. We take the factor $\frac{1}{m}$ because $0\le c_j<
      m$ and thus $\frac{1}{m}\vec
      c(v,v')\in[0,1]^{m^{\ell-k}}$, and we can safely use
      $\lsig$-activation. We add this FNN on top of the $\FF'$ and
      obtain an FNN $\FF''$ that computes the transformation
      \[
        \big(\langle S_{\xi^{(t_1)}}(v)\rangle, \ldots, \langle S_{\xi^{(t_r)}}(v)\rangle, \langle S_{\xi^{(t_{r+1})}}(v')\rangle, \ldots, \langle S_{\xi_{i_{s}}}(v')\rangle\big)\mapsto \frac{1}{m}\vec
        c(v,v').
      \]
      The message function $\msg^{(t)}$ of the GNN layer $\FL^{(t)}$ computes the
      function
      \[
        \big(\Cx^{(t-1)}(v),\Cx^{(t-1)}(v')\big)\to \frac{1}{m}\vec
        c(v,v')
      \]
      which we can implement by an FNN based on $\FF''$. 
      Aggregating, we obtain the signal $\Cz$ such that
      \[
        \Cz(v)=\sum_{v'\in
          N(v)}\msg^{(t)}(v,v')=\frac{1}{m}\sum_{v'\in N(v)}\vec
        c(v,v')
      \]
      Thus  for
      $(a_0,\ldots,a_{\ell-k-1})\in M^{\ell-k}$ and 
      $j=\sum_{i=0}^{\ell-k-1}a_im^i$ we have
      \begin{align*}
        \Cz(v)_j&=\frac{1}{m}\sum_{v'\in N(v)}\vec
                  c(v,v')\\
        &=\frac{1}{m}\sum_{v'\in
          N(v)}\sem{\zeta}^G(v,v',a_{0},\ldots,a_{\ell-k-1})\\
        &=\frac{1}{m}\sem{\xi}^G(v,a_{0},\ldots,a_{\ell-k-1}).
      \end{align*}
      Our final task will be to transform the vector
      $\Cz(v)\in[0,1]^{m^{\ell-k}}$ to the vector $\angles{S_\xi}\in\{0,1\}^{\tm}$. In a first step,
      we transform $\Cz(v)$ into a vector
      $\Cz'(v)\in M^{m^{\ell}}$ with entries
      \[
        \Cz'(v)_j=\sem{\xi}^G(v,a_{k+1},\ldots,a_{\ell})
      \]
      for $(a_1,\ldots,a_\ell)\in M^{\ell}$ and
      $j=\sum_{i=0}^{\ell-1}a_{i+1}m^i$. We need to transform $\Cz'(v)$ into
      $\angles{S_\xi(v)}\in\{0,1\}^{\tm}$, which for every
      $(a_1,\ldots,a_\ell)\in M^{\ell}$ has a single $1$-entry in
      position
      $j=\sum_{i=0}^{\ell}a_{i}m^i$ for
      $a_0=\sem{\xi}^G(v,a_{k+1},\ldots,a_{\ell})$ and $0$-entries in
      all positions $j'=\sum_{i=0}^{\ell}a_{i}m^i$ for
      $a_0\neq\sem{\xi}^G(v,a_{k+1},\ldots,a_{\ell})$. We can use Lemma~\ref{lem:decode} for this transformation. 

      Thus we can construct an FNN $\FF''$ that
      transforms the output $\vec z(v)$ of the aggregation into
      $\angles{S_\xi(v)}$. We define the combination function
      $\comb^{(t)}:\Real^{p_{t-1}+m^{\ell-k}}\to\Real^{p_t}$ by
      $
        \comb^{(t)}(\vec x,\vec z):=(\vec x,\FF''(\vec z)).
      $
      Then
      \[
        \comb^{(t)}(\Cx(v),\Cz(v)):=(\vec x,\Cx_\xi(v)).
      \]
      Thus the layer $\FL^{(t)}$ with message function $\msg^{(t)}$
      and combination function $\comb^{(t)}$ satisfies \eqref{eq:81}.
    \end{cs}
    All that remains is to define the last layer $\FL^{(d+1)}$ satisfying
    \eqref{eq:82}. Since $\xi^{(d)}=\phi$, by \eqref{eq:81} with $t=d$,
    the vector $\Cx_\phi(v)=\angles{S_\phi(v)}$ is the projection of
    $\Cx^{(d)}(v)$ on the last $\tm$ entries. As $\phi$ has no free
    number variables, we have
    \[
      \angles{S_\phi(v)}=
      \begin{cases}
        \vec 1&\text{if }G\models\phi(v),\\
        \vec 0&\text{if }G\not\models\phi(v).
      \end{cases}
    \]
    In particular, the last entry of $\angles{S_\phi(v)}$ and hence of
    $\Cx^{(d)}(v)$ is $1$ if $G\models\phi(v)$ and $0$ otherwise. Thus
    all we need to do on the last layer is project the output on the
    last entry.

    This completes the construction.
\end{proof}

The following theorem directly implies Theorem~\ref{theo:main1} stated
in the introduction.

\begin{theorem}\label{theo:converse}
  Let ${\CQ}$ be a unary query on $\CGS^\bool_p$. Then the following are equivalent:
  \begin{enumerate}
  \item ${\CQ}$ is definable in ${\GCnu}$.
  \item There is a polynomial-size bounded-depth family
    of rational piecewise-linear GNNs using only $\lsig$-activations and
    $\SUM$-aggregation that computes ${\CQ}$.
  \item There is a rpl-approximable polynomial weight bounded-depth
    family of GNNs that computes ${\CQ}$.
  \end{enumerate}
\end{theorem}

\begin{proof}
  The implication (1)$\implies$(2) follows from
  Lemma~\ref{lem:converse} in the special case that the $U_i$ are
  Boolean, that is, only take values in $\{0,1\}$. We can then replace
  them by the unary relations $P_i$ that we usually use to represent
  Boolean signals.

  The implication (2)$\implies$(3) is trivial.

  The implication (3)$\implies$(1) is Corollary~\ref{cor:snonuniform}.
\end{proof}

Again, we have a version of our theorem for GNNs with global readout. The proofs can easily be adapted. 

\begin{theorem}\label{theo:converse-gc}
  Let ${\CQ}$ be a unary query on $\CGS^\bool_p$. Then the following are equivalent:
  \begin{enumerate}
  \item ${\CQ}$ is definable in ${\GCgcnu}$.
  \item There is a polynomial-size bounded-depth family
    of rational piecewise-linear GNNs with global readount using only $\lsig$-activations and
    $\SUM$-aggregation that computes ${\CQ}$.
  \item There is a rpl-approximable polynomial weight bounded-depth
    family of GNNs with global readout that computes ${\CQ}$.
  \end{enumerate}
\end{theorem}

We also have a version 
  Theorem~\ref{theo:converse} for 1-GNNs and the modal fragment. For
  once, let us state the theorem explicitly.

\begin{theorem}\label{theo:converse-modal}
  Let ${\CQ}$ be a unary query on $\CGS^\bool_p$. Then the following are equivalent:
  \begin{enumerate}
  \item ${\CQ}$ is definable in ${\MCnu}$.
  \item There is a polynomial-size bounded-depth family
    of rational piecewise-linear 1-GNNs using only $\lsig$-activations and
    $\SUM$-aggregation that computes ${\CQ}$.
  \item There is a rpl-approximable polynomial weight bounded-depth
    family of 1-GNNs that computes ${\CQ}$.
  \end{enumerate}
\end{theorem}

\begin{proof}
  The implication (2)$\implies$(3) is trivial, and the implication
  (3)$\implies$(1) is the modal version of
  Corollary~\ref{cor:snonuniform} (see Remark~\ref{rem:mc5}).

  The implication (1)$\implies$(2) follows from a modal version of
  Lemma~\ref{lem:converse} where $\phi(x)$ is a ${\MCnu}$-formula and
  we obtain a family $\CN$ of 1-GNNs. We need to adapt the proof of
  the lemma. Actually, the proof becomes simpler in this case. Let us
  first understand the structure of the $\MC$-formula
  $\phi(x)$. Recall the definition of $\MC$: in $\MC$-formulas, the counting terms of the
  form \eqref{eq:4} are are only permitted if $x_i$ does not occur
  freely in $\phi$ or any of the $\theta_j$. Since counting terms of
  the form \eqref{eq:5} do not affect the vertex variables, it follows
  that an $\MC$-formula with only one free vertex variable cannot
  contain a subformula or subterm with two free vertex variables,
  unless it is an atomic subformula that occurs as a guard in a term
  of the form \eqref{eq:4}.

  With this initial observation, let us turn to the proof of modal
  version of Lemma~\ref{lem:converse}. We follow the proof of the
  original lemma. In a counting term of the form \eqref{eq:80}, the
  subformula $\psi$ can only have $x'$ as a free vertex
  variable. Moreover, we do not have any edge expressions (in the
  terminology of the original proof) except for the guards $E(x,x')$
  in terms \eqref{eq:80}. We proceed inductively exactly as in the
  proof of the original lemma. The only adaptations necessary
  are in Case~6, where we consider a counting term
  \[
    \xi=\#(x',y^1<\theta,\ldots,y^k<\theta).
    \big(E(x,x')\wedge\psi\big)
  \]
  This case becomes easier now, because we can
  assume that $\psi$ is a vertex formula with free variable $x'$. By
  induction, we have already computed the $\tilde m$-ary signal
  $\Cx_\psi$ with $\Cx_\psi(v)=\angles{S_\psi(v)}$ as part of the
  output of the previous layer of the 1-GNN we construct.

  As in the original proof, we assume $y^i=y_{i}$ and hence that the
  free variables of $\xi=\xi(x,y_{k+1},\ldots,y_{\ell})$ are among
  $x,y_{k+1},\ldots,y_{\ell}$. We let
  \[
    \zeta(x',y_{k+1},\ldots,y_\ell)\coloneqq\#(y_1<\theta,\ldots,y_k<\theta).\psi.
  \]
  Then for all $v\in V(G)$ and
  $a_{k+1},\ldots,a_\ell\in M$ we have
  \[
    \sem{\xi}^G(v, a_{k+1},\ldots,a_\ell)=\sum_{v'\in
      N(v)}\sem{\zeta}^G(v',a_{k+1},\ldots,a_\ell).
  \]
  Now we let
  \[
    R_\zeta(v')\coloneqq\{(a_{0},\ldots,a_{\ell-k-1},b)\bigmid
    b<\sem{\zeta}^G(v',a_{0},\ldots,a_{\ell-k-1})\}\subseteq
    M^{\ell-k+1},
  \]
  define $\angles{R_\zeta(v')}$ accordingly. We construct a threshold circuit
  and from this an FNN $\FF'$ that computes the transformation
  \[
    \angles{S_\psi(v')}\mapsto\angles{R_\zeta(v')}.
  \]
  We let $\vec c(v')=(c_0,\ldots,c_{m^{\ell-k}-1})\in M^{m^{\ell-k}}$
  be the vector defined by as follows: for
  $(a_0,\ldots,a_{\ell-k-1})\in M^{\ell-k}$ and
  $j=\sum_{i=0}^{\ell-k-1}a_im^i$ we let
  \[
    c_j\coloneqq\sem{\zeta}^G(v',a_{0},\ldots,a_{\ell-k-1}) =\sum_{b\in M}\langle R_\zeta(v')\rangle_{j+bm^{\ell-k}}.
  \]
  Since we can compute the sum by a single FNN-layer, we can modify
  $\FF'$ to obtain an FNN $\FF''$ that computes the tranformation
   \[
    \angles{S_\psi(v')}\mapsto\frac{1}{m}\vec c(v').
  \]
 Then the message function of the 1-GNN-layer $\FL^{(t)}$ we are constructing
 computes the transformation
 \[
   \Cx^{(t-1)}(v')\mapsto\frac{1}{m}\vec c(v')\eqqcolon\msg^{(t)}(v').
 \]
 We aggregate and obtain the signal $\Cz(v)$ defined by
 \[
   \Cz(v) =\sum_{v'\in
     N(v)}\msg^{(t)}(v')=\frac{1}{m}\sum_{v'\in N(v)}\vec
   c(v')
 \]
 with
 \[
   \Cz(v)_j =\frac{1}{m}\sum_{v'\in
     N(v)}\sem{\zeta}^G(v',a_{0},\ldots,a_{\ell-k-1})=\frac{1}{m}\sem{\xi}^G(v,a_{0},\ldots,a_{\ell-k-1}).
 \]
 From
 now on, the proof continues exactly as the proof of the original lemma.
  \end{proof}

\begin{myremark}\label{rem:converse}
  Let us finally address a question which we we already
  raised at the end of Section~\ref{sec:uniform}. Is every unary query
  definable in $\GC$ computable by single rational piecewise linear or at
  least by an rpl approximable GNN? In other words: do we really need
  families of GNNs in Theorem~\ref{theo:converse}, or could we just
  use a single GNN?

  It has been observed in \cite{RosenbluthTG23} that the answer to
  this question is ``no''. Intuitively, the reason is that GNNs cannot
  express ``alternating'' queries like nodes having an even degree. To
  prove this, we analyse the behaviour of GNNs on stars $S_n$ with $n$
  leaves, for increasing $n$. The signal at the root node that the GNN
  computes is approximately piecewise polynomial as a function of
  $n$. However, a function that is $1$ for all even natural numbers
  $n$ and $0$ for all odd numbers is very far from polynomial.
  \uend
\end{myremark}

\section{Random Initialisation}
\label{sec:ri}
A GNN with random initialisation receives finitely many random
features together with its $p$-dimensional input signal.  We assume
that the random features are chosen independently uniformly from the
interval $[0,1]$. We could consider other distributions, like the
normal distribution $N(0,1)$, but in terms of expressiveness this
makes no difference, and the uniform distribution is easiest to
analyse. The random features at different vertices are chosen
independently. As in \cite{AbboudCGL21}, we always assume that GNNs
with random initialisation have global readout.\footnote{There is no
  deeper reason for this choice, it is just that the results get
  cleaner this way. This is the same reason as in \cite{AbboudCGL21}.}

We denote the uniform distribution on $[0,1]$ by $\uni$, and for a
graph $G$ we write $\Cr\sim\uni^{r\times V(G)}$ to denote that
$\Cr\in\CS_r(G)$ is obtained by picking the features $\Cr(v)_i$
independently from $\uni$. Moreover, for a signal $\Cx\in\CS_p(G)$, by
$\Cx\Cr$ we denote the $(p+r)$-dimensional signal with
$\Cx\Cr(v)=\Cx(v)\Cr(v)$. Formally, a \emph{$(p,q,r)$-dimensional GNN with ri} 
is a GNN $\FN$ with global readout of input dimension $p+r$ and output dimension $q$. It computes a random
variable mapping pairs $(G,\Cx)\in\CGS_p$ to the space
$\CS_q(G)$, which we view as a product measurable space
$\Real^{q\times V(G)}$ equipped with a Borel $\sigma$-algebra (or Lebesgue
$\sigma$-algebra, this does not matter here). Abusing (but hopefully also
simplifying) notation, we use $\FR$ to denote a GNN that we interpret
as a GNN with random initialisation, and we use $\tilde\FR$ to
denote the random variable. Sometimes it is also convenient
to write $\FR(G,\Cx)\coloneqq(G,\tilde\FR(x))$. It is not hard to
show that the mapping $\tilde\FR$ is measurable with respect
to the Borel $\sigma$-algebras on the product spaces
$\CS_p(G)=\Real^{(p+r)\times V(G)}$ and $\CS_q(G)=\Real^{q\times
  V(G)}$. Here we use that the activation functions of $\FR$ are continuous.
To define the probability
distribution of the random variable $\tilde\FR$, for all
$(G,\Cx)\in\CGS_p$ and all events (that is,
measurable sets) $\CY\subseteq  \CS_q(G)$ we let
\begin{equation}
  \label{eq:83}
    \Pr\big(\tilde\FR(G,\Cx)\in
  \CY\big)\coloneqq\Pr_{\Cr\sim\uni^{r\times V(G)}}\big(\tilde\FR(G,\Cx\Cr)\in \CY\big),
\end{equation}
where on the left-hand side we interpret $\FR$ as
a $(p,q,r)$-dimensional GNN with ri and on the right-hand side 
just as an ordinary GNN of input dimension $(p+r)$ and output
dimension $q$.

For a $(p,q,r)$-dimensional GNN with ri we call $p$ the \emph{input
  dimension}, $q$ the \emph{output dimension}, and $r$ the
\emph{randomness dimension}.

Let ${\CQ}$ be a unary query on $\CGS_p^{\bool}$. We say that
a GNN with ri $\FR$ of input dimension $p$ and output dimension $1$ \emph{computes} ${\CQ}$ if for all
$(G,\Cb)\in\CGS_p^\bool$ and all $v\in V(G)$ it holds that
\[
  \begin{cases}
    \Pr\big(\tilde\FR(G,\Cb)\ge\frac{3}{4}\big)\ge\frac{3}{4}&\text{if }{\CQ}(G,\Cb)=1,\\
    \Pr\big(\tilde\FR(G,\Cb)\le\frac{1}{4}\big)\ge\frac{3}{4}&\text{if }{\CQ}(G,\Cb)=0.
  \end{cases}
\]
It is straightforward to extend this definition to families
$\CR=(\FR^{(n)})_{n\in\Nat}$ of GNNs with ri.

By a fairly standard probability amplification result, we can make the
error probabilities exponentially small.

\begin{lemma}\label{lem:ri1}
  Let ${\CQ}$ be a unary query over $\CGS_p$ that is computable by family
  $\CR=(\FR^{(n)})_{n\in\Nat}$ of GNNs
  with ri, and let $\pi(X)$ be a polynomial. Then there is a family $\CR'=(\FR')^{(n)})_{n\in\Nat}$ of GNNs
  with ri such that the following holds.
  \begin{itemize}
    \item[(i)] $\CR'$ computes ${\CQ}$, and for every $n$, every 
      $(G,\Cb)\in\CGS_p^\bool$ of order $n$, and every $v\in V(G)$,
      \begin{equation}
        \label{eq:84}
         \begin{cases}
    \Pr\big(\tilde\CR'(G,\Cb)=1\big)\ge 1-2^{-\pi(n)}&\text{if }{\CQ}(G,\Cb)=1,\\
    \Pr\big(\tilde\CR'(G,\Cb)=0\big)\ge 1-2^{-\pi(n)}&\text{if }{\CQ}(G,\Cb)=0.
  \end{cases}
      \end{equation}
\item[(ii)] The weight of $(\FR')^{(n)}$ is polynomially bounded in the
  weight of $\FR^{(n)}$ and $n$.
\item[(iii)] The depth of $(\FR')^{(n)}$ is at most the depth of $\FR^{(n)}$
  plus $2$.
\item[(iv)] The randomness dimension of $(\FR')^{(n)}$ is polynomially
  bounded in $n$ and the randomness dimension
  of $\FR^{(n)}$.
  \end{itemize}
\end{lemma}

\begin{proof}
  We just run sufficiently many copies of $\CR$ in parallel
  (polynomially many suffice) and then
  take a majority vote in the end, which is responsible for one of the
  additional layers. This way we obtain a family $\CR''$ that achieves
  \[
    \begin{cases}
    \Pr\big(\tilde\CR''(G,\Cb)\ge \frac{3}{4}\big)\ge 1-2^{-\pi(n)}&\text{if }{\CQ}(G,\Cb)=1,\\
    \Pr\big(\tilde\CR'(G,\Cb)\le\frac{1}{4}\big)\ge 1-2^{-\pi(n)}&\text{if }{\CQ}(G,\Cb)=0.
  \end{cases}
\]
To get the desired $0,1$-outputs, we apply the transformation
$\lsig(2x-\frac{1}{2})$ to the output on an additional
layer.
\end{proof}

\begin{lemma}\label{lem:ri2}
  Let ${\CQ}$ be a unary query over $\CGS_p$ that is computable by an
  rpl-approximable polynomial-weight, bounded-depth family $\CR$ of
  GNNs with ri. Then there is an order-invariant
  $\GCgcnu$-formula that defines ${\CQ}$.
\end{lemma}

\begin{proof}
  Suppose that $\CR=(\FR^{(n)})_{n\in\Nat}$, and for every $n$, let
  $r^{(n)}$ be the randomness dimension of $\FR^{(n)}$. Viewed as a
  standard GNN, $\FR^{(n)}$ has input dimension $p+r^{(n)}$ and output
  dimension $1$. By the previous lemma, we may assume that our family
  satisfies \eqref{eq:84} for $\pi(X)=3p X^3$.

  Our first step is to observe that we can safely truncate the random
  numbers, which we assume to be randomly chosen from $[0,1]$, to
  $O(n)$ bits. This follows from Corollary~\ref{cor:gnnbound2a}: truncating the numbers to $cn$ bits
  means that we replace the random signal $\Cr$ by a $\Cr'$ such
  that $\inorm{\Cr-\Cr'}\le 2^{-cn}$, and the corollary implies that if we
  choose $c$ sufficiently large, we will approximate the original GNN
  up to an additive error of $\frac{1}{10}$. Thus for some constant $c$ we may assume
  that the random strings are not drawn uniformly from $[0,1]$, but
  from the set
  \[
    U_n\coloneqq\left\{\left.\sum_{i=0}^{cn-1}a_i2^{-i-1}\;\right|\;a_0,\ldots,a_{cn-1}\in\{0,1\}\right\}.
  \]
  Let us denote the uniform distribution on this set by $\CU_n$
  Hence for  for every $n$, every 
  $(G,\Cb)\in\CGS_p^\bool$ of order $n$, and every $v\in V(G)$,
  \begin{equation}
    \label{eq:85}
    \begin{cases}
      \Pr_{\Cr\sim \CU_n^{r^{(n)}\times V(G)}}\big(\tilde\CR(G,\Cb\Cr)\ge\frac{9}{10}\big)\ge 1-2^{-\pi(n)}&\text{if }{\CQ}(G,\Cb)=1,\\
      \Pr_{\Cr\sim \CU_n^{r^{(n)}\times V(G)}}\big(\tilde\CR(G,\Cb\Cr)\le\frac{1}{10}\big)\ge 1-2^{-\pi(n)}&\text{if }{\CQ}(G,\Cb)=0.
    \end{cases}
  \end{equation}
  Next, we want to apply Theorem~\ref{theo:snonuniform} in the version for GNNs
  with global readout and the logic $\GCgcnu$. Let $\vec X_r$ be an
  r-schema of type $\rtp{\ttv\ttn}$.
  Suppose that $(G,\Cb)\in\CGS_p^\bool$,
  and let $\Ca$ be an assignment over $G$. We view $(G,\Cb)$ as a $p$-labelled graph
  here, that is, as an $\{E,P_1,\ldots,P_p\}$-structure. So
  $((G,\Cb),\Ca)$ is the pair consisting of this structure together
  with the assignment $\Ca$. 
  Let
  \begin{equation}
    \label{eq:86}
  \Cr(v)\coloneqq\Big(\num{\vec X_r}^{((G,\Cb),\Ca)}(v,0),\ldots, 
    \num{\vec X_r}^{((G,\Cb),\Ca)}(v,r^{(n)}-1)\Big).
  \end{equation}
  In the following we will write $(G,\Cb,\Cr)$ instead of
  $((G,\Cb),\Ca)$ if $\Cr$ is obtained from some assignment $\Ca$ via
  \eqref{eq:86}, always assuming that $\num{\vec
    X_r}^{((G,\Cb),\Ca)}(v,k)=0$ for $k\ge r^{(n)}$.
  Then if $\phi(x)$ is a formula whose only free variables are among
  $x$ and the relation and function variables appearing in
  $\vec X_r$, the value $\sem{\phi}^{((G,\Cb),\Ca)}$ only depends on
  $(G,\Cb)$, $\Cr$
  and $v\coloneqq\Ca(x)$, and we may ignore the rest of $\Ca$. In particular, we may write
  $(G,\Cb,\Cr)\models\phi(v)$ instead of
  $\sem{\phi}^{(G,\Cb,\Ca)}=1$. 

  With is notation at hand, let us apply
  Theorem~\ref{theo:snonuniform-gc}. Replacing $W$ by the constant $1$ and $W'$ by $10$, we obtain an
  r-expression $\rexp$ in $\GCgc$ such that for all
  $(G,\Cb)\in\CGS_p$ of order $n$ and $\Cr\in\CS_{r^{(n)}}(G)$ defined
    as in \eqref{eq:86}, if $\inorm{\Cr}\le 1$, then for all $v\in V(G)$ we have
\[
  \Big|\tilde\CR(G,\Cb,\Cr)(v)-\num{\rexp}^{(G,\Cb,\Cr)}(v)\Big|\le\frac{1}{10}.
\]
Note that if $\Cr(v)_i\in U_n$ for
  $0\le i<r^{(n)}$, then we have $\inorm{\Cr}\le 1$.
Hence, 
 \begin{equation*}
    \begin{cases}
      \Pr_{\Cr\sim \CU_n^{r^{(n)}\times V(G)}}\big(\num{\rexp}^{(G,\Cb,\Cr)}(v)\ge\frac{8}{10}\big)\ge 1-2^{-\pi(n)}&\text{if }{\CQ}(G,\Cb)=1,\\
      \Pr_{\Cr\sim \CU_n^{r^{(n)}\times V(G)}}\big(\num{\rexp}^{(G,\Cb,\Cr)}(v)\le\frac{2}{10}\big)\ge 1-2^{-\pi(n)}&\text{if }{\CQ}(G,\Cb)=0.
    \end{cases}
  \end{equation*}
  From $\rexp$ we obtain a formula $\phi(x)$, just saying
  $\rexp\ge\frac{1}{2}$, such that
  \begin{equation}
    \label{eq:87}
   \Pr_{\Cr\sim
     \CU_n^{r^{(n)}\times V(G)}}\big(\sem{\phi}^{(G,\Cb,\Cr)}(v)={\CQ}(G,\Cb)(v)\big)\ge 1-2^{-\pi(n)}.
 \end{equation}
 Our next step will be to simplify the representation of the random
 features in the signal $\Cr\in U_n^{r^{(n)}\times V(G)}$. Each number in $U_n$ can
 be described by a subset of $\{0,\ldots,cn-1\}$: the set $S$
 represents the number $\sum_{s\in S}2^{-s-1}$. Thus we can represent
 $\Cr\in U_n^{r^{(n)}\times V(G)}$ by a relation $R\subseteq
 V(G)\times\{0,\ldots,r^{(n)}-1\}\times\{0,\ldots,cn-1\}$, and we can
 transform $\phi(x)$ into a formula $\phi'(x)$ that uses a relation
 variable $X_r$ of type $\tta\ttv\ttn\ttn\ttz$ instead of the r-schema
 $\vec X_r$ such that  
 \begin{equation}
    \label{eq:88}
   \Pr_{R}\Big(\sem{\phi'}^{(G,\Cb,R)}(v)={\CQ}(G,\Cb)(v)\Big)\ge 1-2^{-\pi(n)},
 \end{equation}
 where the probability is over all $R$ chosen uniformly at random from
 $V(G)\times\{0,\ldots,r^{(n)}-1\}\times\{0,\ldots,cn-1\}$ and
 $(G,\Cb,R)$ is $((G,\Cb),\Ca)$ for some assignment $\Ca$ with
 $\Ca(X_r)=R$. 

 The next step is to introduce a linear order and move towards an
 order invariant formula. We replace the relation variable $X_r$ of
 type $\tta\ttv\ttn\ttn\ttz$ by a relation variable $Y_r$ of type
 $\tta\ttn\ttn\ttn\ttz$, and we introduce a linear order $\les$ on
 $V(G)$. Then we replace atomic subformulas $X_r(x',y,y')$ of
 $\phi'(x)$ by
 \[
   \exists y''\le\ord\big(\# x.x\les x' = \#
   y'''\le\ord. y'''\le y''\wedge Y_r(y'',y,y')\big)
 \]
 The equation $\# x.x\les x' = \#
   y'''\le\ord. y'''\le y''$ just says that $x'$ has the same position
   in the linear order $\les$ on $V(G)$ as $y''$ has in the natural
   linear order $\le$. So basically, we store the random features for
   the $i$th vertex in the linear order $\les$ in the $i$-entry of
   $Y_r$. We obtain a new formula $\phi''(x)$ satisfying
 \begin{equation}
    \label{eq:89}
   \Pr_{R}\Big(\sem{\phi''}^{(G,\Cb,\les,R)}(v)={\CQ}(G,\Cb)(v)\Big)\ge 1-2^{-\pi(n)},
 \end{equation}
 where the probability is over all $R$ chosen uniformly at random from
 $\{0,\ldots,n-1\}\times\{0,\ldots,r^{(n)}-1\}\times\{0,\ldots,cn-1\}$.
 Importantly, this holds for all linear orders $\les$ on $V(G)$. Thus
 in some sense, the formula is order invariant, because ${\CQ}(G,\Cb)(v)$
 does not depend on the order. However, the set of $R\subseteq
 \{0,\ldots,n-1\}\times\{0,\ldots,r^{(n)}-1\}\times\{0,\ldots,cn-1\}$
 for which we have $\sem{\phi''}^{(G,\Cb,\les,R)}(v)={\CQ}(G,\Cb)(v)$ may
 depend on $\les$; \eqref{eq:89} just says that this set contains
 all $R$ except for an exponentially small fraction.

 In the final step, we apply a standard construction to turn
 randomness into non-uniformity, which is known as the ``Adleman
 trick''. It will be convenient to let
 \[
   \Omega\coloneqq 2^{
     \{0,\ldots,n-1\}\times\{0,\ldots,r^{(n)}-1\}\times\{0,\ldots,cn-1\}}.
 \]
 This is the sample space from which we draw the relations $R$
 uniformly at random.
 By \eqref{eq:89}, for each triple $(G,\Cb,\les)$ consisting
 of a graph $G$ of order $n$, a signal $\Cb\in\CS_p^\bool(G)$, and
 a linear order $\les$ on $V(G)$ there is a ``bad'' set ${B}(G,\Cb,\les)\subseteq\Omega$ 
 such that
 \begin{equation}
   \label{eq:90}
   \forall R\not\in{B}(G,\Cb,\les):\; \sem{\phi''}^{(G,\Cb,\les,R)}(v)={\CQ}(G,\Cb)(v)
 \end{equation}
 and
 \begin{equation}
   \label{eq:91}
   \frac{\big|{B}(G,\Cb,\les)\big|}{|\Omega|}\le 2^{-\pi(n)}.
 \end{equation}
 Observe that the number of triples $(G,\Cb,\les)$ is bounded
 from above by $2^{n^2+ pn+n\log n}$ and that $n^2+pn+n\log n<\pi(n)$.
 Thus we have
 \[
   \frac{\left|\bigcup_{(G,\Cb,\les)}{B}(G,\Cb,\les)\right|}{|\Omega|}\le
   \sum_{(G,\Cb,\les)}\frac{\big|{B}(G,\Cb,\les)\big|}{|\Omega|}\le
   2^{n^2+pn+n\log n}\cdot2^{-\pi(n)}<1.
   \]
   This means that there is a $R^{(n)}\in \Omega\setminus
   \bigcup_{(G,\Cb,\les)}{B}(G,\Cb,\les)$, and by \eqref{eq:90} we
   have
   \begin{equation}
     \label{eq:92}
     \sem{\phi''}^{(G,\Cb,\les,R^{(n)})}(v)={\CQ}(G,\Cb)(v) 
   \end{equation}
   for all graphs $G$ of order $n$, signals $\Cb\in\CS_p^\bool(G)$, and
 linear orders $\les$ on $V(G)$. Let
 $R^*\coloneqq\bigcup_{n\in\Nat}\{n\}\times R^{(n)}\subseteq\Nat^4$. We
 use $R^*$ as built-in numerical relation (in addition to the
 numerical relations already in $\phi''$). In $\phi''(x)$ we replace
 atomic subformulas $Y_r(y,y',y'')$ by $R^*(\ord,y,y',y'')$. Then it
 follows from \eqref{eq:92} that the resulting formula is
 order-invariant and defines ${\CQ}$.
\end{proof}

Before we prove the converse of the previous lemma, we need another
small technical lemma about FNNs.

\begin{lemma}\label{lem:random-function}
  Let $k,n\in\PNat$. There is a rational piecewise
  linear FNN $\FF$ of input and output
  dimension $1$, size polynomial in $n$ and $k$
  such that
  \[
    \Pr_{r\sim\CU}\big(\FF(r)\in\{0,\ldots,n-1\}\big)\ge 1-2^{-k},
  \]
  and for all $i\in\{0,\ldots,n-1\}$, 
  \[
    \frac{1-2^{-k}}{n}\le\Pr_{r\sim\CU}\big(\FF(r)=i\big)\le \frac{1}{n}.
  \]
  Furthermore, $\FF$ only uses $\relu$-activations.
\end{lemma}

\begin{proof}
   For all $a\in\Real$ and $\epsilon>0$, let $f_{\epsilon,a}:\Real\to\Real$ be defined by
   $f_{\epsilon,a}(x)\coloneqq
     \lsig\big(\frac{1}{\epsilon}x-\frac{a}{\epsilon}\big)$. Then
     \[
       \begin{cases}
         f_{\epsilon,a}(x)=0&\text{if }x\le a,\\
         0\le f_{\epsilon,a}(x)\le 1&\text{if }a\le
         x\le a+\epsilon,\\
         f_{\epsilon,a}(x)=1&\text{if }a+\epsilon\le x.
       \end{cases}
     \]
     Furthermore, for $a,b\in\Real$ with $a+2\epsilon\le b$ and
     $g_{\epsilon,a,b}\coloneqq
     f_{\epsilon,a}-f_{\epsilon,b-\epsilon}$ we have
     \[
       \begin{cases}
         g_{\epsilon,a,b}(x)=0&\text{if }x\le a,\\
         0\le g_{\epsilon,a,b}(x)=1&\text{if }a\le
         x\le{a+\epsilon},\\
         g_{\epsilon,a,b}(x)=1&\text{if }a+\epsilon\le x\le
         b-\epsilon,\\
         0\le g_{\epsilon,a,b}(x)=1&\text{if }b-\epsilon\le
         x\le{b},\\
         g_{\epsilon,a,b}(x)=0&\text{if }b\le x.
       \end{cases}
     \]
     Let $\ell\coloneqq\ceil{\log n}$ and $\epsilon\coloneqq
     2^{-k-\ell}$. 
     In the following, let 
     For $0\le i\le n$, let $a_i\coloneqq \frac{i}{n}$. Let
     $a_i^-\coloneqq \floor{2^{k+\ell+2}\frac{i}{n}}2^{-k-\ell-2}$ and
     $a_i^+\coloneqq
     \ceil{2^{k+\ell+2}\frac{i}{n}}2^{-k-\ell-2}$. Then $a_i^-\le
     a_i\le a_i^+$ and $a_i-a_i^-\le\frac{\epsilon}{4}$,
     $a_i^+-a_i\le\frac{\epsilon}{4}$. Moreover,
     $a_i^--a_{i-1}^+\ge\frac{1}{n}-{\epsilon}{2}\ge
     \frac{\epsilon}{2}$. Let $a_i^{++}\coloneqq
     a^+_{i}+\frac{\epsilon}{4}$ and $a_i^{--}\coloneqq
     a^-_{i}-\frac{\epsilon}{4}$. Then $a_i^{--}\le a_i^- \le a_i\le a_i^+\le a_i^{++}$ and
$a_i-a_i^{--}\le\frac{\epsilon}{2}$,
$a_i^{++}-a_i\le\frac{\epsilon}{2}$.

     For $1\le i\le n$, let $I_i\coloneqq [a_{i-1},a_i]$ and
     $J_i\coloneqq [a_{i-1}^{++},a_i^{--}]$. Then $J_i\subseteq
     I_i$. The length of $I_i$ is
     $\frac{1}{n}$, and the length of $J_i$ is at least
     $\frac{1}{n}-\epsilon\ge \frac{1}{n}(1-2^{-k})$. Thus the
     probability that a randomly chosen $r\in[0,1]$ ends up in one of
     the intervals $J_i$ is at least $1-2^{-k}$. Moreover, for every
     $i$ we have
     \[
       \frac{1-2^{-k}}{n}\le
       \frac{1}{n}-\epsilon\le\Pr_{r\sim\CU}(r\in J_i)\le\frac{1}{n}.
     \]
     Let $g_i\coloneqq g_{\frac{\epsilon}{4},a_{i-1}^+,a_i^-}$. Then
     $g(r)=1$ for $r\in J_i$, $g_i(r)=0$ for $r\not\in I_i$, and $0\le
     g_i(r)\le 1$ for $r\in I_i\setminus J_i$.

     We use the first two layers of our FNN $\FF$ to compute $g_i$ of the
     input for all $i$. That is, on the second level, $\FF$ has $n$ nodes
     $v_1,\ldots,v_n$, and $f_{\FF,v_i}(r)=g_i(r)$. As the intervals
     $I_i$ are disjoint, at most one
     $v_i$ computes a nonzero value, and with probability at least
     $2-2^{-k}$, at least one of the nodes takes value $1$. On the last
     level, for each $i$ there is an edge of weight $i-1$ from $v_i$
     to the output node. Then if $r\in J_i$ the output is $i-1$, and
     the assertion follows.
\end{proof}

For the following lemma, recall that $\GCgcnu$ denotes the extension
of $\GC$ with global counting and built-in numerical relations.
   
\begin{lemma}\label{lem:ri3}
  Let ${\CQ}$ be a unary query over $\CGS^\bool_p$ that is definable by an
  order-invariant $\GCgcnu$-formula. Then there is a polynomial-size
  bounded-depth family $\CR$ of rational piecewise-linear GNNs with ri that
  computes ${\CQ}$.

  Furthermore, the GNNs in $\CR$ only use $\relu$-activations and $\SUM$-aggregation.
\end{lemma}

\begin{proof}
  Let $\phi(x)$ be an order-invariant $\GCgcnu$-formula that defines
  ${\CQ}$. This
  means that for all $p$-labelled graphs $(G,\Cb)\in\CGS_p^\bool$,
  all linear orders $\les$ on $V(G)$, and all $v\in V(G)$, we have
  \[
    (G,\Cb,\les)\models\phi(v)\iff {\CQ}(G,\Cb)(v)=1.
  \]
  We want to exploit that if we choose the random features for each
  vertex, independently for all vertices, then with high probability, they
  are all distinct and thus they give us a linear order on the
  vertices.

  However, in order to be able to apply Lemma~\ref{lem:converse}, we
  need to carefully limit the randomness.  Let $U_1,U_2,U_3$ be function variables of type $\ttv\to\ttn$. 
  We let $\phi'(x)$ be the formula obtained from $\phi$ by
  replacing each atomic subformula $x\les x'$ by the formula
  \begin{gather*}
    U_1(x)<U_1(x')\vee \big(U_1(x)=U_1(x')\wedge U_2(x)< U_2(x')\big)\\
    \vee \big(U_1(x)=U_1(x')\wedge U_2(x)=U_2(x')\wedge U_3(x)\le
    U_3(x')\big).
  \end{gather*}
    That
  is, we order the vertices lexicographically by their
  $U_i$-values. If no two vertices have identical $U_i$-values for
  $i=1,2,3$, then this yields a linear order.

  Let $(G,\Cb)\in\CGS_p^\bool$ of order $n\coloneqq|G|$. For functions
  $F_1,F_2,F_3:V(G)\to \{0,\ldots,n-1\}$ and $v\in V(G)$, we write
  $(G,\Cb,F_1,F_2,F)\models\phi'(v)$ instead of
  $((G,\Cb),\Ca)\models\phi'$ for some and hence every assignment
  $\Ca$ with $\Ca(U_i)=F_i$ and $\Ca(x)=v$.
  Let us call $F_1,F_2,F_3: V(G)\to \{0,\ldots,n-1\}$ \emph{bad} if
  there are distinct  $v,w\in V(G)$ such that $F_i(v)=F_i(w)$ for $i=1,2,3$, and call them \emph{good} otherwise. Observe
  that for randomly chosen $F_1,F_2,F_3$, the probability that they
  are bad is at most $\frac{1}{n}$.

  By the construction of $\phi'$ from $\phi$, if $F_1,F_2,F_3$ are
  good then for all $v\in V(G)$ we have 
  \[
    (G,\Cb,F_1,F_2,F_3)\models\phi'(v)\iff {\CQ}(G,\Cb)(v)=1.
  \]
  By Lemma~\ref{lem:converse} in its version for GNNs with global
  readout and the logic $\FOCnu[2]$, there is a polynomial-size
  bounded-depth family $\CN=(\FN^{(n)})_{n\in\Nat}$ of rational
  piecewise-linear GNNs of input dimension $p+3$ such that for all
  $(G,\Cb)\in\CGS_p$ of order $n$ and all functions
  $F_1,F_2,F_3:V(G)\to\{0,\ldots,n-1\}$ the following holds. Let
  $\Cu\in\CS_3(G)$ be the signal defined by
  $\Cu(v)=\big(F_1(v),F_2(v),F_3(v)\big)$. Then for all $v\in V(G)$ we
  have $\tilde\CN(G,\Cb\Cu)\in\{0,1\}$ and
  \[
    \tilde\CN(G,\Cb\Cu)=1\iff(G,\Cb,F_1,F_2,F_3)\models\phi'(v).
  \]
  Thus if $F_1,F_2,F_3$ are good,
   \[
    \tilde\CN(G,\Cb\Cu)={\CQ}(G,\Cb).
  \]
  Thus all we need to do is use the random features to generate
  three random functions from $V(G)\to\{0,\ldots,n-1\}$. At first
  sight this seems easy, because the random features from $[0,1]$
  contain ``more randomness'' than the functions. However, in fact it
  is not possible, essentially because we cannot map the interval $[0,1]$ to a
  discrete subset of the reals of more than one element by a continuous function. But it is good
  enough to do this approximately, and for this we can use
  Lemma~\ref{lem:random-function}. We use this lemma to create a GNN
  layer $\FL^{(n)}$ that takes a random random signal $\Cr\in\CS_3(G)$ and
  computes a signal $\Cu\in\CS_3(G)$ such that with high probability,
  $\Cu(v)_i\in\{0,\ldots,n-1\}$, and $\Cu(v)_i$ is almost uniformly
  distributed in $\{0,\ldots,n-1\}$, for all $i,v$.
  This is good enough to guarantee that with high probability the
  functions $F_1,F_2,F_3$ defined by $\Cu$ are good. Thus if we
  combined $\FL^{(n)}$ with $\FN^{(n)}$ for all $n$, we obtain a
  family of GNNs with ri that computes ${\CQ}$.
\end{proof}

\begin{myremark}
  With a little additional technical work, we could also prove a
  version of the lemma where the GNNs only use
  $\lsig$-activations. But it is not clear that this is worth the
  effort, because actually $\relu$ activations are more important (in
  practice) anyway. Recall that $\lsig$ can be simulated with $\relu$,
  but not the other way around.
  \uend
\end{myremark}

Finally, we can prove the following theorem, which implies
Theorem~\ref{theo:main2} stated in the introduction.

\begin{theorem}\label{theo:ri}
  Let ${\CQ}$ be a unary query on $\CGS^\bool_p$. Then the following are equivalent:
  \begin{enumerate}
  \item ${\CQ}$ is definable in order-invariant $\GCgcnu$.
  \item There is a polynomial-size bounded-depth family
    of rational piecewise-linear GNNs with ri, using only
    $\relu$-activations and $\SUM$-aggregation, that computes ${\CQ}$.
  \item There is a rpl-approximable polynomial-weight bounded-depth
    family of GNNs with ri that computes ${\CQ}$.
  \item ${\CQ}$ is in $\TC^0$.
  \end{enumerate}
\end{theorem}

\begin{proof}
  The implication (1)$\implies$(2) is Lemma~\ref{lem:ri3}.  The
  implication (2)$\implies$(3) is trivial. The implication
  (3)$\implies$(1) is Lemma~\ref{lem:ri2}. And finally, the
  equivalence (1)$\iff$(4) is Corollary~\ref{cor:GCgc_TC0}.
\end{proof}

\section{Conclusions}
We characterise the expressiveness of graph neural networks in terms
of logic and Boolean circuit complexity. While this forces us to
develop substantial technical machinery with many tedious details,
the final results, as stated in the introduction, are surprisingly simple and clean:
GNNs correspond to the guarded fragment of first-order logic with
counting, and with random initialisation they exactly characterise
$\TC^0$. One reason I find this surprising is that GNNs carry out real
number computations, whereas the logics and circuits are discrete
Boolean models of computation.

We make some advances on the logical side that may be of independent
interest. This includes our treatment of rational arithmetic, 
non-uniformity and built-in relations on unordered structures, and 
unbounded aggregation (both sum and max). The latter may also shed new light on the relation between
first-order logic with counting and the recently introduced weight
aggregation logics \cite{BergeremS21} that deserves further study.

Our results are also interesting from a (theoretical) machine-learning-on-graphs
perspective. Most
importantly, we are the first to show limitations of GNNs with random
initialisation. Previously, it was only known that exponentially large
GNNs with ri can approximate all functions on graphs of order $n$
\cite{AbboudCGL21}, but no upper bounds were known for the more
realistic model with a polynomial-size restriction. Another
interesting consequence of our results (Theorems~\ref{theo:converse}
and \ref{theo:ri}) is that arbitrary
GNNs can be simulated by GNNs using only $\SUM$-aggregation and
$\relu$-activations with only a polynomial overhead in size. This was partially known \cite{XuHLJ19} for GNNs
distinguishing two graphs, but even for this weaker result the proof
of \cite{XuHLJ19} requires exponentially large GNNs because it is
based on the universal approximation theorem for feedforward neural
networks \cite{XuHLJ19}. It has recently been shown that such a
simulation of arbitrary GNNs by GNNs with $\SUM$-aggregation is not
possible in a uniform setting \cite{RosenbluthTG23}.

We leave it open to extend our results to higher-order GNNs
\cite{MorrisRFHLRG19} and the bounded-variable fragments of $\FOC$. It
might also be possible to extend the results to larger classes of
activation functions, for example, those approximable by
piecewise-polynomial functions (which are no longer Lipschitz
continuous). Another question that remains open is whether the
``Uniform Theorem'', Theorem~\ref{theo:uniform}, or
Corollary~\ref{cor:uniform} have a converse in terms of some form of
uniform families of GNNs. It is not even clear, however, what a
suitable uniformity notion for families of GNNs would be. We only
consider random initialisation in combination with global readout. The
exact expressiveness of polynomial-weight bounded-depth families of
GNNs with random initialisation, but without global readout, remains
unclear.

The most interesting extensions of our results would be to
GNNs of unbounded depth. Does the correspondence between circuits and
GNNs still hold if we drop or relax the depth restriction? And,
thinking about uniform models, is there a descriptive complexity
theoretic characterisation of recurrent GNNs? As a first step in this
direction, the logical expressiveness of recurrent GNNs has recently
been studied in \cite{PfluegerCK24}.
 
\printbibliography

@article{AamandCINRSSW22,
	author = {Aamand, Anders and Chen, Justin and Indyk, Piotr and Narayanan, Shyam and Rubinfeld, Ronitt and Schiefer, Nicholas and Silwal, Sandeep and Wagner, Tal},
	journal = {Advances in Neural Information Processing Systems (NeurIPS 2022)},
	pages = {27333--27346},
	title = {Exponentially improving the complexity of simulating the Weisfeiler-Lehman test with graph neural networks},
	url = {https://proceedings.neurips.cc/paper_files/paper/2022/hash/af0ad514b9cda46bd49e14ee11e2672f-Abstract-Conference.html},
	volume = {35},
	year = {2022}}

@inproceedings{GeertsSV22,
	author = {Geerts, Floris and Steegmans, Jasper and Van den Bussche, Jan},
	booktitle = {International Symposium on Foundations of Information and Knowledge Systems (FoIKS 2022)},
	doi = {10.1007/978-3-031-11321-5_2},
	organization = {Springer},
	pages = {20--34},
	title = {On the expressive power of message-passing neural networks as global feature map transformers},
	year = {2022}}

@article{MorrisLMRKGFB23,
	author = {Christopher Morris and Yaron Lipman and Haggai Maron and Bastian Rieck and Nils M. Kriege and Martin Grohe and Matthias Fey and Karsten M. Borgwardt},
	journal = {Journal of Machine Learning Research},
	number = {333},
	pages = {1-59},
	status = {JOU},
	title = {Weisfeiler and {Leman} go Machine Learning: The Story so far},
	url = {https://jmlr.org/papers/v24/22-0240.html},
	volume = {24},
	year = {2023}}

@inproceedings{GroheR24,
	author = {Martin Grohe and Eran Rosenbluth},
	booktitle = {Proceedings of the 39th Annual ACM/IEEE Symposium on Logic in Computer Science (LICS 2024)},
	doi = {10.1145/3661814.3662093},
	pages = {40:1--40:14},
	title = {Are Targeted Messages More Effective?},
	year = {2024}}

@inproceedings{PfluegerCK24,
	author = {Maximilian Pflueger and David Tena Cucala and Egor V. Kostylev},
	booktitle = {Proceedings of the 38th {AAAI} Conference on Artificial Intelligence (AAAI 2024)},
	doi = {10.1609/AAAI.V38I13.29377},
	editor = {Michael J. Wooldridge and Jennifer G. Dy and Sriraam Natarajan},
	pages = {14608--14616},
	publisher = {{AAAI} Press},
	title = {Recurrent Graph Neural Networks and Their Connections to Bisimulation and Logic},
	year = {2024}}

@inproceedings{RosenbluthTG23,
	author = {Eran Rosenbluth and Jan Toenshoff and Martin Grohe},
	booktitle = {Proceedings of the 32nd International Joint Conference on Artificial Intelligence (IJCAI 2023)},
	doi = {10.24963/ijcai.2023/464},
	pages = {4172-4179},
	title = {Some Might Say All You Need Is Sum},
	year = {2023}}

@inproceedings{BergeremS21,
	author = {Steffen van Bergerem and Nicole Schweikardt},
	bibsource = {dblp computer science bibliography, https://dblp.org},
	booktitle = {Proceedings of the 29th {EACSL} Annual Conference on Computer Science Logic ({CSL} 2021)},
	doi = {10.4230/LIPIcs.CSL.2021.10},
	editor = {Christel Baier and Jean Goubault{-}Larrecq},
	pages = {10:1--10:18},
	publisher = {Schloss Dagstuhl - Leibniz-Zentrum f{\"{u}}r Informatik},
	series = {LIPIcs},
	title = {Learning Concepts Described By Weight Aggregation Logic},
	url = {https://doi.org/10.4230/LIPIcs.CSL.2021.10},
	volume = {183},
	year = {2021}}

@inproceedings{GradelO93,
	author = {Erich Gr{\"a}del and Martin Otto},
	booktitle = {Proceedings of the 6th Workshop on Computer Science Logic (CSL 1992), Selected Papers},
	editor = {E. B{\"o}rger and G. J{\"a}ger and H. Kleine B{\"u}ning and S. Martini and M.M. Richter},
	pages = {231-247},
doi          = {10.1007/3-540-56992-8\_15},
publisher = {Springer Verlag},
	series = {Lecture Notes in Computer Science},
	title = {Inductive definability with counting on finite structures},
	volume = {702},
	year = {1993}}

@inproceedings{Immerman87,
	author = {Neil Immerman},
	booktitle = {Proceedings of the 2nd IEEE Symposium on Structure in Complexity Theory (SCT 1987)},
	pages = {194-202},
	title = {Expressibility as a complexity measure: results and directions},
	year = {1987}}

@book{KuperLP00,
	editor = {Gabriel M. Kuper and Leonid Libkin and Jan Paredaens},
	publisher = {Springer},
	title = {Constraint Databases},
	year = {2000}}

@article{SiegelmannS95,
	author = {Hava T. Siegelmann and Eduardo D. Sontag},
	bibsource = {dblp computer science bibliography, https://dblp.org},
	doi = {10.1006/jcss.1995.1013},
	journal = {J. Comput. Syst. Sci.},
	number = {1},
	pages = {132--150},
	title = {On the Computational Power of Neural Nets},
	url = {https://doi.org/10.1006/jcss.1995.1013},
	volume = {50},
	year = {1995}}

@inproceedings{MaassSS91,
	author = {Wolfgang Maass and Georg Schnitger and Eduardo D. Sontag},
	bibsource = {dblp computer science bibliography, https://dblp.org},
	booktitle = {32nd Annual Symposium on Foundations of Computer Science (FOCS 1991)},
	doi = {10.1109/SFCS.1991.185447},
	pages = {767--776},
	publisher = {{IEEE} Computer Society},
	title = {On the Computational Power of Sigmoid versus Boolean Threshold Circuits},
	url = {https://doi.org/10.1109/SFCS.1991.185447},
	year = {1991}}

@article{KarpinskiM97,
	author = {Marek Karpinski and Angus Macintyre},
	bibsource = {dblp computer science bibliography, https://dblp.org},
	doi = {10.1006/jcss.1997.1477},
	journal = {J. Comput. Syst. Sci.},
	number = {1},
	pages = {169--176},
	title = {Polynomial Bounds for {VC} Dimension of Sigmoidal and General Pfaffian Neural Networks},
	url = {https://doi.org/10.1006/jcss.1997.1477},
	volume = {54},
	year = {1997}}

@article{Maass97,
	author = {Wolfgang Maass},
	doi = {10.1137/S0097539793256041},
	journal = {{SIAM} J. Comput.},
	number = {3},
	pages = {708--732},
	title = {Bounds for the Computational Power and Learning Complexity of Analog Neural Nets},
	volume = {26},
	year = {1997}}

@inproceedings{SatoYK21,
	author = {Ryoma Sato and Makoto Yamada and Hisashi Kashima},
	booktitle = {Proceedings of the 2021 {SIAM} International Conference on Data Mining ({SDM} 2021)},
	doi = {10.1137/1.9781611976700.38},
	editor = {Carlotta Demeniconi and Ian Davidson},
	pages = {333--341},
	publisher = {{SIAM}},
	title = {Random Features Strengthen Graph Neural Networks},
	year = {2021}}

@inproceedings{AbboudCGL21,
	author = {Ralph Abboud and Ismail Ilkan Ceylan and Martin Grohe and Thomas Lukasiewicz},
	booktitle = {Proceedings of the 30th International Joint Conference on Artificial Intelligence (IJCAI 2021)},
	doi = {10.24963/ijcai.2021/291},
	editor = {Zhi-Hua Zhou},
	pages = {2112--2118},
	title = {The Surprising Power of Graph Neural Networks with Random Node Initialization},
	year = {2021}}

@article{Morgan65,
	author = {H.L. Morgan},
	journal = {Journal of Chemical Documentation},
	number = {2},
	pages = {107--113},
	publisher = {ACS Publications},
	title = {The generation of a unique machine description for chemical structures -- a technique developed at chemical abstracts service.},
	volume = {5},
	year = {1965}}

@article{WeisfeilerL68,
	author = {B.Y. Weisfeiler and A.A. Leman},
	journal = {NTI, Series 2},
	note = {English translation by G.~Ryabov},
	url = {https://www.iti.zcu.cz/wl2018/pdf/wl_paper_translation.pdf},
	title = {The reduction of a graph to canonical form and the algebra which appears therein},
	year = {1968}}

@inproceedings{XuHLJ19,
	author = {Keyulu Xu and Weihua Hu and Jure Leskovec and Stefanie Jegelka},
	bibsource = {dblp computer science bibliography, https://dblp.org},
	booktitle = {7th International Conference on Learning Representations ({ICLR} 2019)},
	publisher = {OpenReview.net},
	title = {How Powerful are Graph Neural Networks?},
	url = {https://openreview.net/forum?id=ryGs6iA5Km},
	year = {2019}}

@inproceedings{MorrisRFHLRG19,
	author = {Christopher Morris and
                  Martin Ritzert and
                  Matthias Fey and
                  William L. Hamilton and
                  Jan Eric Lenssen and
                  Gaurav Rattan and
                  Martin Grohe},
	booktitle = {Proceedings of the 33rd AAAI Conference on Artificial Intelligence (AAAI 2019)},
	doi = {10.1609/aaai.v33i01.33014602},
	publisher = {{AAAI} Press},
	title = {Weisfeiler and {L}eman Go Neural: Higher-order Graph Neural Networks},
	volume = {4602-4609},
	year = {2019}}

@article{ChamiAPRM22,
	author = {Ines Chami and Sami Abu{-}El{-}Haija and Bryan Perozzi and Christopher R{\'{e}} and Kevin Murphy},
	journal = {Journal of Machine Learning Research},
	number = {89},
	pages = {1-64},
	title = {Machine Learning on Graphs: {A} Model and Comprehensive Taxonomy},
	volume = {23},
	year = {2022},
url          = {https://jmlr.org/papers/v23/20-852.html}}

@article{ScarselliGTHM09,
	author = {Franco Scarselli and Marco Gori and Ah Chung Tsoi and Markus Hagenbuchner and Gabriele Monfardini},
	journal = {IEEE Transactions on Neural Networks},
	number = {1},
	pages = {61-80},
 doi          = {10.1109/TNN.2008.2005605},
 	title = {The graph neural network model},
	volume = {20},
	year = {2009}}

@inproceedings{KuskeS17,
	author = {Dietrich Kuske and Nicole Schweikardt},
	booktitle = {32nd Annual {ACM/IEEE} Symposium on Logic in Computer Science ({LICS} 2017)},
	doi = {10.1109/LICS.2017.8005133},
	pages = {1--12},
	publisher = {{IEEE} Computer Society},
	title = {First-order logic with counting},
	year = {2017}}

@inproceedings{Hesse01,
	author = {William Hesse},
	booktitle = {Proceedings of the 28th International Colloquium on Automata, Languages and Programming ({ICALP} 2001)},
	doi = {10.1007/3-540-48224-5_9},
	editor = {Fernando Orejas and Paul G. Spirakis and Jan van Leeuwen},
	pages = {104--114},
	publisher = {Springer},
	series = {Lecture Notes in Computer Science},
	title = {Division Is in Uniform {TC}\({}^{\mbox{0}}\)},
	volume = {2076},
	year = {2001}}

@book{HajekP93,
	author = {Petr H\'ajek and Pavel Pudlak},
	publisher = {Springer},
	series = {Perspectives in Mathematical Logic},
	title = {Metamathematics of First-Order Arithmetic},
	year = {1993}}

@article{HesseAB02,
	author = {William Hesse and Eric Allender and David A. Mix Barrington},
	doi = {10.1016/S0022-0000(02)00025-9},
	journal = {J. Comput. Syst. Sci.},
	number = {4},
	pages = {695--716},
	title = {Uniform constant-depth threshold circuits for division and iterated multiplication},
	volume = {65},
	year = {2002}}

@inproceedings{BarceloKM0RS20,
	author = {Pablo Barcel{\'{o}} and Egor V. Kostylev and Mika{\"{e}}l Monet and Jorge P{\'{e}}rez and Juan L. Reutter and Juan Pablo Silva},
	bibsource = {dblp computer science bibliography, https://dblp.org},
	booktitle = {8th International Conference on Learning Representations ({ICLR} 2020)},
	publisher = {OpenReview.net},
	title = {The Logical Expressiveness of Graph Neural Networks},
	url = {https://openreview.net/forum?id=r1lZ7AEKvB},
	year = {2020}}

@inproceedings{GilmerSRVD17,
	author = {Justin Gilmer and Samuel S. Schoenholz and Patrick F. Riley and Oriol Vinyals and George E. Dahl},
	bibsource = {dblp computer science bibliography, https://dblp.org},
	booktitle = {Proceedings of the 34th International Conference on Machine Learning ({ICML} 2017)},
	editor = {Doina Precup and Yee Whye Teh},
	pages = {1263--1272},
	publisher = {{PMLR}},
	series = {Proceedings of Machine Learning Research},
	title = {Neural Message Passing for Quantum Chemistry},
	url = {http://proceedings.mlr.press/v70/gilmer17a.html},
	volume = {70},
	year = {2017}}

@inproceedings{Grohe21,
	author = {Martin Grohe},
	bibsource = {dblp computer science bibliography, https://dblp.org},
	booktitle = {Proceedings of the 36th Annual {ACM/IEEE} Symposium on Logic in Computer Science ({LICS} 2021)},
	doi = {10.1109/LICS52264.2021.9470677},
	pages = {1--17},
	publisher = {{IEEE}},
	title = {The Logic of Graph Neural Networks},
	url = {https://doi.org/10.1109/LICS52264.2021.9470677},
	year = {2021}}

@phdthesis{Bennett62,
	author = {James H. Bennett},
	school = {Princeton University},
	title = {On Spectra},
	year = {1962}}

@book{Immerman99,
	author = {Neil Immerman},
	bibsource = {dblp computer science bibliography, https://dblp.org},
	doi = {10.1007/978-1-4612-0539-5},
	isbn = {978-1-4612-6809-3},
	publisher = {Springer},
	series = {Graduate texts in computer science},
	title = {Descriptive complexity},
	url = {https://doi.org/10.1007/978-1-4612-0539-5},
	year = {1999}}

@book{Vollmer99,
	author = {Heribert Vollmer},
	bibsource = {dblp computer science bibliography, https://dblp.org},
	doi = {10.1007/978-3-662-03927-4},
	isbn = {978-3-540-64310-4},
	publisher = {Springer},
	series = {Texts in Theoretical Computer Science. An {EATCS} Series},
	title = {Introduction to Circuit Complexity - {A} Uniform Approach},
	url = {https://doi.org/10.1007/978-3-662-03927-4},
	year = {1999}}

@article{ChandraSV84,
	author = {Ashok K. Chandra and Larry J. Stockmeyer and Uzi Vishkin},
	bibsource = {dblp computer science bibliography, https://dblp.org},
	doi = {10.1137/0213028},
	journal = {{SIAM} J. Comput.},
	number = {2},
	pages = {423--439},
	title = {Constant Depth Reducibility},
	url = {https://doi.org/10.1137/0213028},
	volume = {13},
	year = {1984}}

@article{BarringtonIS90,
	author = {David A. Mix Barrington and Neil Immerman and Howard Straubing},
	bibsource = {dblp computer science bibliography, https://dblp.org},
	doi = {10.1016/0022-0000(90)90022-D},
	journal = {J. Comput. Syst. Sci.},
	number = {3},
	pages = {274--306},
	title = {On Uniformity within NC{\({^1}\)}},
	url = {https://doi.org/10.1016/0022-0000(90)90022-D},
	volume = {41},
	year = {1990}}

\end{document}
